\documentclass[aps,prx,twocolumn,superscriptaddress,nofootinbib]{revtex4}
\usepackage{amsfonts}
\usepackage{amsmath}
\usepackage{amssymb}
\usepackage{amsthm}
\usepackage{graphicx}
\usepackage{bm}
\usepackage{color}
\usepackage{mathrsfs}
\usepackage[colorlinks,bookmarks=true,citecolor=blue,linkcolor=red,urlcolor=blue]{hyperref}
\usepackage{appendix}
\usepackage{float}
\usepackage[export]{adjustbox}
\newcommand{\RNum}[1]{\uppercase\expandafter{\romannumeral #1\relax}}

\newcommand{\beq}{\begin{eqnarray} }
\newcommand{\eeq}{\end{eqnarray} }
\newcommand{\Beq}{\begin{eqnarray*} }
\newcommand{\Eeq}{\end{eqnarray*} }

\newtheorem{thm}{Theorem}
\newtheorem{lem}{Lemma}
\newtheorem{cor}{Corollary}
\theoremstyle{rmk}
\newtheorem*{rmk}{Remark}

\begin{document}
\title{Quasi-Nambu-Goldstone modes in many-body scar models}
\author{Jie Ren}
\thanks{These authors contributed to this work equally.}
\affiliation{Beijing National Laboratory for Condensed Matter Physics and Institute of Physics, Chinese Academy of Sciences, Beijing 100190, China}
\affiliation{University of Chinese Academy of Sciences, Beijing 100049, China}
\author{Yu-Peng Wang}
\thanks{These authors contributed to this work equally.}
\affiliation{Beijing National Laboratory for Condensed Matter Physics and Institute of Physics, Chinese Academy of Sciences, Beijing 100190, China}
\affiliation{University of Chinese Academy of Sciences, Beijing 100049, China}
\author{Chen Fang}
\email{cfang@iphy.ac.cn}
\affiliation{Beijing National Laboratory for Condensed Matter Physics and Institute of Physics, Chinese Academy of Sciences, Beijing 100190, China}
\affiliation{Songshan Lake Materials Laboratory, Dongguan, Guangdong 523808, China}
\affiliation{Kavli Institute for Theoretical Sciences, Chinese Academy of Sciences, Beijing 100190, China}

\begin{abstract}
From the quasisymmetry-group perspective [Phys. Rev. Lett. 126, 120604 (2021)], we show the universal existence of collective, coherent modes of excitations in many-body scar models in the degenerate limit, where the energy spacing in the scar tower vanishes.
The number of these modes, as well as the quantum numbers carried by them, are given, not by the symmetry of the Hamiltonian, but by the quasisymmetry of the scar tower: hence the name quasi-Nambu-Goldstone modes.
Based on this, we draw a concrete analogy between the paradigm of spontaneous symmetry breaking and the many-body scar physics in the degenerate limit.
\end{abstract}
\maketitle

\section{Introduction}
The theory of spontaneous symmetry breaking is one of the underpinning pillars of condensed matter physics.
It states that the macroscopic ground state at sufficiently low temperature usually develops an ``order'', so that the ordered state, also called the symmetry-broken state, has less symmetry compared with that of the underlying Hamiltonian.
An example is the ferromagnetic state developed from a Heisenberg model \cite{heissenberg}: while the Hamiltonian has full spin SU(2) symmetry, the ferromagnetic state only retains spin U(1) symmetry, corresponding to the rotation along the ordered moments.
Generally speaking, the symmetry $G$ of a Hamiltonian may break down to a lower symmetry $H\subset{G}$ in the ground state of the same Hamiltonian, and the quotient space $Q\equiv{G/H}$ describes the orientation of the order parameter.
In the example, $G=\mathrm{SU(2)}$, $H=\mathrm{U(1)}$, and $Q=G/H=S^2$ is a sphere of unit vectors, where $\mathbf{v}\in{Q}$ denotes a possible direction of magnetization.
One key observation is that while the ground state lowers the symmetry, the low-lying modes of collective excitations \cite{NG-1,NG-2,NG-3} over the ground state ``recover'' some of the broken symmetry.
These modes correspond to global rotations of the order parameter from one vector $\mathbf{v}\in{Q}$ to another vector $\mathbf{v}'$.
In the absence of long-range interactions, even if these rotations have small momenta, that is, some slow spatial modulation, the modes remain coherent and continue to zero energy as the momentum goes to zero.
These are the Nambu-Goldstone modes (NGM), of which phonons and magnons are typical examples.

Recently, in a Rydberg-atom experiment \cite{lukin2017}, quasi-periodic dynamics in certain observables was observed in the time evolution of a product state having finite energy density. 
Independently, theorists discovered in the Affleck-Kennedy-Lieb-Tasaki (AKLT) model \cite{AKLT-1987} a tower of exact eigenstates having equal energy spacing \cite{AKLT-1}. 
The phenomenon observed in Ref.~\cite{lukin2017} was termed the quantum-many-body scar \cite{PXP-1} and reproduced in simplified spin models \cite{PXP-2,PXP-3,PXP-4,PXP-5,PXP-6,PXP-7,PXP-8,PXP-9,PXP-10,PXP-11} that possess nearly-equally-spaced eigenstates. 
The tower of states in the AKLT model was also related to QMBS \cite{AKLT-2, AKLT-4}, and similar towers were discovered in various spin and fermion models known as the exact scar models \cite{AKLT-3,AKLT-5,Hubbard-1,Hubbard-2,XY-1,XY-2,Onsager,domain-wall-0,domain-wall,rainbow-1,rainbow-2,Bull2019,Hudomal2019,isingladder,PRA.109.023310,review-1,review-2,review-3}.
In terms of the energy spectrum, exact scar dynamics corresponds to the existence of a tower of eigenstates, called the scar tower, that has equal energy spacing $hT^{-1}$ buried in the middle of a non-integrable spectrum.
In Refs.~\cite{qsymm-1}, it is shown that many scar towers can be unified under a group theoretic framework (see also Refs.~\cite{qsymm-2,Klebanov-1,Klebanov-2,Klebanov-3}), where the tower of states form a representation of a group $\tilde{G}$ that is larger than the symmetry of the Hamiltonian $G$.
Operations in $\tilde{G}$ in general do not preserve the Hamiltonian, but preserve the scar tower, and if $\tilde{G}$ is a Lie group, the equal energy spacing can be easily created by adding a constant external field coupled to any generator of $\tilde{G}$.
In Ref.~\cite{qsymm-1}, $\tilde{G}$ is called the quasisymmetry group, and a routine is proposed to generate model Hamiltonians for nearly all known scar towers~\cite{qsymm-1,qsymm-2}.

Inspired by the recent proposal of the asymptotic scar in the XY-model \cite{asymscar,moudgalya2023symmetries,kunimi2024proposal}, in this Letter, we establish the general existence of modes of collective excitations in exact bosonic scar models, and relate their defining properties, such as the number of modes and the quantum numbers they can transport, to the quasisymmetry group $\tilde{G}$.
This is illustrated using two examples, the ferromagnetic scar model and the AKLT scar model in the main text, then proved for any scar model having an arbitrary $\tilde{G}$ in Appendix \ref{apx:qngt} in the degenerate limit.
Here the degenerate limit is where the energy spacing vanishes, and the entire scar tower becomes a degenerate-energy subspace.
%This result furnishes an analogy between the collective modes of the scar models and the (type-b) Nambu-Goldstone modes \cite{type-b-1,type-b-2,type-b-3} from spontaneous symmetry breaking.
We call these collective modes the quasi-Nambu-Goldstone modes (qNGM), based on the following similarities to the type-B \cite{type-b-1,type-b-2,type-b-3,herzogarbeitman2022manybody} NGM appearing in ferromagnetism:
(i) the number of qNGM is the same as the rank of the quasisymmetry group $\tilde{G}$, and the number of type-B NGM is the rank of the symmetry group $G$; (ii) the lifetime $\tau\rightarrow\infty$ as the momentum of the mode $q\rightarrow0$; and (iii), each qNGM carries the charge of $\tilde{G}$, and each NGM carries that of $G$.
There are also key differences: (i) the dispersion of qNGM near the zero momentum is linear, while that of type-B NGM is quadratic; (ii) at $t\gg\tau$, the excitations decay, and the system returns to the ground state in the case of NGM, but to chaotic, thermal states in the case of qNGM.

\section{qNGM in the ferromagnetic scar model}
There are many scar spin models, with the spin-1 XY model serving as a prominent example. 
In these models, the scar space, up to certain onsite unitary transformations, consists of ferromagnetic (FM) states. 
These states are parameterized by spherical coordinates:
\begin{equation}
	|\text{FM}_{\theta,\phi}\rangle = \includegraphics[width=0.1\textwidth,valign=c]{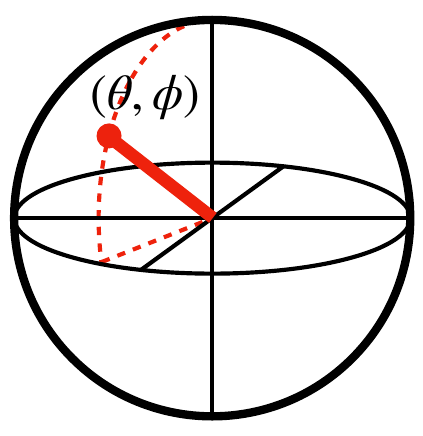} \otimes \cdots \otimes \includegraphics[width=0.1\textwidth,valign=c]{pics/spin-typical}.
\end{equation}
For a spin-$s$ model of $N$ sites, the scar space $\mathcal{H}_{\mathrm{SO(3)}} = \{|\text{FM}_{\theta,\phi}\rangle |\forall \theta,\phi\}$ form an $S=Ns$ irreducible representation of the quasisymmetry group $\tilde{G} = \mathrm{SO(3)}$, invariant under spin rotation.
To be concrete, we consider the following Hamiltonian
\begin{equation}\label{eq:jdm_ham}
\hat{H}=\sum_{j=1}^N\left[J\hat{\mathbf{s}}_j\cdot\hat{\mathbf{s}}_{j+1}-D_1(\hat{\mathbf{s}}_j\times\hat{\mathbf{s}}_{j+1})_z-D_2(\hat{\mathbf{s}}_j\times\hat{\mathbf{s}}_{j+2})_x\right],
\end{equation}
where $\hat{\mathbf{s}}_i$'s are {spin-$s$} operators, and the periodic boundary condition is assumed here and after.
Here the $D_{1,2}$ terms break the spin conservation of any component, and introduce ``frustration'' to the model.
Here ``frustration'' means that the scar states, being eigenstates of $\hat H$, are not eigenstates of certain individual terms in $\hat H$.
One may intuit that $|\text{FM}_{\theta,\phi}\rangle$'s are eigenstates: (i) each inner product takes the maximum value, and (ii) each cross product vanishes as all spins are parallel.
This argument is helpful and the conclusion correct, albeit that (ii) is incorrect.
(The readers are referred to Appendix \ref{apx:fsm} for the rigorous proof.)

\begin{figure}
\begin{centering}
\includegraphics[width=\linewidth]{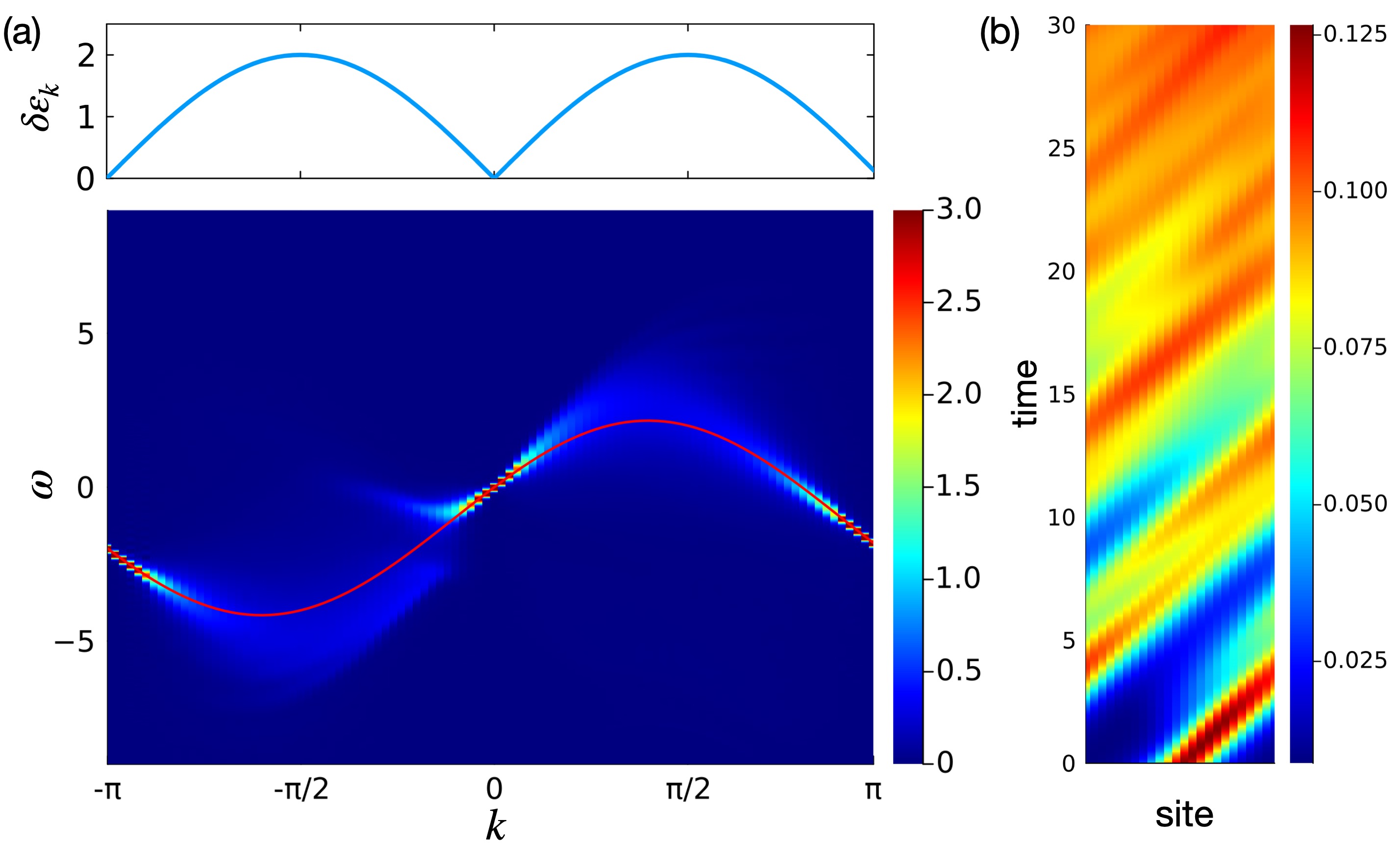}
\par\end{centering}
\caption{Spin-1/2 ferromagnetic scar model. (a) Bottom: Spectral function $A(k,\omega)$ of a spin-1/2 chain with parameters ($J=1, D_1=3, D_2=2$), computed utilizing the MPS-based Chebyshev series expansion method \cite{chemps}. The system size is $L=100$, and the Chebyshev vectors are approximated by a bond dimension of $\chi=30$ MPSs. We numerically evaluate the Chebyshev series up to order $200$. In presenting the result, we cut off values exceeding $3.0$. The red line denotes the ``band'' $\varepsilon_k$. Top: Energy variance $\delta\varepsilon_k$ of $|k\rangle$. Note that in this model there is a vanishing variance at $k=\pi$, which is accidental (see Appendix \ref{apx:fsm}). This zero will be absent by adding more terms. (b) Simulation of many-body evolution from a coherent wave-packet on spin-1/2 chain with size $L=24$. To show a long coherent time, we choose the parameters ($J=1, D_1=3, D_2=0.5$). The heatmap shows the local expectation of the magnon density $n_j=\langle \hat s_j^z\rangle + 1/2$. }
\label{fig:jdm}
\end{figure}

An analogy is observed between the scar model and the FM Heisenberg model $\hat{H}'=\sum_{j}J\hat{\mathbf{s}}_j\cdot\hat{\mathbf{s}}_{j+1}$ ($J<0$), as both share the FM manifold $\mathcal{H}_{\mathrm{SO(3)}}$ as their scar space or ground state space, respectively. 
The generators of the SO(3) symmetry play analogous roles in both scenarios: in the scar context, they act as ladder operators, generating the tower states; whereas, in the ground state symmetry-breaking case, they correspond to the NGM. 
Notably, the NGMs are not confined to the fixed momentum $k_0$ determined by the group generator (for the ferromagnetic scar, the homogeneity means $k_0=0$); 
small spatial modulation also yields small energies.

From SO(3) symmetry, we pick without loss of generality the all-down state $|\text{FM}_{\pi,0}\rangle \equiv \left|\Downarrow\right\rangle$ as the symmetry breaking ground state.
An NGM with momentum $k$ is generated by the magnon operator:	
\begin{equation}
	\hat{a}^\dag_k=\frac{1}{\sqrt{N}}\sum_j\hat{s}^+_j\exp(ikj).
\end{equation}
The single magnon states $|k\rangle\equiv\hat{a}^\dag_k \left|\Downarrow\right\rangle$ exhibit a dispersion $\varepsilon_k = sJk^2$, vanishing as $k\rightarrow 0$.
The properties of $|k\rangle$ become evident when examining the spectral function:
\begin{equation}\label{eq:spectral_function}
	A(k,\omega)=\mathrm{Im}\int \frac{dt}{i\pi} e^{i\omega{t}}
	\left\langle \Downarrow\right| \Theta(t)[\hat{a}_k(t),\hat{a}^\dag_k] \left|\Downarrow\right\rangle.
\end{equation}
The NGM in the Heisenberg model $\hat H'$ is evident as a sharp coherent signal detectable by neutron scattering experiments \cite{FM-1, FM-2}.

When $\hat H'$ undergoes a sudden global quench to the scar Hamiltonian Eq.(\ref{eq:jdm_ham}), the original ground state $\left|\Downarrow\right\rangle$ becomes a scar initial state showing revival dynamics in the presence of a transverse field \cite{qsymm-1,qsymm-2}.
This property is encoded in the same spectral function Eq.(\ref{eq:spectral_function}), despite the change in the Hamiltonian. 
The gaplessness of the NGM is governed by symmetry and locality, which are retained in the scar system, rendering $|k\rangle$ significant.
Indeed, Fig.\ref{fig:jdm}(a) illustrates a linear dispersion and vanishing energy variance as $k \rightarrow 0$.
The exact expressions for the energy expectation and variance can be obtained (See Appendix \ref{apx:fsm} for the proof):
\begin{eqnarray}\label{eq:disp}
	\varepsilon_k&\equiv&\langle{k}|\hat{H}-E_0|k\rangle= 2s D_1k+O(k^2),\\
	\nonumber
	\delta\varepsilon_k&\equiv&\sqrt{\langle{k}|(\hat{H}-E_0)^2|k\rangle-\varepsilon_k^2}=\sqrt{2s}|D_2k|+O(k^2).
\end{eqnarray}
The linear dispersion and vanishing energy variance facilitate the effective transport of quantum numbers. 
Consider a Gaussian wave packet constructed from $|k\rangle$'s, representing a quasiparticle with momentum $k=0$ and position $x=x_0$:
\begin{equation}
	|x_0\rangle=\int{dp}\exp\left(ipx_0-\frac{p^2}{2\delta{k}^2}\right)|p\rangle.
\end{equation}
From Eq.(\ref{eq:disp}), we find that this wave packet is a coherent quasiparticle with lifetime $\tau=|D_2|^{-1}\delta{k}^{-1}$, transporting an integer spin $\delta{S}_z=\hbar$ along the chain with velocity $v=2sD_1$. 
Fig.\ref{fig:jdm}(b) illustrates this coherent propagation using the model parameters provided in the caption. 
It is noteworthy that as $S_z$ is not conserved in $H$, the quantum number transported by the quasi-magnon is not related to the symmetry group $G$, but rather the quasi-symmetry group $\tilde{G}$.

We remark here that the scar initial state need not necessarily be an FM state; rather, it can possess arbitrary order, provided it is parameterized by a Lie group. An example of such a configuration frequently encountered is a ``$\pi$-ferromagnetic'' spin configuration:
\begin{equation}
	|\pi\text{FM}_{\theta,\phi}\rangle 
	= \includegraphics[width=0.1\textwidth,valign=c]{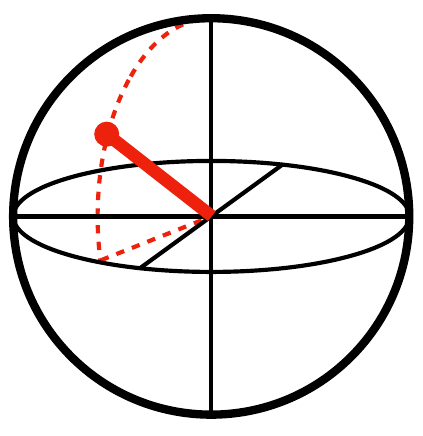}\otimes \includegraphics[width=0.1\textwidth,valign=c]{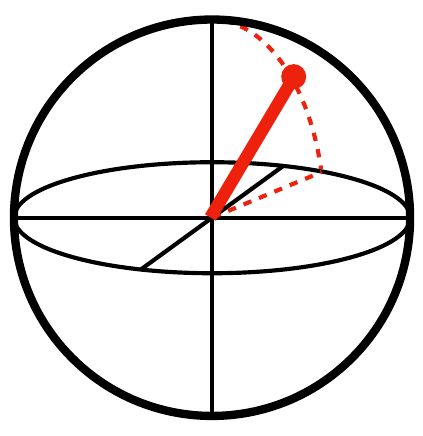}\otimes \cdots,
\end{equation}
where the spin configuration consists of even-odd spins that differ by a $\pi$-rotation around the $z$-axis.
The symmetric space spanned by $|\pi\text{FM}_{\theta,\phi}\rangle$ also possesses SO(3) symmetry, but the definition of spin rotation differs, with the generator being $\{\hat s^x,\hat s^y, \hat s^z\}$ on the odd sites while $\{-\hat s^x,-\hat s^y, \hat s^z\}$ on the even sites.

\section{qNGM in the AKLT scar model}

\begin{figure*}
\includegraphics[width=\linewidth]{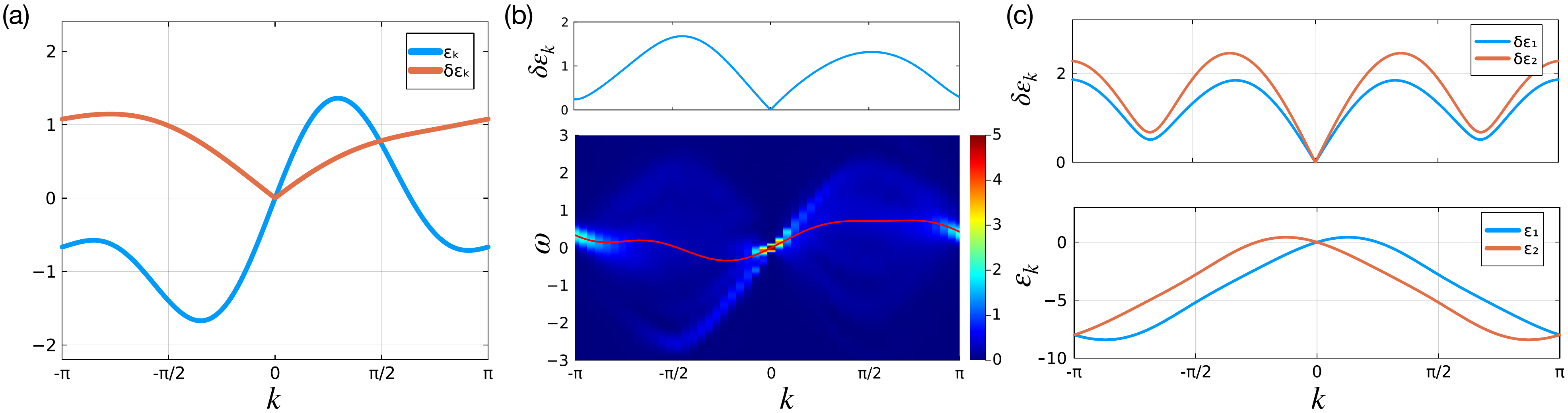}
\caption{(a)(b) Illustrations of two types of qNGMs for the generalized AKLT model defined in Eq.(\ref{eq:gen_aklt}). (c) Depiction of the SU(3) scar model, detailed in Appendix \ref{apx:qngmsu3}. (a) The energy expectation and variance for the z-qNGM state $|k,x\rangle$. (b) Bottom: Spectral function $\tilde{A}(k,\omega)$ of the x-qNGM for the generalized AKLT model. The computation employs a Chebyshev expansion for an $L=50$ system with a maximum bond dimension $\chi=30$, and the Chebyshev series is calculated up to order $n=100$. The red line represents the dispersion of the energy expectation $\varepsilon_k$. Top: Energy variance $\delta\varepsilon_k$ of the x-qNGM. Note that at $k\approx \pi$, there seems to be a coherent mode, but the energy variance is always positive. (c) Bottom: Energy expectation of two modes in the SU(3) scar model. Top: Energy variance of two modes.}
\label{fig:aklt}
\end{figure*}

In the context of FM scar model, the qNGM emerges through spatial modulation in the symmetry generator $\hat{a}^\dag$.
However, most exact scar spaces do \emph{not} have a clear symmetry structure, and the scar initial states are not product states as in the previous case, but matrix-product states \cite{mps-1,mps-2}, of which the AKLT model is a prime example.

Consider the Hamiltonian
\begin{equation}\label{eq:AKLT_scar}
	\hat{H}_\text{AKLT}=\sum_{j}\frac{J}{2}\left[\hat{\mathbf{l}}_j\cdot\hat{\mathbf{l}}_{j+1}+\frac{1}{3}(\hat{\mathbf{l}}_j\cdot\hat{\mathbf{l}}_{j+1})^2\right]-h\hat{l}^z_j,
\end{equation}
where $\hat{\mathbf{l}}_j$ is the spin-1 operator at site $j$.
In Ref.~\cite{AKLT-1}, the authors discovered a scar tower in Eq.(\ref{eq:AKLT_scar}), stemming from the ground state $|\text{AKLT}\rangle$, and generated by a ``one-way ladder'' $\hat Q^\dagger \equiv \sum_j(-1)^j (\hat l^+)^2$. 
This ladder creates a tower of scar states by $|n+1\rangle=\hat Q^\dagger |n\rangle$, but $\hat Q|n+1\rangle\ne|n\rangle$, precluding the formation of a symmetry sector. 
The energy spacing within the scar tower is $\Delta{E}=2(J-h)$. 
Setting $J=h$ renders the scar tower degenerate.

Although lacking an apparent symmetry structure, Ref.~\cite{dsymm} regards the degenerate scar space as a ``deformed'' SO(3)-symmetric space, constructed via a two-step process involving two Hilbert spaces: the ``prototype space'' and the physical space, denoted by the subscripts on kets.
In this specific example, the prototype space comprises a chain of $N$ $s=1/2$ spins, while the physical space comprises a chain of $l=1$ spins. 
The scar initial state takes the form: $|\Psi_{\theta,\phi}\rangle_\text{phy} = \hat{W}|\pi\text{FM}_{\theta,\phi}\rangle$, where the operator $\hat W$, termed the deforming operator, is a matrix product operator \cite{mpo-1,mpo-2} (MPO) mapping the prototype space to the physical one under periodic boundary conditions. (PBC).
The deformed state $|\Psi_{\theta,\phi}\rangle_\text{phy}$ represents a PBC MPS \cite{mps-1,mps-2} also parameterized by $\theta$ and $\phi$, graphically represented as:
\begin{equation}\label{eq:deformed_ssb}
	|\Psi_{\theta,\phi}\rangle_\text{phy} 
	=\includegraphics[width=0.5\linewidth,valign=c]{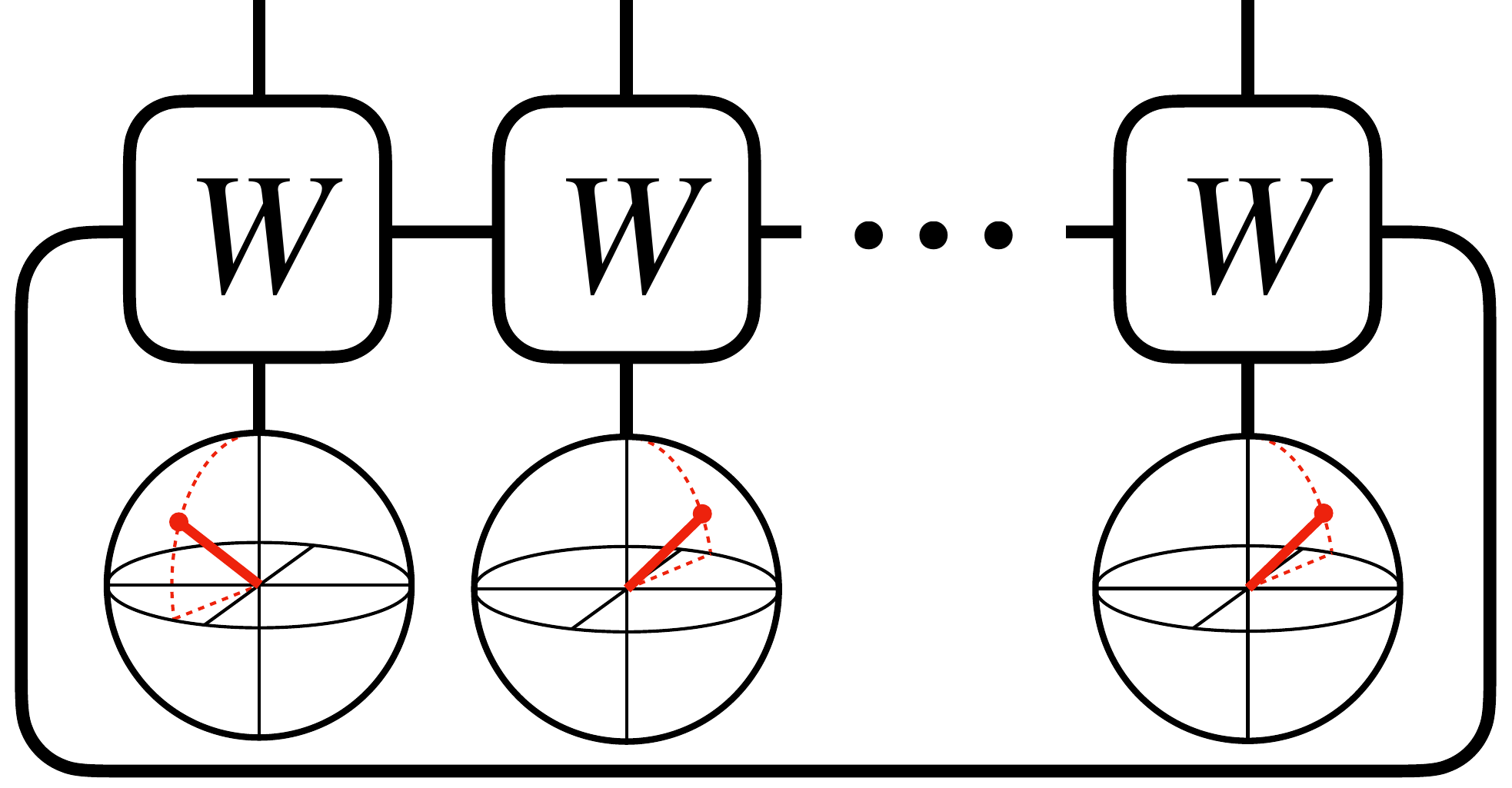},
\end{equation}
where the deforming operator is described by a single four-leg tensor $W^{ls;\alpha\beta}$, with $s=\uparrow,\downarrow$ contracting with the spin index in the prototype space, $l=1,0,-1$ representing the spin index in the physical space, and $\alpha,\beta$ denoting the auxiliary legs encoding entanglement.
For scar dynamics to occur, an additional term is necessary to generate equally spaced energy levels within the tower. This implies that the physical tower must retain at least one U(1) symmetry, represented by:
\begin{equation}
	\hat{W}\exp(-i\hat{S}^z\theta)|\psi\rangle_\text{pro}=\exp(-2i\hat{L}^z\theta)\hat{W}|\psi\rangle_\text{pro}.
\end{equation}
Ref.~\cite{dsymm} elucidates that for the AKLT scar space, $W^{ls;\alpha\beta}$ takes the form:
\begin{equation}
\begin{aligned}
	&W^{1\downarrow} = \sigma^+,\quad 
	W^{0\downarrow} = -\frac{\sigma^z}{\sqrt 2}, \quad 
	W^{-1\downarrow} =  -\sigma^-,\\
	&W^{1\uparrow} = -\sigma^-,\quad 
	W^{0\uparrow} = W^{-1\uparrow} = 0.
\end{aligned}
\end{equation}
Despite that $\hat W$ breaks the spin rotation symmetry in the physical space, the hidden symmetry generators in the prototype space enable the extension of the qNGM theorem to more general scenarios. 
For the general theorem and proof, please refer to the Appendix \ref{apx:dqngt}.

Due to the existence of U(1) quasisymmetry, the spin-$z$ rotation holds particular significance, leading to two distinct types of qNGMs, termed the $z$-qNGM and the $x$-qNGM. 
The $z$-qNGM originates from the physical state:
\begin{equation}
	|\Psi_{\pi,0}\rangle_\text{phy} 
	=\hat{W}\left[\includegraphics[width=0.07\textwidth,valign=c]{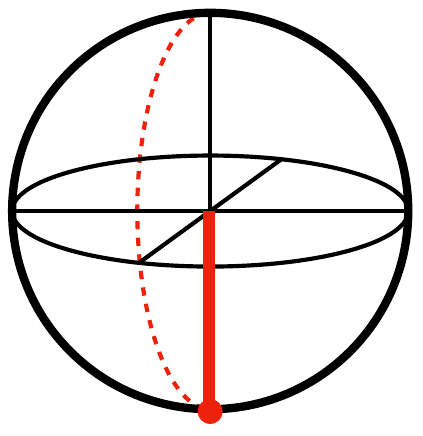}\otimes \includegraphics[width=0.07\textwidth,valign=c]{pics/spin-down}\otimes \cdots\right].
\end{equation}
Utilizing the hidden symmetry within the prototype space, we define the $z$-qNGM as:
\begin{equation}
	|k,z\rangle = \hat{W} \left[ \hat S_{\pi+k}^+\left|\downarrow\cdots\downarrow\right\rangle_\text{pro} \right].
\end{equation}
Notably, when $k=0$, this state precisely corresponds to the first excited state in the AKLT scar tower. 
Hence, the $z$-qNGM takes the form of an AKLT tower state generated by a spatially modulated ladder operator $\hat{Q}^+$. 
Fig.\ref{fig:aklt}(a) illustrates the energy expectation and variance for $|k,z\rangle$. 
It is worth mentioning that in the figure, the computation is actually performed utilizing the generalized AKLT model:
\begin{equation}\label{eq:gen_aklt}
	\hat H = \hat{H}_\text{AKLT} + \hat V,
\end{equation}
where $\hat V$ represents a frustrated term preserving the AKLT tower while introducing finite velocity to qNGMs. 
The precise form of $\hat V$ is given in Appendix \ref{apx:dqngmaklt}.

While $z$-qNGMs are simple to construct and characterize, they are difficult to obtain from a practical point of view.
In the prototype space, all qNGMs are created by acting $\hat{s}^+_k$ to $\left|\Downarrow\right\rangle_\text{pro}$; but in the physical space, the ladder operators do not in general exist. (They do luckily exist in the specific example of AKLT, but not in more general scar towers.)
Hence, we shift our focus to qNGMs originating from the scar initial state
\begin{equation}
	|\Psi_{\frac{\pi}{2},0}\rangle_\text{phy} 
	= \hat{W}\left[\includegraphics[width=0.07\textwidth,valign=c]{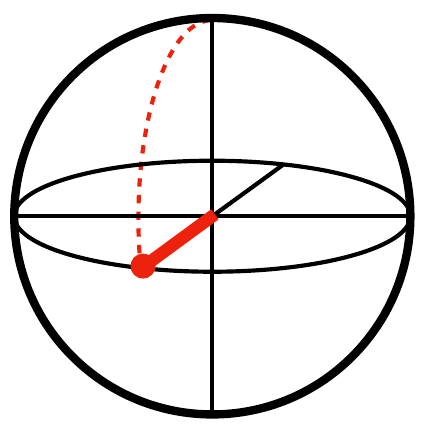}\otimes \includegraphics[width=0.07\textwidth,valign=c]{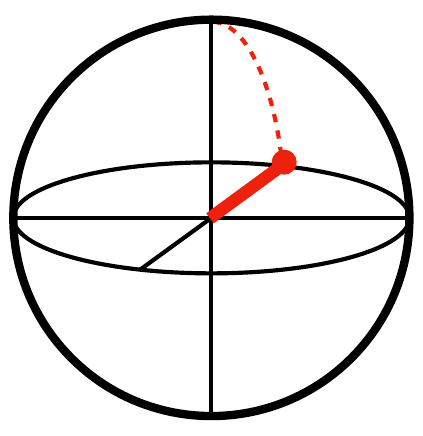}\otimes \cdots \right],
\end{equation}
from which the qNGM can be locally generated. 
Leveraging the U(1) quasisymmetry of the scar space, let $\hat{L}^z_{k}=\sum_j{e}^{ikj}\hat{l}^z_j$, and we define the $x$-qNGM as:
\begin{equation}
	|k, x\rangle \equiv \hat{L}^z_k|\Psi_{\frac{\pi}{2},0}\rangle_\text{phy}.
\end{equation}
Since it can be locally generated, $|k, x\rangle$ proves more relevant for potential experimental detection. 
In Fig.\ref{fig:aklt}(b), we present both the spectral function
\begin{equation}
\tilde{A}(k,\omega)= \operatorname{Im}\int \frac{dt}{i\pi} e^{i\omega t}\langle \Psi_{\frac{\pi}{2},0}| \Theta(t)[\hat{L}^z_{-k}(t),\hat{L}^z_k]|\Psi_{\frac{\pi}{2},0}\rangle,
\end{equation}
and $(\varepsilon_k,\delta\varepsilon_k)$, revealing a nearly coherent mode near $k=0$.

In addition to the AKLT scar mode, Appendix \ref{apx:dqngmother} investigates qNGMs of other scar models featuring deformed symmetric spaces. 
These include the Rydberg-blockaded scar model \cite{domain-wall-0,domain-wall}, the Onsager scar model \cite{Onsager}, and an additional scar tower in the spin-1 XY model \cite{XY-1,XY-2}. 
These findings underscore the universality of qNGM existence in scar models.

\section{General results and remarks}
Using two examples of the FM tower and the AKLT tower, we illustrate the ubiquitous existence of quasi-Nambu-Goldstone modes in quantum many-body scar models.
They are finite-momentum $k\neq 0$, low-entanglement states having linear dispersion and infinite lifetime as $k\rightarrow0$.
To be specific, a wave packet having width $\delta{k}$ in momentum has lifetime $\tau\propto\delta{k}^{-1}$.
(They become the tower eigenstates at $k=0$.)
In both examples, the scar tower has the SU(2) quasisymmetry or deformed symmetry, but more complicated scar towers may be similarly considered.
In Appendix \ref{apx:qngmsu3}, we show that two branches of qNGM exist in an SU(3)-scar model.
(See Fig.\ref{fig:aklt}(c) for energy expectation and variance of two qNGMs.)
Generally, the number of qNGM equals the rank of the underlying quasisymmetry (deformed symmetry), and the dynamics of those qNGM is described by a multi-band Hamiltonian.

Like Nambu-Goldstone modes, qNGM can be excited by local fields from a tower state, and they transport quantum numbers.
But the quantum numbers are those of the quasisymmetry, instead of the symmetry, of the Hamiltonian.
Unlike their low-energy counterpart, qNGM in energy are located within a continuum of \emph{thermal states}, and will relax into those states long after their lifetime (see Fig.\ref{fig:jdm}(b)), since qNGMs are superpositions of \underline{thermal} eigenstates (see Appendix \ref{apx:dqngmaklt} and \ref{apx:qngmpxp}). 

Moreover, the principle of qNGMs from quasisymmetry extends beyond exact scar models to encompass the PXP model, which exhibits a special (approximate) SU(2) quasisymmetry~\cite{PXP-1,PXP-2}. 
Within this framework, we can define generalized qNGMs. 
As detailed in Appendix \ref{apx:qngmpxp}, we delve into the properties of PXP qNGMs by examining both the original PXP model~\cite{PXP-1,PXP-2} and the deformed PXP model~\cite{PXP-4}.
Through numerical simulation, we confirm that these modes also exhibit approximate revival dynamics. 
Notably, for the approximate quantum many-body scars, the gradual thermalization of the small-$k$ qNGMs becomes indistinguishable from the inherent thermalization dynamics of the PXP model. 
Consequently, the qNGM emerges as an interesting initial state for scar dynamics: it dynamically resembles $\mathbb{Z}_2$ evolution, yet its spectrum primarily comprises thermal states rather than scar states.

\begin{appendix}

\section{Quasi-Nambu-Goldestone Theorem}
\label{apx:qngt}
In this section, we present a formal proof of the quasi-Nambu-Goldstone theorem, which predicts the existence of an asymptotic long-live mode associated with quasisymmetry.
To enhance clarity and conciseness, we first review the basis framework of the quasisymmetry framework \cite{qsymm-1} before enunciating the theorem.

\subsection{Quasisymmetry and quasi-Nambu-Goldstone modes}
In this work, we assume the symmetry to be non-Abelian Lie groups, thus possessing high-dimensional irreducible representation. 
The generators of a rank-$n$ Lie group, when expressed in the Cartan-Weyl basis, consist of $n$ pairs of ladder operators and $n$ mutually commuting operators forming the Cartan subalgebra.

The quasisymmetry framework involves the following elements:
\begin{itemize}
	\item a product \textit{anchor state}: $|\Psi_e\rangle = |\psi_e\rangle^{\otimes N}$;
	\item a \textit{quasisymmetry} group $\tilde G$ as a onsite symmetry action: $|\Psi_g\rangle = \hat U_g|\Psi_e\rangle = e^{i \theta_a \hat Q_a}|\Psi_e\rangle$;
	\item a quasisymmetric Hamiltonian $\hat H$ degenerate in the quasisymmetry subspace $\mathcal H_{\tilde G} \equiv \operatorname{span}\{\hat U_g|\Psi_e\rangle |\forall g\in \tilde G\}$;
	\item a spectrum splitting term $\hat H_z = \sum_{j=1}^N \hat q_j^z$ chosen from the generators of the Lie algebra. 
\end{itemize}
We also assume that the anchor state is in the highest-weight representation of $\tilde G$. 
In terms of local state $|\psi_e\rangle$, this means $|\psi_e\rangle$ is the ground state of a local group generator, denoted as $\hat q^z$. 
The product state $|\Psi_e\rangle$ is automatically the eigenstate of $\hat Q^z = \sum_j \hat q_j^z$, and the eigenvalue (corresponding to the weight in mathematics) is maximal.

For a Hamiltonian with quasisymmetry $\tilde G$, on top of the anchor state $|\Psi_e\rangle$, the state 
\begin{equation}
	\hat Q^+|\Psi_e\rangle = \sum_{j=1}^N \hat q_j^+ |\Psi_e\rangle
\end{equation}
is also the eigenstate in the scar space.
Note that we assume the quasisymmetry generator has zero momentum, although other momenta are generally possible, with the generator expressed as 
\begin{equation}
	\hat Q^+_k = \sum_{j=1}^N e^{ikj} \hat q_j^+
\end{equation}
in the general scenario. 
Nonetheless, we can always apply a unitary transformation 
\begin{equation}
	\hat U = \bigotimes_{j=1}^N \exp\left(-i k j \hat q_j^z \right),\quad
	\hat U \hat Q_k^+ \hat U^\dagger = \hat Q^+,
\end{equation}
where $[\hat q_j^z, \hat q_j^+] = \hat q_j^+$, to shift the momentum to zero without breaking the translational symmetry of the local Hamiltonian, and the following theorem remains valid. Hence, we will focus on the case where $k=0$.

The \textit{quasi-Goldstone mode} is defined as 
\begin{equation}
	|k\rangle \equiv \frac{1}{\mathcal N_k} \hat Q_k^+|\Psi_e\rangle = \frac{1}{\mathcal N_k}\sum_j e^{ikj} \hat q_j^+|\Psi_e\rangle,
\end{equation}
where $k$ is a nonzero crystal momentum.
In the forthcoming sections, we demonstrate that any scar model governed by the quasisymmetry group $\tilde G$ possesses a quasi-Goldstone mode, characterized by an energy variance that converges to zero as $k$ approaches zero.

\subsection{Quasi-Nambu-Golstone theorem}

\begin{thm}[Quasi-Nambu-Golstone theorem]
For a translational invariant local Hamiltonian $\hat H$ with quasisymmetry $\tilde G$ and ladder operator $\hat Q$.
The energy expectation 
\begin{equation*}
	\varepsilon_k \equiv E_k-E_0= \langle k|\hat H|k\rangle -\langle 0|\hat H|0\rangle 
\end{equation*}
and the energy variance
\begin{equation*}
	\delta \varepsilon^2_k \equiv \langle k|\left(\hat H-\langle k|\hat H|k\rangle\right)^2|k\rangle
\end{equation*}
is a continuous function of $k$, consequently converging to zero as $k$ approaches zero.
\end{thm}
\begin{rmk}
The original Goldstone theorem is formulated within the context of field theory, assuming the dispersion $\varepsilon_k$ as a continuous function of $k$, with the gaplessness being evident due to $\varepsilon_{k=0}=0$. 
In contrast, our approach in this study begins with a more concrete lattice model, abstaining from presuming the existence of a well-defined field theory. 
For a system comprising $N$ sites, the normalization $\mathcal{N}(N,k)$, energy $\varepsilon(N,k)$, and energy variance $\delta \varepsilon^2(N,k)$ are dependent on both $N$ and $k$. 
We shall establish that, as the system approaches the thermodynamic limit with $N\rightarrow \infty$, the resulting energy $\varepsilon_k$ and variance $\delta \varepsilon^2_k$ become well-defined and continuous functions of $k$.
\end{rmk}

\begin{lem}[Normalization]
The normalization factor $\mathcal N_k$ of the quasi-Goldstone mode $|k\rangle$ is independent of $k$, and is proportional to $\sqrt N$:
\begin{equation*}
	\mathcal N_k = \mathcal N =\sqrt N \Delta q,
\end{equation*}
where $\Delta q$ is determined by:
\begin{equation*}
	\Delta q^2 = \langle \psi_e|\hat q^- \hat q^+|\psi_e\rangle -|\langle\psi_e|\hat q^+|\psi_e\rangle|^2.
\end{equation*}
Here $|\psi_e\rangle$ is the local state that defines the anchor state, $\hat q^- = (\hat q^+)^\dagger$, and $\hat q^+$ is the local generator that defines the quasisymmetry action.
\end{lem}
\begin{proof}
The inner product $\langle k|k\rangle$ gives:
\begin{equation*}
	\mathcal N^2_k = \sum_j \langle\Psi_e| \hat q_j^- \hat q_j^+|\Psi_e\rangle + \sum_{j\ne l} e^{ik(j-l)}\langle\Psi_e| \hat q^-_j \hat q_l^+|\Psi_e\rangle.
\end{equation*}
Since $\Psi_e$ is a translational-invariant product state, the expectation $\langle \hat q_j^+ \hat q_l\rangle$ for $j\ne l$ is simply $\langle \hat q^+\rangle\langle \hat q\rangle=|\langle \hat q\rangle|^2$.
Also, for nonzero crystal momentum $k$,
\begin{equation}\label{eq:sum-to-zero}
\begin{aligned}
	\sum_{j,l=1}^N e^{ik(l-j)}\langle \hat q_j^- \hat q_l^+\rangle &= |\langle \hat q^+\rangle|^2\sum_{j,l=1}^N e^{ik(l-j)} = 0 \\
	\quad \Longrightarrow\quad
	\sum_{j\ne l} e^{ik(j-l)}\langle \hat q_j^- \hat q_l^+\rangle &= -\sum_j |\langle \hat q_j^+\rangle|^2.
\end{aligned}
\end{equation}
The normalization thus becomes a finite summation.
From the translational invariance, we get a compact expression:
\begin{equation*}
\begin{aligned}
	\mathcal N_k^2 &= \sum_j \left[\langle \hat q_j^- \hat q_j^+\rangle - |\langle \hat q_j^+\rangle|^2\right] \\
	&= N \left( \langle \psi_e|\hat q^- \hat q^+|\psi_e\rangle - \langle \psi_e|\hat q^-|\psi_e\rangle\langle \psi_e|\hat q^+|\psi_e\rangle \right),
\end{aligned}
\end{equation*}
and therefore proved the lemma.
\end{proof}

\begin{proof}[\textbf{Proof of the main theorem}]

We assume that the translational invariant Hamiltonian has the form 
\begin{equation}
	\hat H = \sum_{j=1}^N \hat h_j,
\end{equation}
where $\hat h_j$ acts nontrivially on the block of sites [$j,\cdots,j+m$].

\noindent\textbf{Energy dispersion:}
The energy expectation of $|k\rangle$ is
\begin{equation*}
\begin{aligned}
	\varepsilon_k &= \frac{1}{\mathcal N^2} \left( \langle \hat Q_k^- \hat H \hat Q_k^+ \rangle - \langle \hat Q^- \hat H \hat Q^+\rangle \right) \\
	&= \frac{1}{\mathcal N^2} \left(
		\langle \hat Q_k^- [\hat H,\hat Q_k^+]\rangle + \langle \hat Q_k^-  \hat Q_k^+ \hat H \rangle - \langle \hat Q^- \hat Q^+ \hat H\rangle \right) \\
	&= \frac{1}{\mathcal N^2} \langle \hat Q_k^- [\hat H,\hat Q_k^+]\rangle.
\end{aligned}
\end{equation*}	
We need to show next that $\langle \hat Q_k^- [\hat H,\hat Q_k^+]\rangle = N f(k)$ where $f(k)$ is a continuous function.
From the locality of the Hamiltonian, the commutator can be reduced to
\begin{equation*}
\begin{aligned}
	\left[\hat H, \hat Q_k^+\right] 
	&= \sum_l \sum_j e^{ikj} \left[\hat h_l, \hat q_j^+ \right] \\
	&= \sum_l e^{ikl}\left( \sum_{j=0}^{m} e^{ikj} \left[\hat h_l, \hat q_{l+j}^+ \right] \right).
\end{aligned}
\end{equation*}
We can further define an operator 
\begin{equation}
	\hat \Phi_l^+(k) = \sum_{j=0}^{m} e^{ikj} \left[\hat h_l, \hat q_{l+j}^+ \right]
\end{equation}
acting on block $[l,\dots,l+m]$.
$\hat \Phi_l^+(k)$ is translational invariant (its action does not depend on $l$), while it depends on momentum $k$.
Using this notation, the summation can then be simplified to 
\begin{equation*}
\begin{aligned}
	\varepsilon_k =&\ \frac{1}{\mathcal N^2} \sum_{j,l=1}^N e^{ik(l-j)} \langle \hat q_j^- \hat\Phi_l^+(k) \rangle \\
	=&\ \frac{1}{\mathcal N^2} \sum_{l=1}^N\left[ \sum_{j\in [l,l+m]} e^{ik(l-j)} \langle \hat q^-_j \hat\Phi_l^+(k) \rangle \right. \\
	&\left. + \sum_{j\notin [l,l+m]} e^{ik(l-j)} \langle \hat q^-_j \hat\Phi^+_l(k) \rangle\right].
\end{aligned}
\end{equation*}
Since the $|\Psi_e\rangle$ is a product state, the cluster decomposition always holds whenever $\hat q_j^-$ and $\hat\Phi_l^+(k)$ have no spatial overlap:
\begin{equation*}
\begin{aligned}
	\sum_{j\notin[l,l+m]} e^{ik(l-j)}\langle \hat q^-_j \hat\Phi_l^+(k) \rangle
	&= \sum_{j\notin[l,l+m]} e^{ik(l-j)}\langle \hat q^-\rangle \langle \hat\Phi_l^+(k) \rangle \\
	&= -\sum_{j=0}^me^{-ikj} \langle \hat q^- \rangle \langle \hat\Phi^+(k) \rangle,
\end{aligned}
\end{equation*}
where we have use the same summation trick in Eq.~(\ref{eq:sum-to-zero}).
Therefore, 
\begin{equation}\label{eq:energy-expect}
	\varepsilon_k = \frac{1}{\Delta q^2} \sum_{j=0}^m e^{-ikj} \left[\langle \hat q^-_{j} \hat\Phi_{0}^+(k) \rangle - \langle \hat q^-\rangle \langle \hat\Phi^+(k) \rangle\right].
\end{equation}
Since the expression is a finite sum of analytic functions, $\varepsilon_k$ is an analytical function of $k$.
By definition, $\varepsilon_{k=0} = 0$, therefore, $\varepsilon_k \rightarrow 0$ when $k\rightarrow 0$.
We have proved the first part of the theorem.

\noindent\textbf{Energy variance:}
The direct calculation of $\delta \varepsilon_k^2$ gives:
\begin{equation*}
\begin{aligned}
	\delta \varepsilon_k^2 
	=&\ \frac{1}{\mathcal N^2} \left\langle \hat Q_k^- (\hat H - E_k)^2 \hat Q_k^+ \right\rangle \\
	=&\ \frac{1}{\mathcal N^2} \left\langle \hat Q_k^- (\hat H - E_k) \left[\hat H- E_k, \hat Q_k^+\right]\right\rangle \\
	& +\frac{\varepsilon_k}{\mathcal N}\left\langle \hat Q_k^- (\hat H - E_k) \hat Q_k^+\right\rangle.
\end{aligned}
\end{equation*}
The second term on the right-hand side vanishes since $\frac{1}{\mathcal N}\langle \hat Q_k^- \hat H \hat Q_k^+\rangle = E_k$.
The expression then becomes
\begin{equation*}
\begin{aligned}
	\delta \varepsilon_k^2 
	=&\ \frac{1}{\mathcal N^2} \left\langle \hat Q_k^- (\hat H-E_k) \left[\hat H,\hat Q_k^+\right]\right\rangle \\
	=&\ \frac{1}{\mathcal N^2} \left\langle \left[\hat Q_k^-, \hat H-E_k \right] \left[\hat H, \hat Q_k^+\right] \right\rangle \\
	& +\frac{1}{\mathcal N^2} \left\langle (\hat H-E_k)\hat Q_k^- \left[\hat H, \hat Q_k^+\right] \right\rangle\\
	=&\ \frac{1}{\mathcal N^2} \left\langle \left[\hat H,\hat Q_k^+ \right]^\dagger \left[\hat H, \hat Q_k^+\right] \right\rangle - \varepsilon_k^2.
\end{aligned}
\end{equation*}
Using the notation of $\hat\Phi^+(k)$, the commutator is:
\begin{equation*}
	\left[\hat H, \hat Q_k^+\right] 
	= \sum_l e^{ikl} \hat\Phi_l^+(k).
\end{equation*}
Therefore the energy variance is
\begin{equation*}
\begin{aligned}
	\delta \varepsilon_k^2 + \varepsilon_k^2 =&\ \frac{1}{\mathcal N^2} \sum_{l=1}^N \left[\sum_{j=-m}^m e^{-ikj} \langle \hat \Phi^-_{l+j} \hat \Phi_l^+ \rangle \right. \\
	& \left. + \sum_{j\notin[-m,m]} e^{-ikj} \langle \hat\Phi^-_{l+j}\rangle \langle \hat\Phi_l^+ \rangle \right].
\end{aligned}
\end{equation*}
The disconnected sum becomes
\begin{equation*}
	\sum_{j\notin[-m,m]} e^{-ikj} \langle \hat\Phi^-_{l+j}\rangle \langle \hat\Phi_l^+ \rangle
	= -\sum_{j=-m}^m e^{-ikj} |\langle \hat\Phi^+ \rangle|^2.
\end{equation*}
Therefore,
\begin{equation}\label{eq:variance-expect}
\begin{aligned}
	\delta \varepsilon_k^2
	= \frac{1}{\Delta q^2} \sum_{j=-m}^m e^{-ikj} \left[\left\langle \hat\Phi^-_{i_0+j} \hat\Phi_{i_0}^+ \right\rangle-|\langle \hat\Phi \rangle|^2 \right] - \varepsilon_k^2.
\end{aligned}
\end{equation}
Being a finite sum of analytic functions, $\delta \varepsilon_k^2$ is also analytic.
By definition, $\delta \varepsilon_{k=0}^2=0$ and $\delta \varepsilon_{k=0}^2 \ge 0$ for all $k$, the asymptotic behavior can only be:
\begin{equation}
	\Delta H^2(k) \sim O(k^2).
\end{equation}
We, therefore, proved the quasi-Goldstone theorem.
\end{proof}	

\subsection{Remarks and corollaries}
The quasi-Goldstone theorem only requires the energy dispersion to be gapless, i.e., $\varepsilon_{k\rightarrow 0} \rightarrow 0$.
The following corollary further states that, for the quasi-goldstone mode to have nonzero velocity, the Hamiltonian should be frustrated, i.e., the scar state $|\Psi_g\rangle$ is not an eigenstate of local term $h_j$.
\begin{cor}[Frustration and velocity]
For frustration-free Hamiltonians, the energy expectations $\varepsilon_k \propto O(k^2)$.
\end{cor}
\begin{proof}
Without loss of generality, we assume $\langle 0|\hat H|0\rangle=0$. 
In the proof above, the expression for the energy expectation is
\begin{equation*}
	\varepsilon_k = \frac{1}{\mathcal N^2} \langle \hat Q_k^- H \hat Q^+_k\rangle 
	= \frac{1}{\Delta q^2} \sum_{j,l} e^{ik(l-j)} \langle \hat q_j \hat h_0 \hat q^+_l\rangle.
\end{equation*}
The derivative near $k=0$ is
\begin{equation*}
\begin{aligned}
	-i\Delta q^2 \left.\partial_k E_k \right|_{k=0} 
	&= \sum_{j,l} (l-j) \langle \hat q_j \hat h_0 \hat q_l^+\rangle \\
	&= \sum_{j,l} l \langle \hat q_j \hat h_0 q_l^+\rangle - \sum_{j,l} j \langle \hat q_j \hat h_0 \hat q_l^+\rangle.
\end{aligned}
\end{equation*}
The right-hand-side vanishes since $\hat h \sum_j \hat q^+_j|\Psi_e\rangle = 0$ for frustration-free Hamiltonian.
\end{proof}

Our discussion on the quasi-Goldstone mode now is focused on the highest weight state. 
The conclusion however can be generalized to tower states in the quasisymmetric scar models.
Notably, the original asymptotic scar \cite{asymscar} in spin-1 XY model is defined on top of a tower state.

\begin{cor}
For the scar tower state taking the form of the standard weight in the Lie algebra representation:
\begin{equation*}
	|\Psi_{\bm M}\rangle = \hat Q^+_{\bm\lambda}\cdots \hat Q^+_{\bm \beta}\hat Q^+_{\bm \alpha} |\Psi_e\rangle
\end{equation*}
The state 
\begin{equation*}
	|k,\bm{M}\rangle \equiv \sum_{j=1}^N e^{ikj}  \hat q^+_j |\Psi_{\bm M}\rangle
\end{equation*}
also has the asymptotic behaviors as $|k\rangle$.
\end{cor}

\begin{proof}
To prove that $|k,\bm{M}\rangle$ is an asymptotic scar, we only need to check that the untwisted state $|\Psi_{\bm M}\rangle$ is
\begin{itemize}
	\item invariant under permutations;
	\item has a converged local density matrix.
\end{itemize}
The permutational invariance implies that the expectation for two non-overlapping operators
\begin{equation}
	\langle \Psi_{\bm M}|\hat A_i \hat B_j|\Psi_{\bm M}\rangle
\end{equation}
is independent of their position.
Therefore, the summation trick Eq.~(\ref{eq:sum-to-zero}) still applies, and we still have Eqs.~(\ref{eq:energy-expect}) and (\ref{eq:variance-expect}).
The only difference is that the expectation $\langle\cdots\rangle$ is now taken from $|\Psi_{\bm M}\rangle$, of which local expectation is size-dependent.
However, it can be shown that those states have convergent local density matrices, and therefore the local expectation is well-defined in the thermodynamic limit.
\end{proof}

\section{Ferromagnetic Scar Model}
\label{apx:fsm}
In this section, we discuss properties of the spin-$s$ chain with Heisenberg and Dzyaloshinskii-Moriya (DM) interactions, described by the Hamiltonian:
\begin{equation}\label{eq:jdm}
\begin{aligned}
\hat{H} =&\ \sum_{j=1}^N[J\hat{\mathbf{s}}_j\cdot\hat{\mathbf{s}}_{j+1}-D_1(\hat{\mathbf{s}}_j\times\hat{\mathbf{s}}_{j+1})_z \\
	& - D_2(\hat{\mathbf{s}}_j\times\hat{\mathbf{s}}_{j+2})_x +h \hat s^z_j].
\end{aligned}
\end{equation}
Here, we add an external field to split lift the degeneracy of the scar tower.
When $s=1/2$, the Heisenberg term is integrable, while $D_1$ and $D_2$ terms break the integrability and the conservation of any spin component.

Note that when $h=0$, besides translational symmetry, the model also features a $\mathbb Z_2$ reflection+spin flip symmetry:
\begin{equation}
	\mathcal P: \hat s^{x/y}_j \rightarrow \hat s^{x/y}_{N-j},\quad 
	\hat s^{z}_j \rightarrow -\hat s^{z}_{N-j}.
\end{equation}
Then in the compatible sector ($k=0$ or $k=\pi$), the symmetry sector shall be further reduced to examine the level statistics.
This extra is relieved once we set $h \ne 0$ as it breaks such $\mathbb Z_2$ symmetry.

\begin{figure*}
	%\captionsetup{justification=raggedright,singlelinecheck=false}
    \begin{minipage}[t]{0.33\textwidth}
        \centering
        \includegraphics[width=\textwidth]{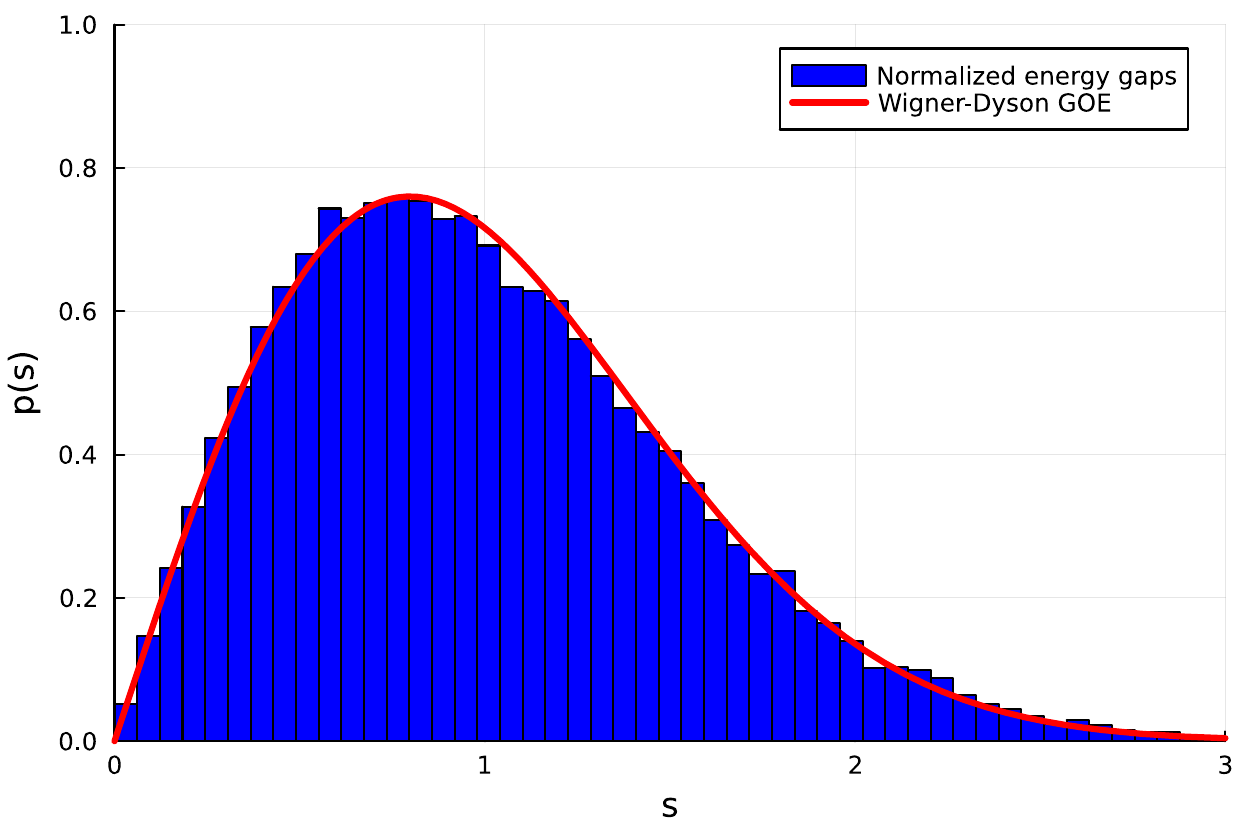}
        \begin{picture}(0,0)
            \put(-85,130){(a1)}
        \end{picture}
    \end{minipage}%
    \begin{minipage}[t]{0.33\textwidth}
        \centering
        \includegraphics[width=\textwidth]{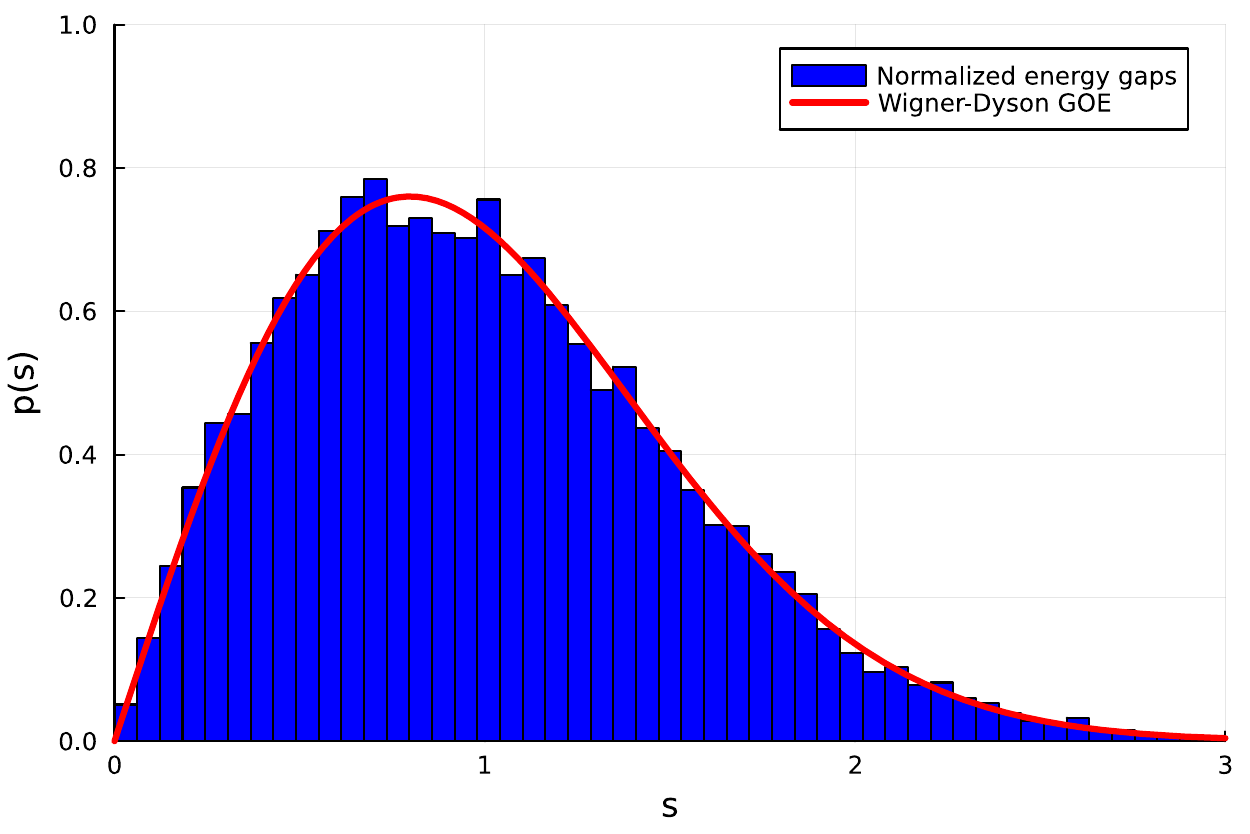}
        \begin{picture}(0,0)
            \put(-85,130){(b1)}
        \end{picture}
    \end{minipage}%
    \begin{minipage}[t]{0.33\textwidth}
        \centering
        \includegraphics[width=\textwidth]{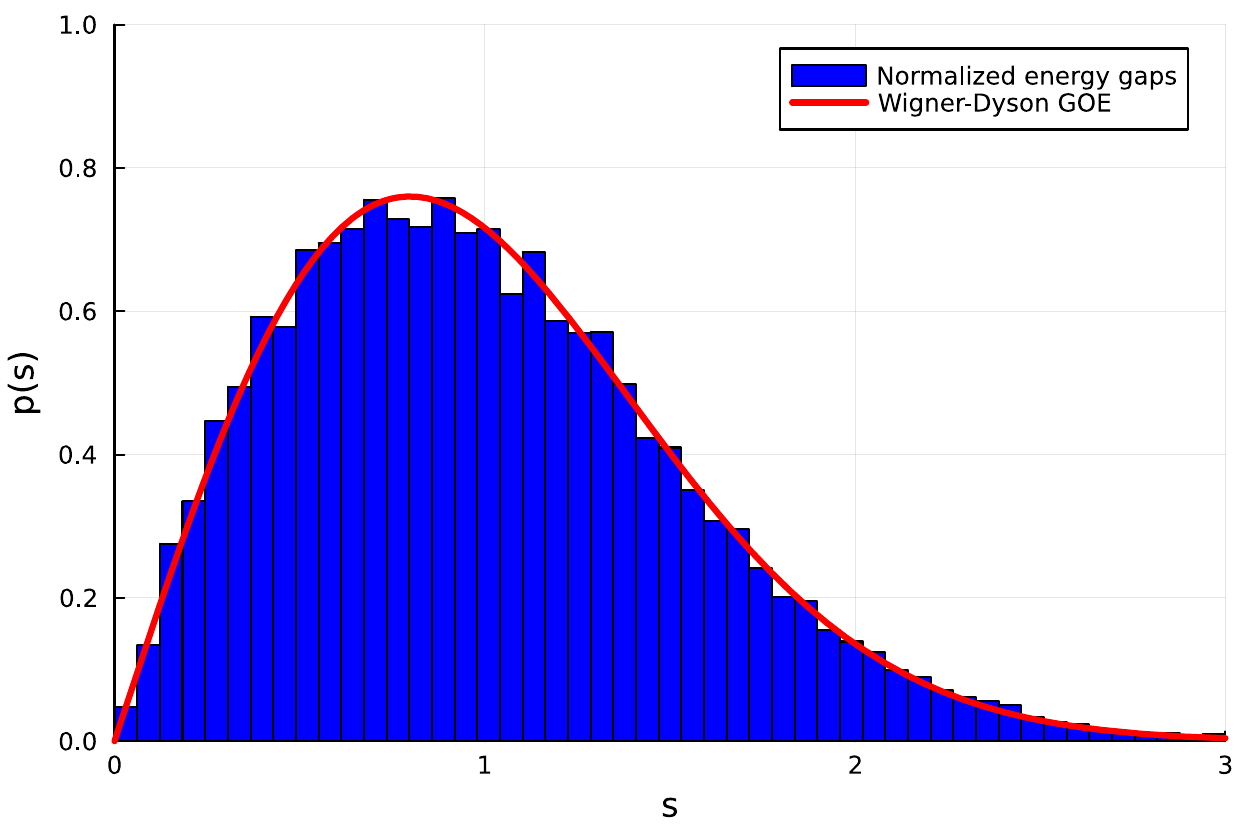}
        \begin{picture}(0,0)
            \put(-85,130){(c1)}
        \end{picture}
    \end{minipage}\\[1ex]
    \begin{minipage}[t]{0.33\textwidth}
        \centering
        \includegraphics[width=\textwidth]{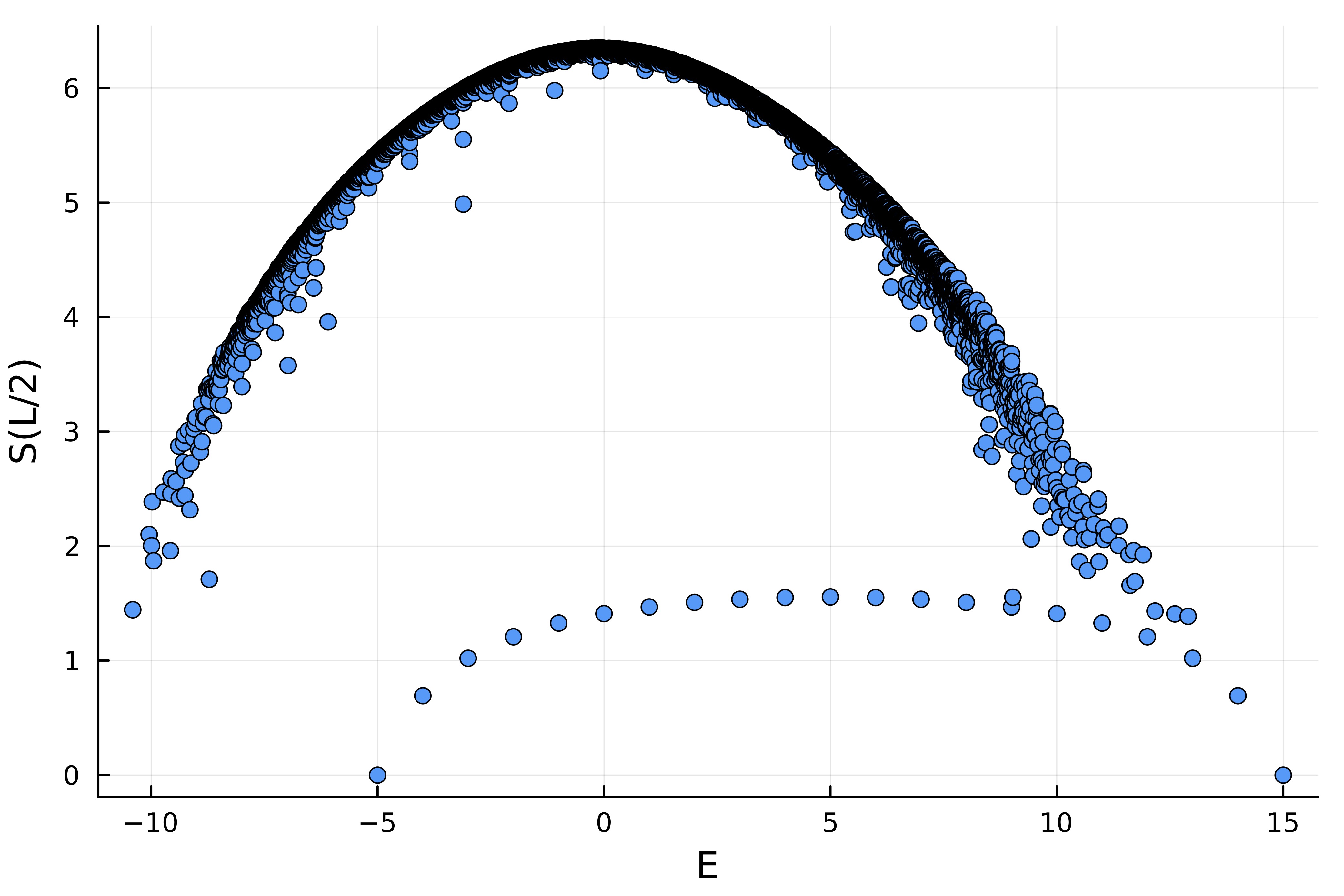}
        \begin{picture}(0,0)
            \put(-85,130){(a2)}
        \end{picture}
    \end{minipage}%
    \begin{minipage}[t]{0.33\textwidth}
        \centering
        \includegraphics[width=\textwidth]{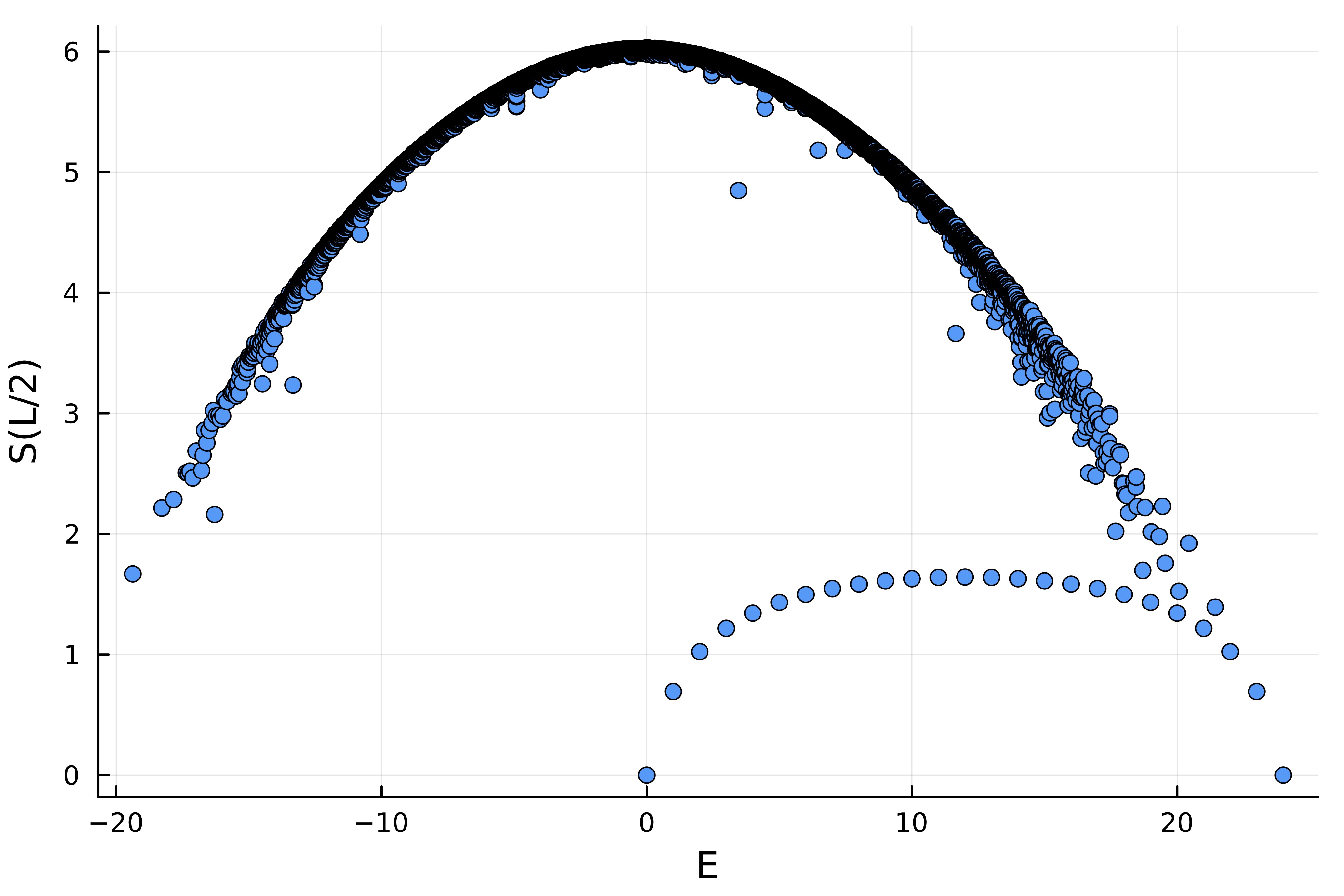}
        \begin{picture}(0,0)
            \put(-85,130){(b2)}
        \end{picture}
    \end{minipage}%
    \begin{minipage}[t]{0.33\textwidth}
        \centering
        \includegraphics[width=\textwidth]{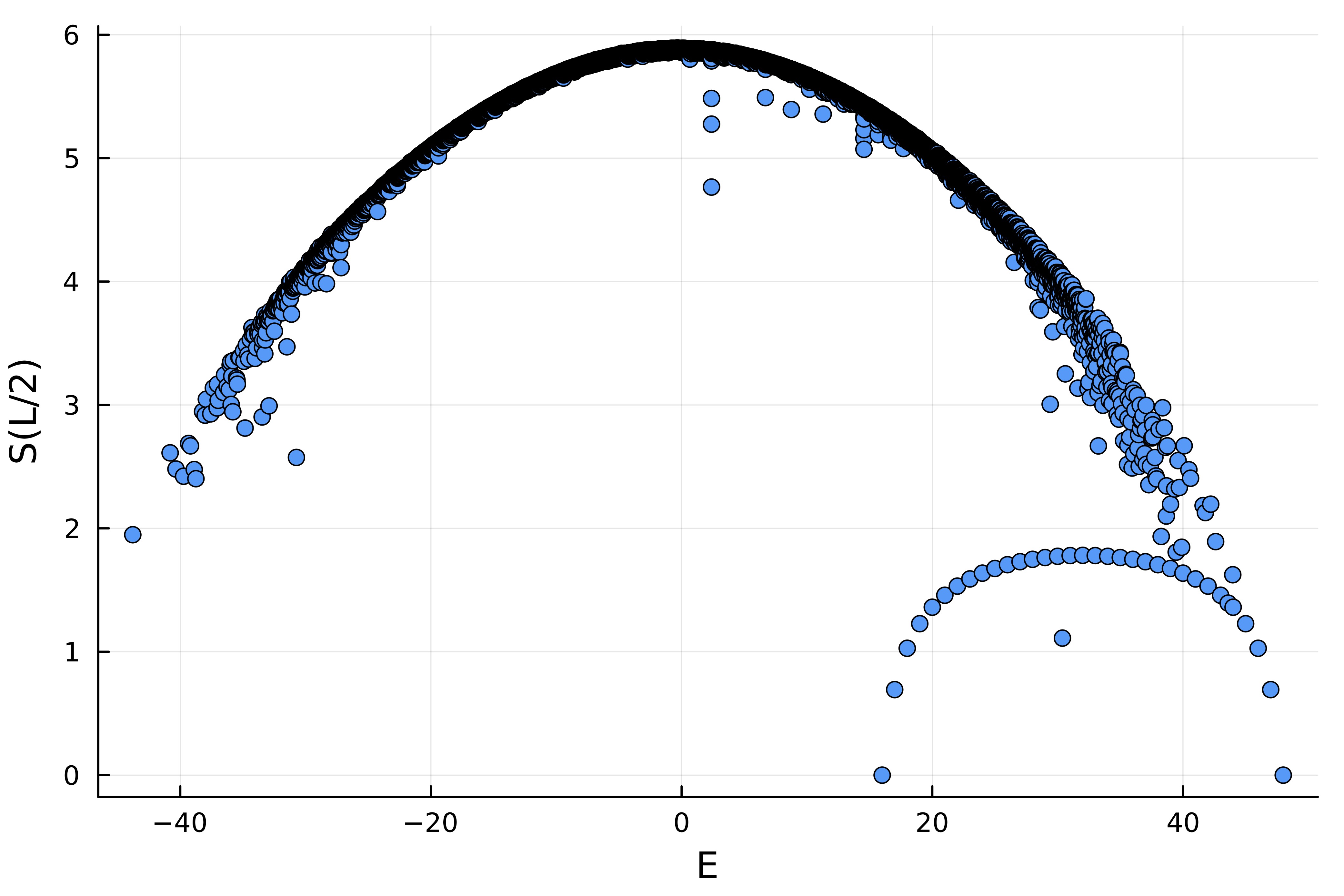}
        \begin{picture}(0,0)
            \put(-85,130){(c2)}
        \end{picture}
    \end{minipage}
    \caption{Level statistics and entanglement entropy in ferromagnetic scar models Eq.~(\ref{eq:jdm}), with $J=1$, $D_1=-1/2$, $D_2=1/3$, and $h=1$, focusing on the ($k=0$) sector. (a1)(a2) Spectrum of the $L=20$, $s=1/2$ spin chain, with Hilber space dimension $\dim \mathcal H = 52488$. (b1)(b2) Spectrum of the $L=12$, $s=1$ spin chain, with Hilber space dimension $\dim \mathcal H = 44368$. (c1)(c2) Spectrum of the $L=8$, $s=2$ spin chain, with Hilber space dimension $\dim \mathcal H = 48915$. (a1)(b1)(c1) Statistics of energy gaps in the middle half the ($k=0$) sector show Wigner-Dyson GOE form, indicating non-integrability. (a2)(b2)(c2) Entanglement entropy for eigenstates in the ($k=0$) sector reveals scar behavior across different spin values, with tower states displaying notably lower entanglement entropy.}
    \label{fig:JDM}
\end{figure*}
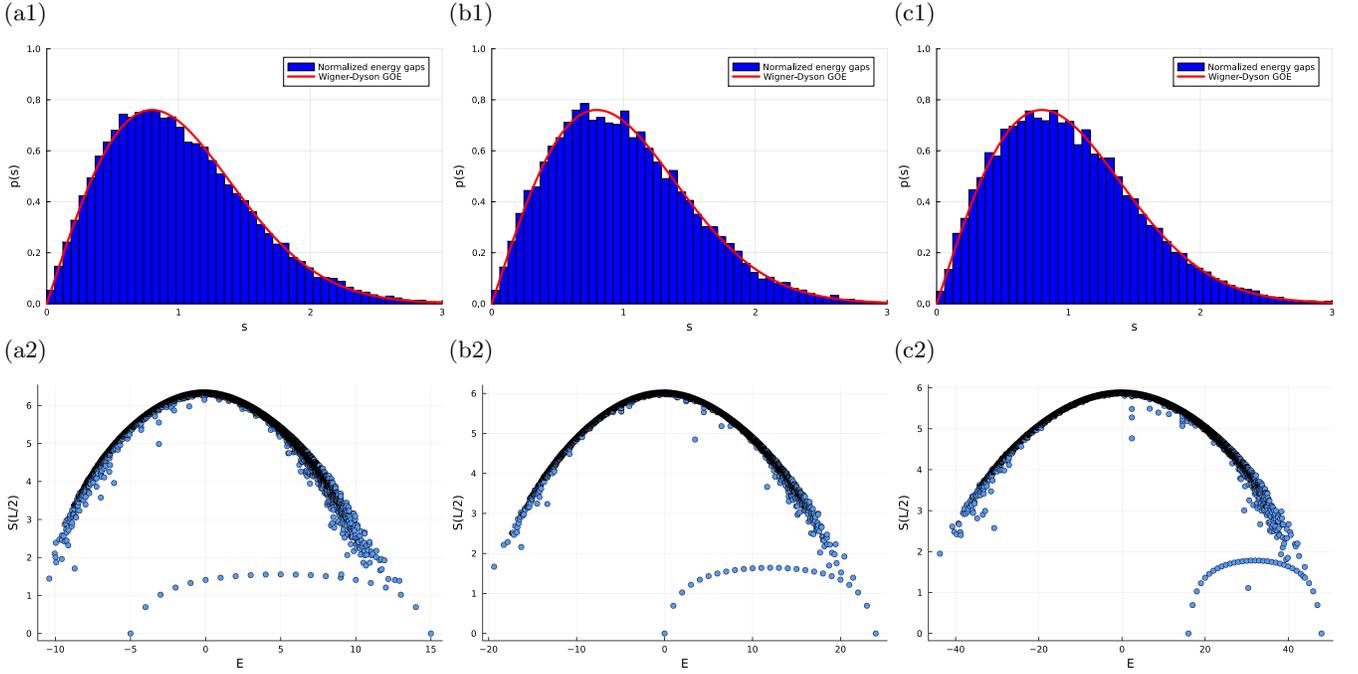

By performing the exact diagonalization on the $k=0$ sectors, we examine the spectrum of the ferromagnetic models and demonstrate numerically that for different choices of $s$, the model in Eq.~(\ref{eq:jdm}) displays a clear quantum many-body scar phenomenon (see Fig.~\ref{fig:JDM}).

\subsection{Exact scar tower on spin-$s$ chain}

\begin{thm}
The ferromagnetic subspace in the spin-$s$ chain 
\begin{equation*}
	\mathcal{H}_\text{FM} \equiv \operatorname{span} \left\{ \left. (\hat S^+)^n |-s\rangle^{\otimes N} \right| n=0,1,\dots,2sN \right\},
\end{equation*}
is the scar space of the Hamiltonian Eq.~(\ref{eq:jdm}).
\end{thm}

\begin{proof}
The Heisenberg interaction 
\begin{equation*}
	\hat J_d = \sum_{j=1}^N \hat{\mathbf s}_j \cdot \hat{\mathbf s}_{j+d}
\end{equation*}
manifestly has the SO(3) symmetry, and the ferromagnetic is its eigenspace.
We will then prove that the Dzyaloshinskii-Moriya  type of interaction
\begin{equation}
	\hat{\mathbf{D}}_d \equiv \sum_{j=1}^N \hat{\mathbf{s}}_j \times \hat{\mathbf{s}}_{j+d}
\end{equation}
will annihilate the spin-$s$ ferromagnetic space $\mathcal{H}_\text{FM}$. 
The first observation is that the vector $(\hat D_d^x,\hat D_d^y,\hat D_d^z)$ form a vector representation of SO(3). 
The ferromagnetic space is SO(3) invariant.
As the result, proofing $\hat D^z_d \mathcal{H}_\text{FM}=0$ is enough to show that $\hat{\mathbf{D}}_d \mathcal{H}_\text{FM}=0$.
The $z$-component $\hat D^z_d$ can be expressed as
\begin{equation}
	\hat D^z_d = \frac{i}{2} \sum_j (\hat s_j^+ \hat s_{j+d}^- - \hat s_j^- \hat s_{j+d}^+).
\end{equation}
It automatically annihilates $|0\rangle$ because of the $\hat s^-$ operators.
What we need to show is that 
\begin{equation}\label{apx:scar-tower-eq}
	\hat D^z_d |n\rangle = \hat D^z_d (\hat S^+)^n|0\rangle = 0.
\end{equation}
We will need the following commutation relation:
\begin{equation}
\begin{aligned}
	\left[\hat D^z_d, \hat S^+\right]
	&= \frac{i}{2} \sum_j \left[\hat s_j^+ \hat s_{j+d}^- - \hat s_j^- \hat s_{j+d}^+, \hat s_j^++\hat s_{j+d}^+\right] \\
	&= -i \sum_j \hat s_j^+ (\hat s^z_{j+d}-\hat s_{j-d}^z).
\end{aligned}
\end{equation}
Therefore,
\begin{equation}
	\left[[\hat D^z_d, \hat S^+],\hat S^+ \right]=0.
\end{equation}
Using the commutation relation, Eq.~(\ref{apx:scar-tower-eq}) becomes
\begin{equation}
\begin{aligned}
	\hat D^z_d |n\rangle
	&= \hat D^z_d \hat S^+|n-1\rangle \\
	&= (S^+)^{n-1}[\hat D^z_d, \hat S^+]|0\rangle + \hat S^+ \hat D^z_d |n-1\rangle.
\end{aligned}
\end{equation}
For the first part, since $|0\rangle$ is the eigenstate of all $\hat s^z$
\begin{equation}
	[\hat D^z_d, \hat S^+]|0\rangle
	= si \sum_j (\hat s_j^+ - \hat s_{j+d}^+)|0\rangle = 0.
\end{equation}
Therefore,
\begin{equation}
	\hat D^z_d |n\rangle = \hat S^+ \hat D^z_d |n-1\rangle
	= (\hat S^+)^n \hat D^z_d |0\rangle = 0.
\end{equation}
We therefore proved Eq.~(\ref{apx:scar-tower-eq}).
Thus $\mathcal{H}_\text{FM}$ is the eigenspace of Hamiltonian Eq.~(\ref{eq:jdm}).
\end{proof}

\subsection{Energy dispersion and variance}
In the following, we always consider the degenerate limit by setting $h=0$.
For the ferromagnetic model, the quasisymmetry is SO(3), generated by usual spin generators:
\begin{equation}
	\hat q^\pm = \hat s^\pm, \quad 
	\hat q^z = \hat s^z.
\end{equation}
The anchor state is the product of polarized states:
\begin{equation}
	|\Psi_e\rangle = \bigotimes_{j=1}^N |\psi_e\rangle,\quad 
	|\psi_e\rangle = |-s\rangle.
\end{equation}
The energy dispersion and variance can be directly calculated using Eq.~(\ref{eq:energy-expect}) and Eq.~(\ref{eq:variance-expect}). 
First, we calculate the normalization $\mathcal N^2 = N \Delta q^2$, where by definition:
\begin{widetext}
\begin{equation*}
	\Delta q^2 = \langle \psi_e|\hat q^- \hat q^+|\psi_e\rangle
	= \langle \psi_e|[\hat q^-, \hat q^+]|\psi_e\rangle
	= -2 \langle \psi_e|\hat q^z|\psi_e\rangle
	= 2s.
\end{equation*}
Then we calculate the quantity $\hat \Phi_l^+(k)$:
\begin{equation*}
\begin{aligned}
	\Phi_l^+(k) =&\ J\sum_{j=1,2} e^{ikj} \left[ \hat h_{\text{Heisenberg},l},\hat s^+_{l+j} \right]
	-D_1 \sum_{j=1,2} e^{ikj} \left[ \hat h_{\text{DM-1},l},\hat s^+_{l+j} \right]
	-D_2 \sum_{j=1,3} e^{ikj} \left[ \hat h_{\text{DM-2},l},\hat s^+_{l+j} \right] \\
	=&\ J \left[ 
		-\hat s_l^z \hat s_{l+1}^x - e^{ik}\hat s_l^x \hat s_{l+1}^z 
		-i\hat s_l^z \hat s_{l+1}^y -ie^{ik}\hat s_l^y \hat s_{l+1}^z
		+\hat s_l^+ \hat s_{l+1}^z +e^{ik}\hat s_l^z \hat s_{l+1}^+
		\right] \\
	& -D_1 \left[
			-\hat s_{l}^z\hat s_{l+1}^y -ie^{ik}\hat s_{l}^x \hat s_{l+1}^z 
			+i\hat s_{l}^z \hat s_{l+1}^x + \hat s_{l}^y \hat s_{l+1}^z
		\right] \\
	& -D_2 \left[ 
			-i \hat s_{l}^z \hat s_{l+2}^z + e^{2ik}\hat s_{l}^y \hat s_{l+2}^+ 
			- \hat s_{l}^+ \hat s_{l+2}^y +i e^{2ik}\hat s_{l}^z\hat s_{l+2}^z
		\right].
\end{aligned}
\end{equation*}
Note that in Eq.~(\ref{eq:energy-expect}) and Eq.~(\ref{eq:variance-expect}), $\hat\Phi^+(k)$ always acts directly on $|\Psi_e\rangle$. 
Therefore, we can make the following substitution without changing the result:
\begin{equation}
\begin{aligned}
	\hat s^x \simeq  \frac{1}{2}\hat s^+, \quad
	\hat s^y \simeq  -\frac{i}{2}\hat s^+, \quad
	\hat s^z \simeq  -s.
\end{aligned}
\end{equation}
The operator $\hat\Phi_l^+(k)$ can then be simplified to
\begin{equation*}
\begin{aligned}
	\Phi_l^+(k) 
	\simeq & -sJ  \left[ 
		(e^{ik}-1)\hat s_{l+1}^+ + (1- e^{ik})\hat s_l^+ \right] 
	 +isD_1 \left[
			\hat s_{l+1}^+ -e^{ik}\hat s_{l}^+ 
		\right] 
	 -i(e^{i2k}-1)D_2 \left[ 
			s^2 -\frac{1}{2}\hat s_{l}^+ \hat s_{l+2}^+ 
		\right].
\end{aligned}
\end{equation*}
Now we can compute the dispersion using Eq.~(\ref{eq:energy-expect}):
\begin{equation}
\begin{aligned}
	\varepsilon_k 
	&= \frac{1}{2s} \sum_{j} e^{-ikj} \langle  \hat s^-_{j} \hat \Phi_{0}^+(k) \rangle 
	= \frac{1}{2s} \left[ 
		-2s^2J (2-e^{-ik}- e^{ik})
	 +i2s^2D_1 \left(e^{-ik} -e^{ik}\right)\right] \\
	&= 2J s\left[ \cos k -1\right] 
	 +2 s D_1 \sin k.
\end{aligned}
\end{equation}

For the energy variance, we can in principle do the calculation using Eq.~(\ref{eq:variance-expect}).
However, the calculation for this specific case can be significantly simplified.
The simplification essentially comes from the fact that $|k\rangle$ is the eigenstate of the scar Hamiltonian when $D_2=0$, since the symmetry sector ($S^z_\text{tot} = -Ns+1, k=k$) is one-dimension.
We thus expect the variance to be solely controlled by $D_2$.
For concrete proof of the above statement, we first note that in Eq.~(\ref{eq:variance-expect}), the central quantity to calculate is the expectation $\langle \hat\Phi^-_{i_0+j} \hat\Phi_{i_0}^+ \rangle$.
We can insert a resolution:
\begin{equation*}
\begin{aligned}
	\left\langle \Psi_e \left|\hat\Phi^-_{i_0+j} \hat\Phi_{i_0}^+ \right|\Psi_e\right\rangle
	&= \frac{1}{\mathcal N^2}\sum_k \left\langle \Psi_e \left|\hat\Phi^-_{i_0+j} \hat Q_k^+ \right|\Psi_e\right\rangle 
	   \left\langle \Psi_e \left|\hat Q_k^- \hat\Phi_{i_0}^+ \right|\Psi_e\right\rangle
	   + \sum_{|\psi\rangle \notin \{|k\rangle\}}\left\langle \Psi_e \left|\hat\Phi^-_{i_0+j}  \right|\psi\right\rangle 
	   \left\langle \psi \left| \hat\Phi_{i_0}^+ \right|\Psi_e\right\rangle \\
	&= S_1 + S_2.
\end{aligned}
\end{equation*}
\end{widetext}
For the first term, it can be reduced to:
\begin{equation*}
\begin{aligned}
	S_1 &= \frac{1}{2s}\sum_l\sum_{j} e^{-ikj} \langle \hat\Phi_j^- \hat s_l^+\rangle
	\langle \hat s_l^-\hat\Phi_0^+\rangle \\
	&= \frac{1}{2s}  \left(\sum_{j} e^{ik(l-j)} \langle \hat\Phi_j^- \hat s_l^+\rangle\right)
	\left(\sum_l e^{-ikl}\langle \hat s_l^-\hat\Phi_0^+\rangle \right) \\
	&= 2s \varepsilon_k^2
\end{aligned}
\end{equation*}
Together with the normalization factor in the front, the first term and the $-\varepsilon_k^2$ term cancel out. 
For the matrix element $\left\langle \psi \left| \hat\Phi_{i_0}^+ \right|\Psi_e\right\rangle$, only $D_2$ terms contributes, so we can effectively substitute $\hat\Phi^+(k)$ to:
\begin{equation*}
	\Phi_l^+(k) 
	\simeq -i(e^{i2k}-1)D_2 \left[ s^2 -\frac{1}{2}\hat s_{l}^+ \hat s_{l+2}^+ \right].
\end{equation*}
The constant term will be cancelled out with $-|\langle \hat\Phi^+\rangle|^2$, the final term will be:
\begin{equation}\label{eq:jdm-ev}
\begin{aligned}
	\delta \varepsilon_k^2
	&= \frac{1}{2s}  \left[
		D_2^2 \sin^2 k \left\langle \Psi_e \right|\hat s_{l}^- \hat s_{l+2}^-  \hat s_{l}^+ \hat s_{l+2}^+\left|\Psi_e\right\rangle 
	\right] \\
	&= 2s D_2^2 \sin^2 k.
\end{aligned}
\end{equation}
In the $k \rightarrow 0$ limit:
\begin{equation}
\begin{aligned}
	\varepsilon_k &= 2s D_1 k + O(k^2), \\
	\delta\varepsilon_k &= 2sD_2 k + O(k^2).
\end{aligned}
\end{equation}
Indeed, the qNGM shows nonzero velocity near $k=0$, with vanishing energy variance.

\subsection{Accidental coherent mode at $k=\pi$}
From Eq.~(\ref{eq:jdm-ev}), we observe that when $k=\pi$, $\delta \epsilon_k = 0$. 
However, this coherent mode does not arise from quasisymmetry.
It is considered ``accidental'' because (i) more general perturbation can disrupt the coherence of $|k=\pi\rangle$, and (ii) this coherence is only present for the specific initial state $\left|\Downarrow\right\rangle$.

To illustrate, for the $s=1/2$ case, consider a perturbation with a longer range:
\begin{equation}
	\hat H' = \hat H - D_3 \sum_j (\hat{\bm s}_j \times \hat{\bm s}_{j+3})_y
\end{equation}
the energy variance is given by:
\begin{equation}
	\delta\varepsilon_k^2 = D_2^2 \sin^2(k) + D_3^2 \sin^2\left(\frac{3k}{2}\right).
\end{equation}
Despite the presence of quasisymmetry (and consequently qNGM), the accidental coherence at $k=\pi$ is lost.

Additionally, consider the qNGM generated from the initial state:
\begin{equation}
	|k\rangle = \sum_{j=1}^L e^{ikj} \hat s_j^+
	\bigotimes_{j=1}^L \left[\cos\frac{\theta}{2}\left|\downarrow\right\rangle + \sin\frac{\theta}{2} \left|\uparrow\right\rangle\right].
\end{equation}
Here, the energy variance is:
\begin{equation}
	\delta\varepsilon_k^2 = D_1^2 \sin^2\theta \sin^2\left(\frac{k}{2}\right)
	+ D_2^2 \cos^2\theta \sin^2(k).
\end{equation}
When $\theta \ne 0$, the coherence also disappears.

\subsection{Spectral function}
Here we outline the procedure for calculating the spectral function using the matrix-product-state-based algorithm \cite{chemps}. 
The spectral function, denoted as $A(k,\omega)$, characterizes the distribution of spectral weight as a function of frequency $\omega$ at a fixed momentum $k$:
\begin{equation}
\begin{aligned}
	A(k,\omega) 
	&= \operatorname{Im} \int \frac{dt}{i\pi} \ e^{i\omega t} \left\langle \Downarrow\right| \Theta(t) [\hat a_k(t), \hat a_k^\dagger] \left|\Downarrow\right\rangle \\
	&= \operatorname{Im} \sum_n \int_0^\infty \frac{dt}{i\pi} e^{i\omega t} 
	e^{-iE_nt}\left\langle \Downarrow\right| \hat a_k|n\rangle\langle n| \hat a_k^\dagger\left|\Downarrow\right\rangle \\
	&= -\frac{1}{\pi}\operatorname{Im} \sum_n \frac{\left\langle \Downarrow\right| \hat a_k|n\rangle\langle n| \hat a_k^\dagger\left|\Downarrow\right\rangle}{\omega-E_n+i 0^+} \\
	& = \left\langle \Downarrow\right| \hat a_k(t)\delta(\omega-\hat H) \hat a_k^\dagger\left|\Downarrow\right\rangle.
\end{aligned}
\end{equation}
Given a function $f(x)$ defined on the interval $[-1,1]$, we can represent it through a Chebyshev expansion:
\begin{equation}
	f(x) = \frac{1}{\pi \sqrt{1-x^2}}\left[ \mu_0+2 \sum_{n=1}^{N-1} \mu_n T_n(x)\right],
\end{equation}
where $T_n(x)$ are Chebyshev polynomials, defined recursively as
\begin{equation}
\begin{aligned}
	T_0(x)&=1, \quad T_1(x)=x, \\
	T_{n+1}(x)&=2 x T_n(x)-T_{n-1}(x),
\end{aligned}
\end{equation}
and $\mu_n$ is the Chebyshev momentum calculated from:
\begin{equation}
	\mu_n=\int_{-1}^1 d x f(x) T_n(x).
\end{equation}
Note that this expansion is only applicable within the domain $[-1,1]$. 
Should the spectral function span a different range, say $[-W,W]$, we adopt a rescaling strategy for the Hamiltonian:
\begin{equation}
\begin{aligned}
	\hat H &\rightarrow \hat H' = \hat H / W, \\
	A'(k,\omega') &= \left\langle \Downarrow\right| \hat a_k(t)\delta(\omega'-\hat H') \hat a_k^\dagger\left|\Downarrow\right\rangle.
\end{aligned}
\end{equation}
This rescaling allows us to redefine the spectral function within the standard $[-1,1]$ range. 
The transformation from the original spectral function to the rescaled one is:
\begin{equation}
	A(k,\omega) = \frac{1}{W} A\left(k,\frac{\omega}{W}\right).
\end{equation}
Subsequently, we focus our attention on the rescaled spectral function $A'$. 
Without loss of generality, we consider a renormalized Hamiltonian ensuring the spectral function remains defined within $[-1,1]$. 
Leveraging the Chebyshev expansion, we express the momentum 
\begin{equation}
	\mu_n = \left\langle\Downarrow\right|\hat a_k T(\hat H)\hat a_k^\dagger\left|\Downarrow\right\rangle.
\end{equation}
using a recursive definition of Chebyshev polynomials. 
This recursive formulation facilitates the computation process.

To proceed, we introduce the following Chebyshev vectors:
\begin{equation}
\begin{aligned}
	|t_0\rangle &= \hat a_k^\dagger\left|\Downarrow\right\rangle, \\
	|t_1\rangle &= \hat{H}|t_0\rangle, \\
	|t_{n+1}\rangle &= 2\hat{H}|t_{n}\rangle - |t_{n-1}\rangle. 
\end{aligned}
\end{equation}
The momentum is then 
\begin{equation}
	\mu_n = \langle t_0|t_n\rangle.
\end{equation}
However, a key challenge arises due to the many-body nature of the states represented by the vectors $|t_n\rangle$. 
To address this, we resort to an approximation scheme employing matrix-product states with a fixed bond dimension $\chi$. 
Remarkably, even with a modest bond dimension ($\chi=30$), this approximation yields satisfactory results, as demonstrated in Ref.~\cite{chemps}.

A few remarks are in order.
Firstly, the finite truncation of the Chebyshev series may introduce numerical artifacts known as \textit{Gibbs oscillations}. 
To mitigate this issue, a ``smoothing function'' can be applied to alleviate the artificial oscillations. 
The final result of the Chebyshev expansion is thereby defined as
\begin{equation}
	A(k,\omega) = \frac{1}{\pi \sqrt{1-\omega^2}}\left[ g_0\mu_0+2 \sum_{n=1}^{N-1} g_n\mu_n T_n(\omega)\right],
\end{equation}
where $g_n$'s are the \textit{Jackson damping} coefficients (see Ref.~\cite{polynomial-method} for details).
Additionally, truncating the bond dimension inevitably introduces numerical errors. 
If these errors result in function values outside the $[-1,1]$ interval, they can propagate exponentially in subsequent recursive computations. 
To avoid such numerical instabilities, an energy truncation procedure is incorporated into the calculation process, as outlined in \cite{chemps}.

\section{Deformed Quasi-Nambu-Goldstone Theorem}
\label{apx:dqngt}
\subsection{Deformed symmetry and quasi-Nambu-Goldstone modes}

The deformed symmetric space formalism extends the concept of the quasisymmetry scar space. Building upon the framework of the quasisymmetry space, we introduce a deforming matrix-product operator \cite{mpo-1,mpo-2} (MPO) denoted as $\hat W$, depicted as follows:
\begin{equation}
	\hat W = \includegraphics[valign=c,height=1.0cm]{pics/dMPS}.
\end{equation}
The resulting deformed anchor state takes the form of a matrix-product state \cite{mps-1,mps-2} (MPS):
\begin{equation}
	|\tilde\Psi_e\rangle = \hat W|\Psi_e\rangle = |[A \cdots A]\rangle,\quad A= \hat W|\psi_e\rangle.
\end{equation}
The scar space is then constructed as the span of:
\begin{equation}
	\mathcal H_d = \operatorname{span}\left\{\hat W \hat U_g|\Psi_e\rangle|\forall g \in \tilde{G}\right\}.
\end{equation}
Defining the quasi-Goldstone mode for the deformed scar as:
\begin{equation}
\begin{aligned}
	|\tilde k\rangle &\equiv \frac{1}{\mathcal N}\sum_j e^{ikj} \hat W q_j|\Psi_e\rangle \\
	&= \frac{1}{\mathcal N}\sum_j e^{ikj} \left[\includegraphics[valign=c,height=0.45cm]{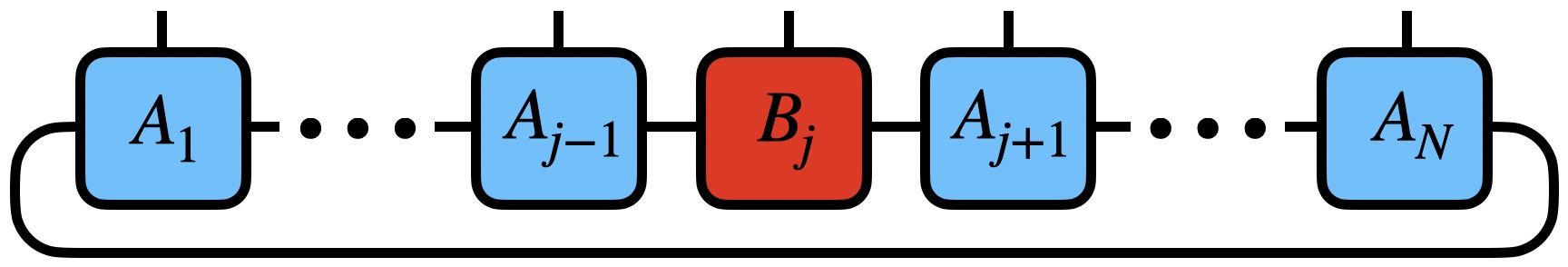}\right].
\end{aligned}
\end{equation}
In the deformed scenario, we assume $\hat W$ behaves as an operator that maps a product state to one with short-range entanglement. 
This implies injectivity of the MPS, ensuring the transfer matrix possesses a non-degenerate dominant eigenvalue. 
We consider the thermodynamic limit, and the tensor $A$ is assumed to be properly normalized so that:\footnote{The tensors in the upper layer represent their complex conjugates. For the sake of notational simplicity, we omit the conjugation sign.}
\begin{equation}
	\langle\tilde\Psi_e|\tilde\Psi_e\rangle
	= \includegraphics[valign=c,height=0.9cm]{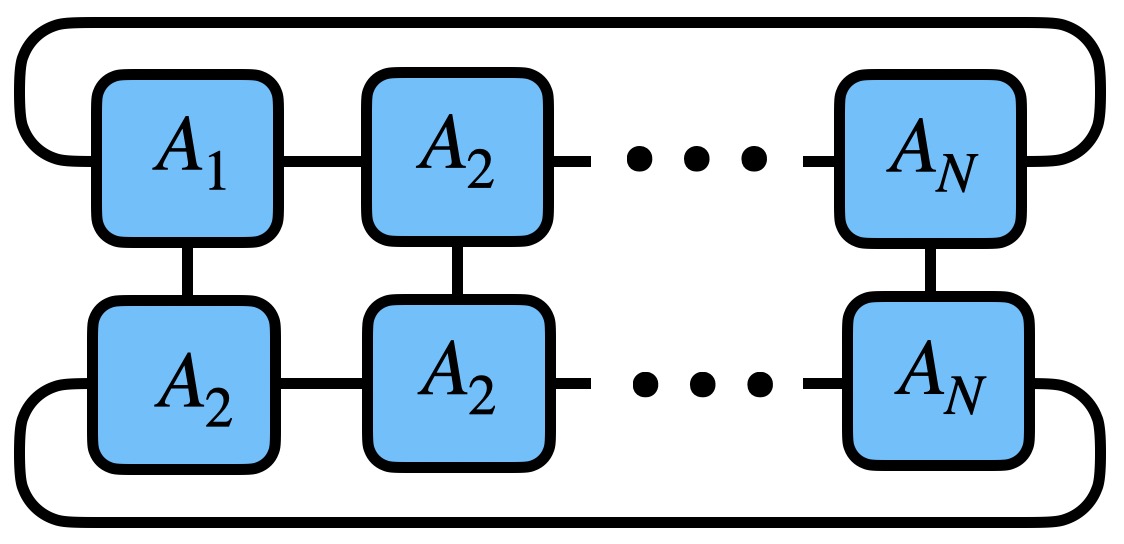}
	\simeq \includegraphics[valign=c,height=0.8cm]{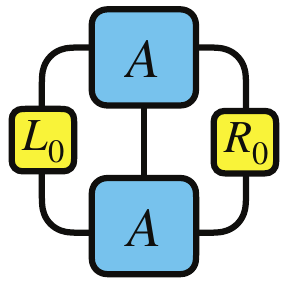} = 1.
\end{equation}
The scar Hamiltonian, assumed in the form
\begin{equation}
	H_d = \sum_{i=1}^N h_{[i,i+1,\cdots,i+m]},
\end{equation}
is degenerate on the subspace $\mathcal H_d$.
Without loss of generality, we assume that $H_d|_{\mathcal H_d}=0$, i.e.,
\begin{equation}\label{eq:h=0}
	\langle\tilde\Psi_e|H_d|\tilde\Psi_e\rangle
	= \includegraphics[valign=c,height=1cm]{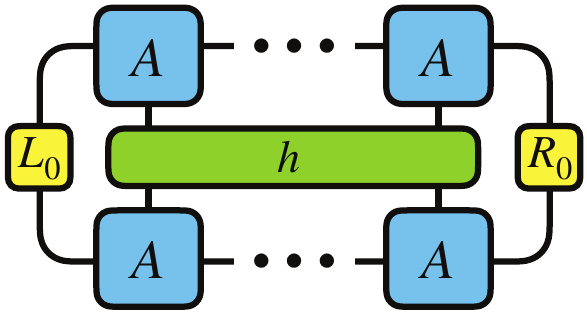} = 0.
\end{equation}
We further assume the transfer matrix of $|\tilde k\rangle$ contains no nontrivial Jordan block, allowing for its decomposition as follows:
\begin{equation}\label{eq:resolution}
	A^*\otimes A = \includegraphics[valign=c,height=0.8cm]{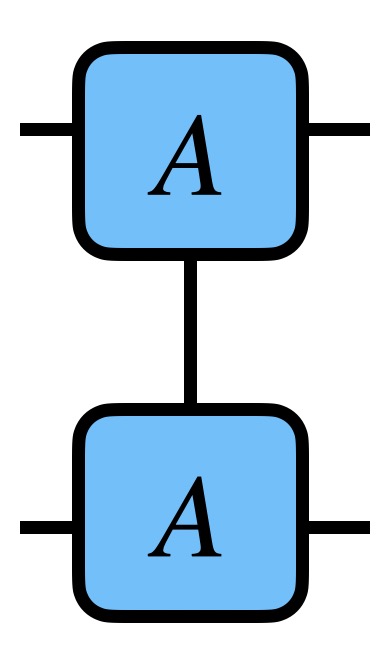} = \includegraphics[valign=c,height=0.7cm]{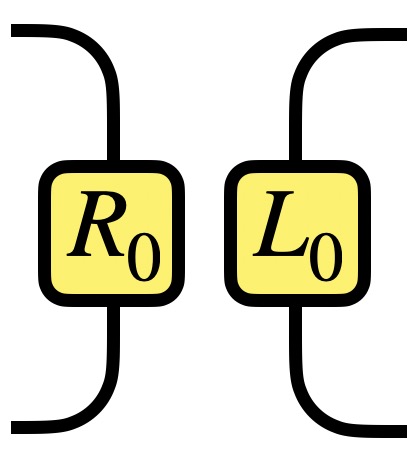} + \sum_{n=1}^{\chi^2-1} \lambda_n \includegraphics[valign=c,height=0.7cm]{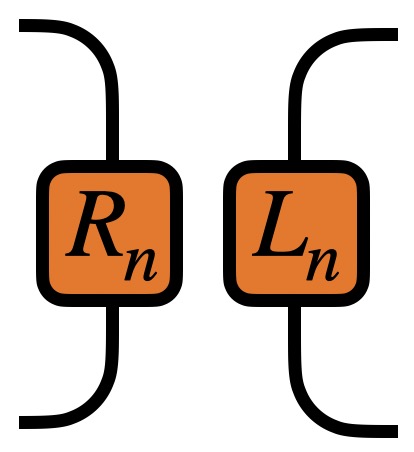}, \quad |\lambda_n| < 1.
\end{equation}
In the case of injective MPS, there is only one dominant eigenvector.

\subsection{Deformed quasi-Nambu-Goldstone theorem}
\begin{thm}[Deformed quasi-Nambu-Goldstone theorem]
For a translational invariant local Hamiltonian $\hat H$ hosting a deformed symmetric space as its scar subspace, both the energy expectation
\begin{equation*}
	\varepsilon_k \equiv \langle \tilde k|\hat H|\tilde k\rangle
\end{equation*}
and the energy variance
\begin{equation*}
	\delta \varepsilon^2_k \equiv \langle \tilde k|\hat H^2|\tilde k\rangle - \varepsilon_k^2
\end{equation*}
are continuous functions of $k$, converging to zero as $k$ approaches zero.
\end{thm}

In the subsequent discussion, we will focus on the finite-size lattice model, taking the $n\rightarrow \infty$ limit.

\begin{widetext}
\begin{lem}[Normalization]
The normalization of $|\tilde k\rangle$ is proportional to $\sqrt N$.
\end{lem}
\begin{proof}
In the tensor-network notation, the norm of the MPS can be represented as
\begin{equation*}
	\mathcal N_k^2 = N\cdot\includegraphics[valign=c,height=0.8cm]{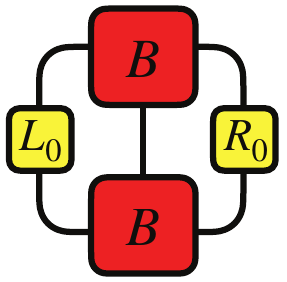}
	+\sum_{j<l} e^{-ik(j-l)}\includegraphics[valign=c,height=0.8cm]{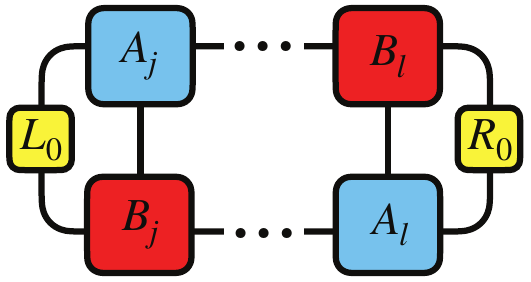}
	+\sum_{j<l} e^{ik(j-l)}\includegraphics[valign=c,height=0.8cm]{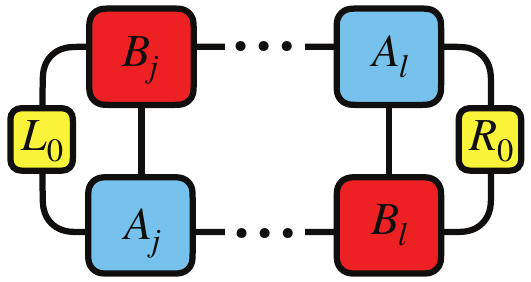}.
\end{equation*}
The second and third terms are complex conjugates of each other.
We evaluate the tensor contraction in the second term using the resolution (\ref{eq:resolution}) of the transfer matrix:
\begin{equation*}
	\includegraphics[valign=c,height=0.8cm]{pics/N-1}
	= \includegraphics[valign=c,height=0.8cm]{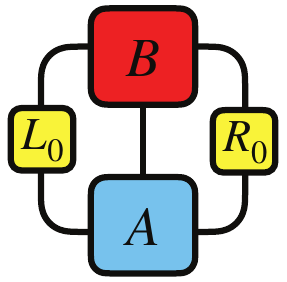}\includegraphics[valign=c,height=0.8cm]{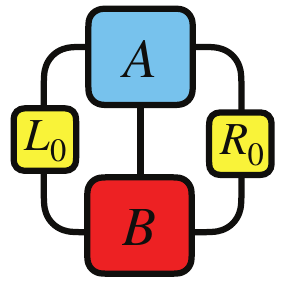}
	+ \sum_{n=1}^{\chi^2-1}\sum_{d>0} \lambda_n^d e^{-ik d}\includegraphics[valign=c,height=0.8cm]{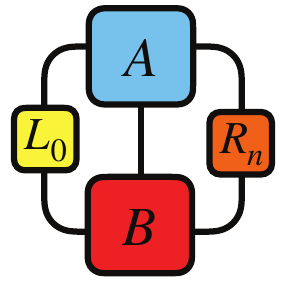}\includegraphics[valign=c,height=0.8cm]{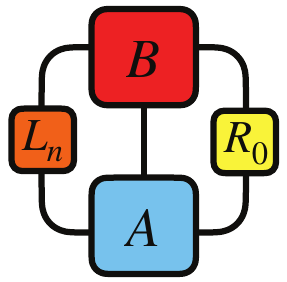}.
\end{equation*}
Similarly, for the third term:
\begin{equation*}
	\includegraphics[valign=c,height=0.8cm]{pics/N-2}
	= \includegraphics[valign=c,height=0.8cm]{pics/N-7}\includegraphics[valign=c,height=0.8cm]{pics/N-8}
	+ \sum_{n=1}^{\chi^2-1}\sum_{d>0} \lambda_n^d e^{+ik d} \includegraphics[valign=c,height=0.8cm]{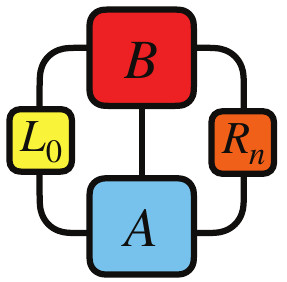}\includegraphics[valign=c,height=0.8cm]{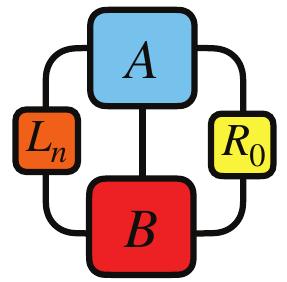}.
\end{equation*}
The summation of all non-decaying parts can be evaluated using the same trick in Eq.~(\ref{eq:sum-to-zero}), and the summation of the decaying parts converges since
\begin{equation*}
	\sum_{d>0} \lambda_n^d e^{\pm ik d} = \frac{\lambda_n}{e^{\mp ik}-\lambda_n}.
\end{equation*}
Thus, the final result is
\begin{equation}
	\mathcal N_k = \sqrt{N} \mathcal N_d(k)
\end{equation}
where $\mathcal N_d$ is defined by the following tensor contractions:
\begin{equation}\label{eq:tensor-norm}
	\mathcal N_d^2 = \includegraphics[valign=c,height=0.8cm]{pics/N-3}
	-\includegraphics[valign=c,height=0.8cm]{pics/N-7} \includegraphics[valign=c,height=0.8cm]{pics/N-8}
	+ \sum_{n=1}^{\chi^2-1}\left[\frac{\lambda_n}{e^{-ik}-\lambda_n} \includegraphics[valign=c,height=0.8cm]{pics/N-5}\includegraphics[valign=c,height=0.8cm]{pics/N-6} + C.c.\right].
\end{equation}
Since it is a finite sum, the lemma is proven.
\end{proof}

\begin{proof}[\textbf{Proof of the vanishing energy:}]
To begin, we introduce a tensor on a block consisting of $m+1$ sites:
\begin{equation*}
	\includegraphics[valign=c,height=0.44cm]{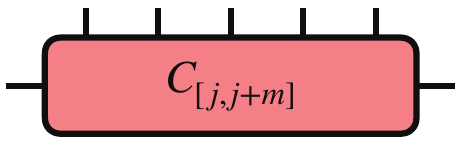}
	=\includegraphics[valign=c,height=0.4cm]{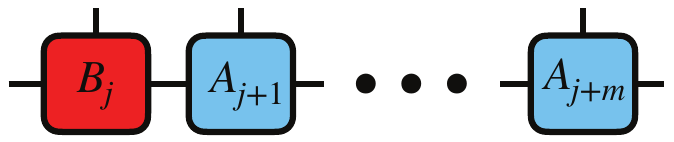} + e^{ik}\includegraphics[valign=c,height=0.4cm]{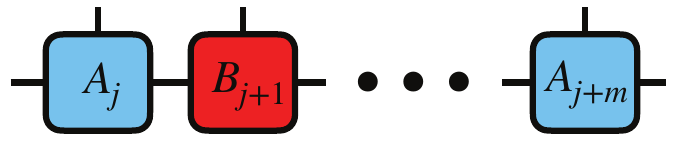} 
	 + \cdots+ e^{ikm}  \includegraphics[valign=c,height=0.4cm]{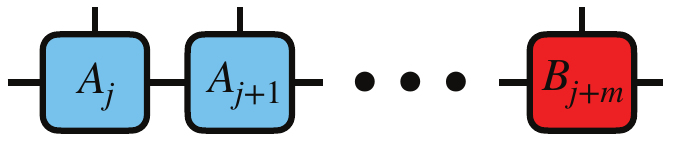}.
\end{equation*}
We decompose the expression of energy expectation into three distinct parts:
\begin{equation*}
\begin{aligned}
	\varepsilon_k =&\ S_1 + S_2 + S_2 \\
	=&\ \frac{1}{\mathcal N_d^2}\cdot \includegraphics[valign=c,height=0.9cm]{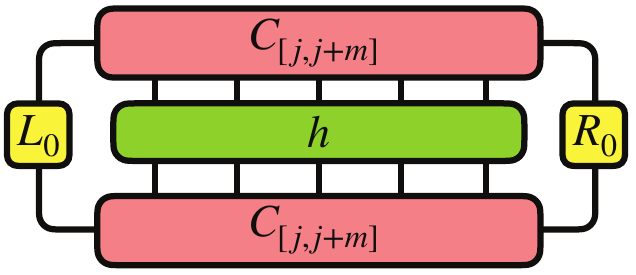} 
	 + \frac{1}{\mathcal N_d^2}\left[\sum_{j\notin[l,l+m]<l}e^{ik(j-l)}\includegraphics[valign=c,height=0.9cm]{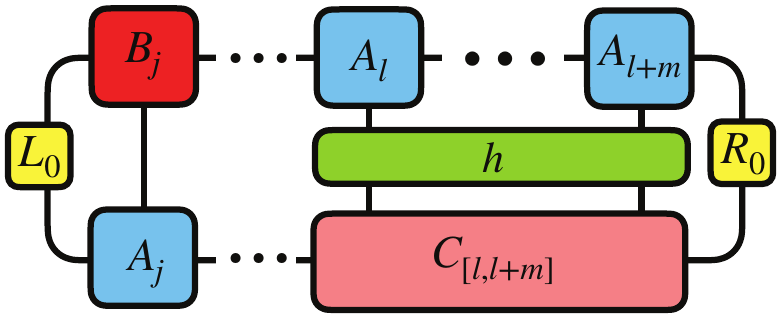} 
	+ C.c.\right] \\
	& + \frac{1}{\mathcal N_d^2}\left[\sum_{j,l\notin [i_0,i_0+m]}e^{ik(l-j)}\includegraphics[valign=c,height=0.9cm]{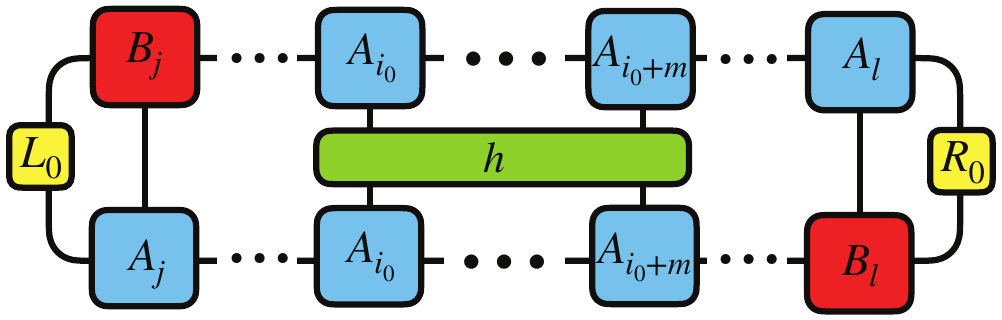}\right].
\end{aligned}
\end{equation*}
Note that the positioning of $B_j$ on the left and $B_l$ on the right is arbitrary, provided they are distinct from the primary block. 

The first term $S_1$ manifests as an analytic function of $k$.
Moving on to $S_2$, inserting the resolution (\ref{eq:resolution}), we partition $S_2$ into three components: $S_2=S_2^A + S_2^B + S_2^C$, where $S_2^A$ encompasses all non-decaying terms:
\begin{equation*}
\begin{aligned}
	S_2^A &= \frac{1}{\mathcal N_d^2}\sum_{j\notin[0,m]} e^{-ikj}\includegraphics[valign=c,height=0.9cm]{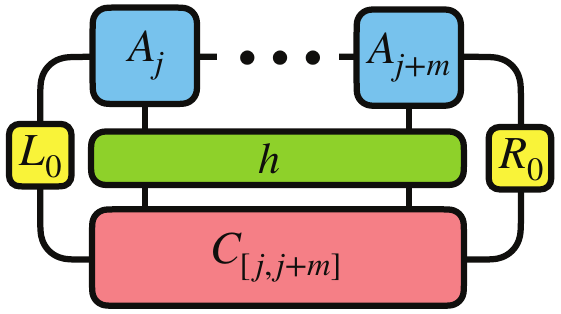} \cdot 
	\includegraphics[valign=c,height=0.9cm]{pics/N-7} +C.c.
\end{aligned}
\end{equation*}
Using the summation trick, we can convert the infinite sum over $j$ to a finite sum:
\begin{equation}
	\sum_{j\notin [0,m]} e^{ikj} = \sum_{j=0}^m e^{ikj} = \frac{1-e^{ik(m+1)}}{1-e^{ik}}.
\end{equation}
The $S_2^B$ part involves a decaying summation:
\begin{equation*}
\begin{aligned}
	S_2^B &= \frac{1}{\mathcal N_d^2} \sum_{n=1}^{\chi^2-1}\sum_{j>0}\left[ \lambda_n^j e^{+ikj}\includegraphics[valign=c,height=0.9cm]{pics/N-9}\includegraphics[valign=c,height=0.9cm]{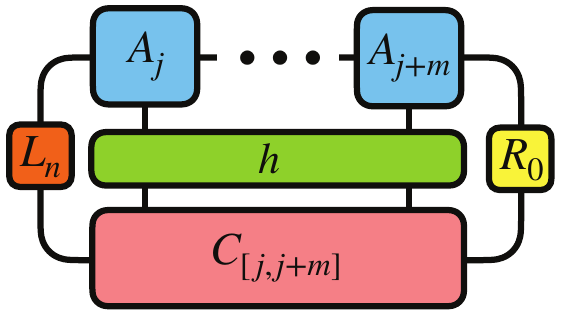}+C.c.\right].
\end{aligned}
\end{equation*}
For the infinite sum over $j$, note that the tensor-contraction part is a constant, and the summation becomes:
\begin{equation}
	\sum_{j>0} = \lambda_n^j e^{+ikj}
	= \frac{\lambda_n}{e^{-ik}-\lambda_n}.
\end{equation}
It is therefore convergent.
Similarly, summation of $S_2^C$ part mirrors that of $S_2^B$ part:
\begin{equation*}
\begin{aligned}
	S_2^C &= \frac{1}{\mathcal N_d^2} \sum_{n=1}^{\chi^2-1}\sum_{j>0}\left[\lambda_n^j e^{-ikj}\includegraphics[valign=c,height=0.9cm]{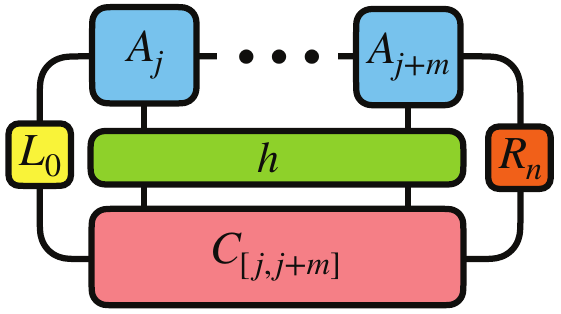}\includegraphics[valign=c,height=0.9cm]{pics/N-6}+C.c.\right],
\end{aligned}
\end{equation*}
where the infinite summation over $j$ is:
\begin{equation}
	 \sum_{j>0} \lambda_n^j e^{-ikj} =\frac{\lambda_n}{e^{+ik}-\lambda_n}.
\end{equation}
Thus, $S_2$ is also an analytic function.

Considering $S_3$, which encompasses all tensor contractions where $B$'s and $h$ are separated, we also partition the summation into different classes.
First, we examine terms where two $B$'s are connected:
\begin{equation*}
\begin{aligned}
	S_3^A &= \frac{1}{\mathcal N_d^2}\sum_{n=1}^{\chi^2-1} \sum_{d>0} \lambda_n^{d} \left[\includegraphics[valign=c,height=0.9cm]{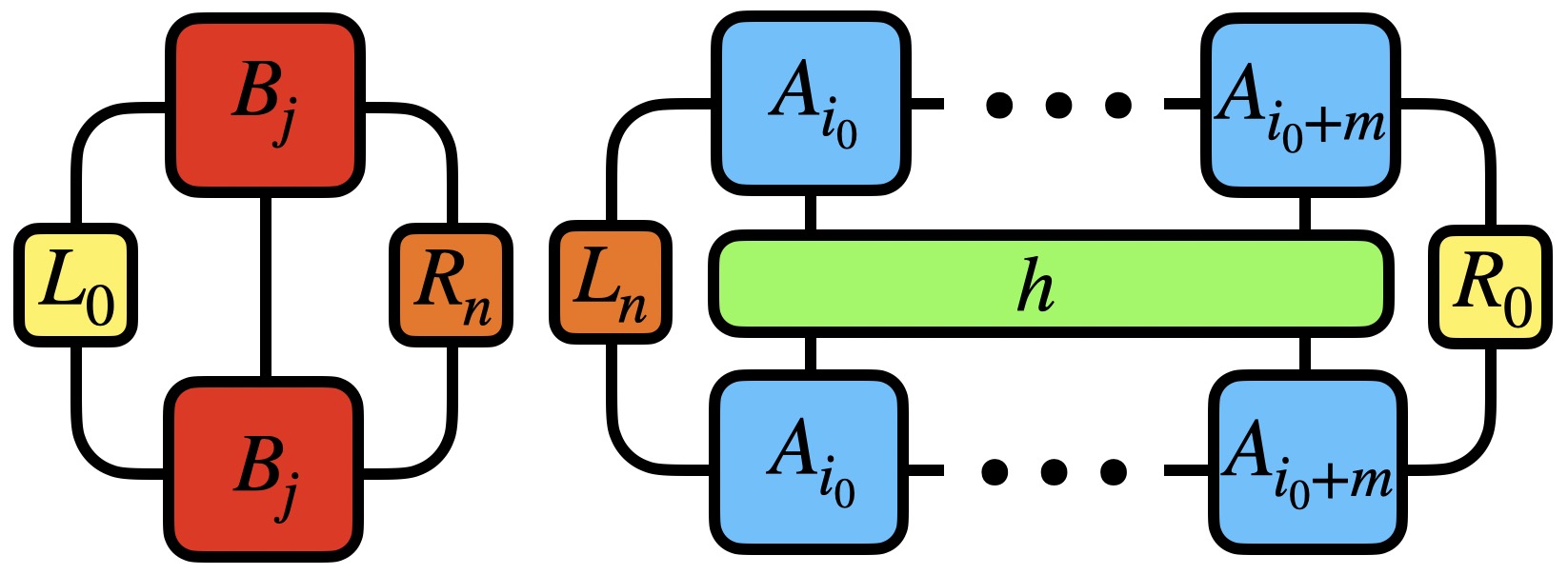} + \includegraphics[valign=c,height=0.9cm]{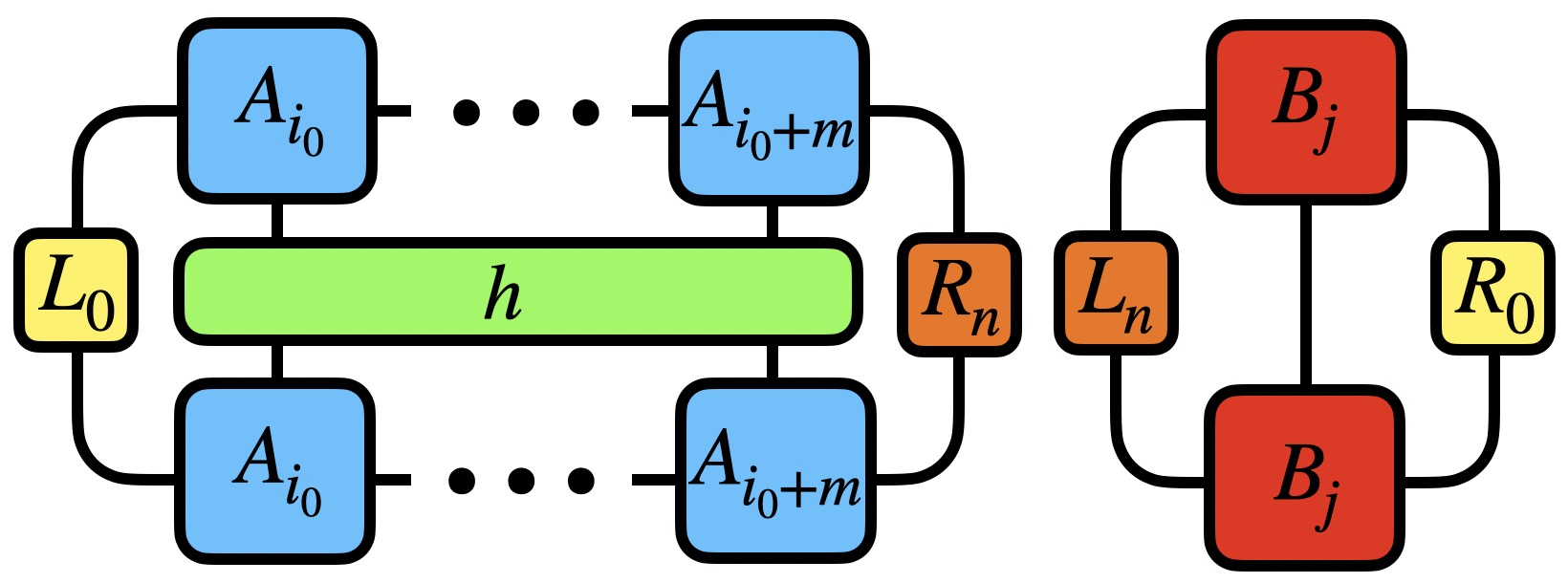} \right], 
\end{aligned}
\end{equation*}
where the infinite sum over $d$ is convergent due to the exponential decay of coefficients:
\begin{equation}
	\sum_{d>0} \lambda_n^d = \frac{\lambda_n}{\lambda_n-1}.
\end{equation}
Then, we explore contractions where $B$'s are separated but with decaying transfer modes:
\begin{equation*}
\begin{aligned}
	S_3^B =&\ \frac{1}{\mathcal N_d^2}\sum_{n,n'=1}^{\chi^2-1}
	\left[ \sum_{j<i_0}\sum_{l<j} \lambda_n^{i_0-j}\lambda_{n'}^{j-l} e^{-ik(j-l)}\includegraphics[valign=c,height=0.9cm]{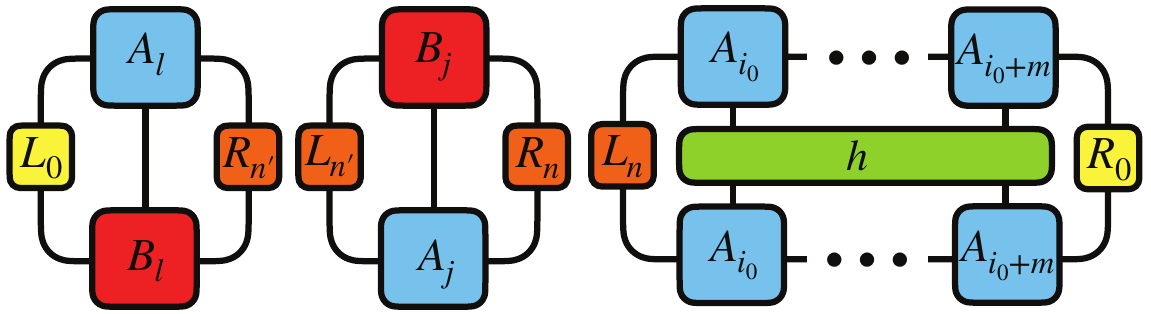} \right. \\
	&+ \sum_{l>i_0+m}\sum_{j>l} \lambda_n^{l-i_0+m} \lambda_{n'}^{j-l}e^{-ik(j-l)}\includegraphics[valign=c,height=0.9cm]{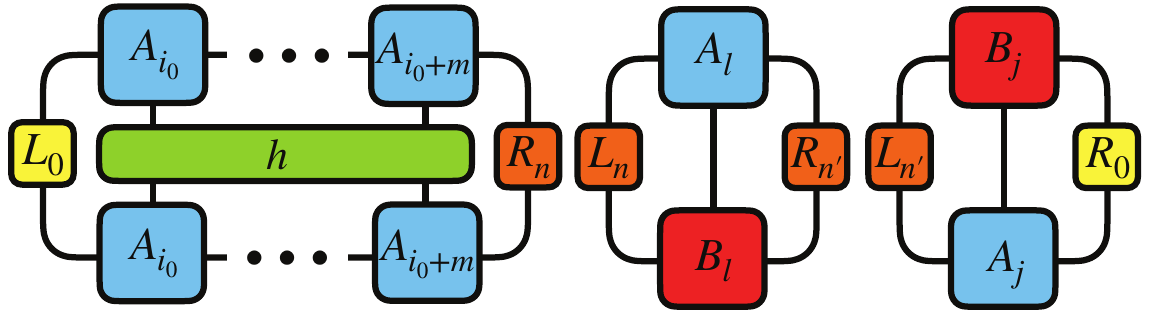} \\
	&+ \left. \sum_{l<i_0}\sum_{j>i_0+m} \lambda_n^{j-i_0-m}\lambda_{n'}^{i_0-l}e^{-ik(j-l)}\includegraphics[valign=c,height=0.9cm]{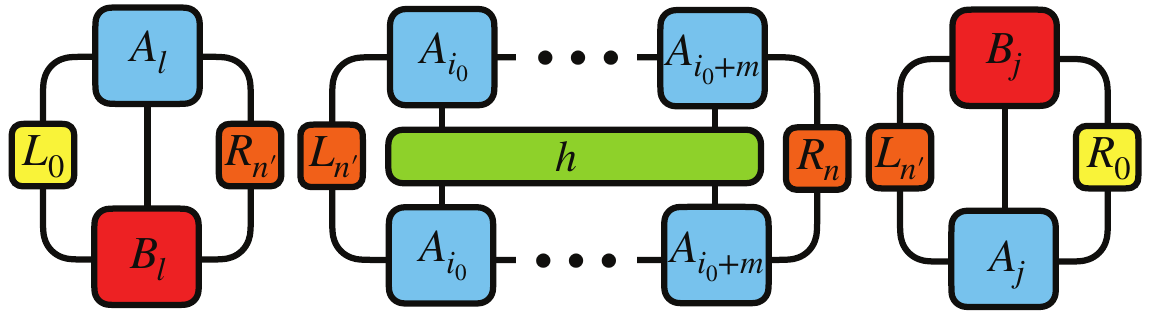} +C.c.\right].
\end{aligned}
\end{equation*}
The infinite summations produces the following convergent factors:
\begin{eqnarray}
	\sum_{j<i_0}\sum_{l<j} \lambda_n^{i_0-j}\lambda_{n'}^{j-l} e^{-ik(j-l)} 
	&=& \frac{\lambda_n}{1-\lambda_n}\frac{\lambda_{n'}}{e^{ik}-\lambda_{n'}}, \\
	\sum_{l>i_0+m}\sum_{j>l} \lambda_n^{l-i_0+m} \lambda_{n'}^{j-l}e^{-ik(j-l)}
	&=& \frac{\lambda_n}{1-\lambda_n} \frac{\lambda_{n'}}{e^{ik}-\lambda_{n'}}, \\
	\sum_{l<i_0}\sum_{j>i_0+m} \lambda_n^{j-i_0-m}\lambda_{n'}^{i_0-l}e^{-ik(j-l)}
	&=& \frac{\lambda_n}{e^{ik}-\lambda_n} \frac{\lambda_{n'}}{e^{ik}-\lambda_{n'}} e^{-ikm}.
\end{eqnarray}
The remaining part of $S_3$ is:
\begin{equation*}
\begin{aligned}
	S_3^C =&\ \frac{1}{\mathcal N_d^2}\sum_{n=1}^{\chi^2-1} \left[ \sum_{j<i_0}\sum_{l\notin[j,i_0+m]} \lambda_n^{i_0-j} e^{-ik(j-l)}\includegraphics[valign=c,height=0.9cm]{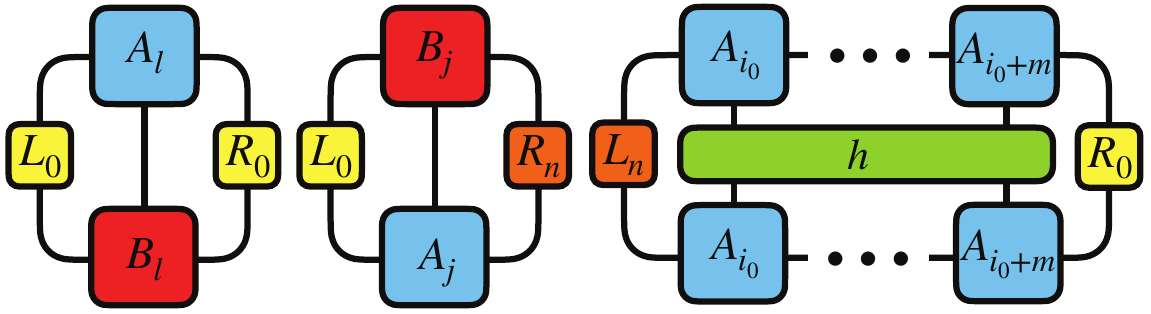} \right. \\
	&+ \left. \sum_{l>i_0+m}\sum_{j\notin[i_0,l]} \lambda_n^{l-i_0-m}e^{-ik(j-l)}\includegraphics[valign=c,height=0.9cm]{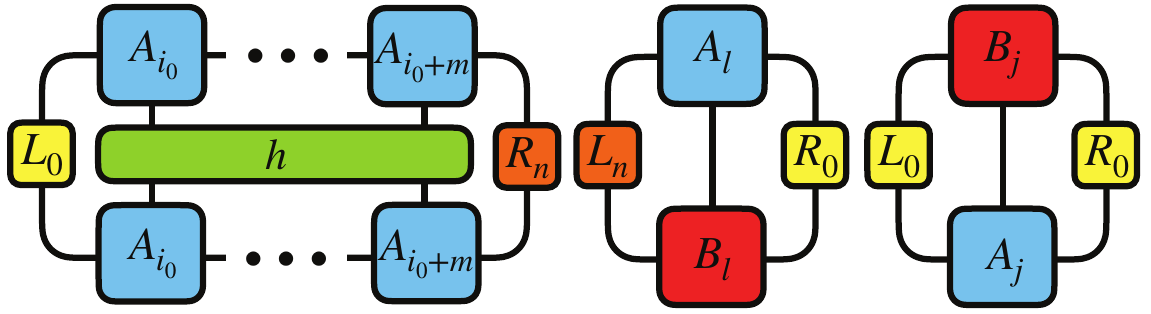} +C.c.\right].
\end{aligned}
\end{equation*}
To demonstrate the convergence of all terms, we first utilize a similar technique as in Eq.~(\ref{eq:sum-to-zero}) to convert the infinite sum outside the interval to a finite sum inside the interval. 
Then, we need to establish the convergence of the following summation:
\begin{equation}
	\sum_{j=1}^\infty \lambda_n^j \sum_{s=0}^{j+m}e^{-ik s}
	= \sum_{j=1}^\infty \lambda_n^j e^{-i k (j+m)} \frac{1-e^{i k (j+m+1)}}{1-e^{i k}}
	= \frac{\lambda_n}{1-\lambda_n}  \frac{(e^{-i k m}+e^{i k}) \lambda_n-e^{2i k}+e^{-i k m}}{\left(1-e^{i k}\right) \left(e^{i k}-\lambda_n\right)}.
\end{equation}
This expression constitutes analytic functions. Thus, we complete the proof.
\end{proof}

\begin{proof}[\textbf{Proof of the vanishing variance:}]
Until this point, throughout our proof of the theorem, we consistently transform expressions into finite sums of analytic functions. 
These expressions not only confirm the absence of gaps in the energy variance but also enable the direct derivation of the variance's analytic form for specific models. 
However, in the case of frustrated scenarios explored in this section, the general result has the form of infinite summations. 
Nonetheless, our primary objective persists: to rigorously establish the absolute convergence of this quantity as a summation of analytic functions.

Here we introduce another notation for two connected Hamiltonian local terms:
\begin{equation}
\begin{aligned}
	\includegraphics[valign=c,height=0.5cm]{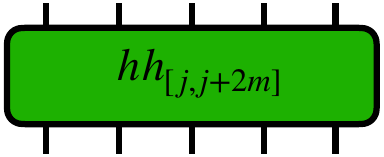}
	\equiv &\ \includegraphics[valign=c,height=0.5cm]{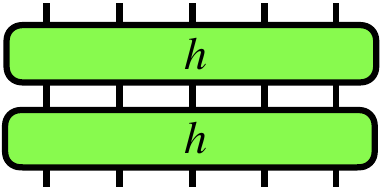} + \includegraphics[valign=c,height=0.5cm]{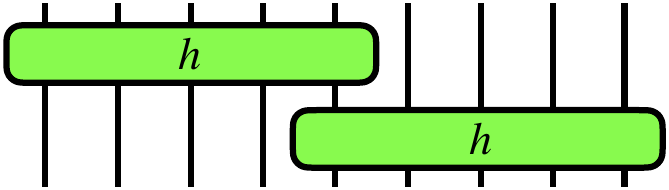} +
	 \includegraphics[valign=c,height=0.5cm]{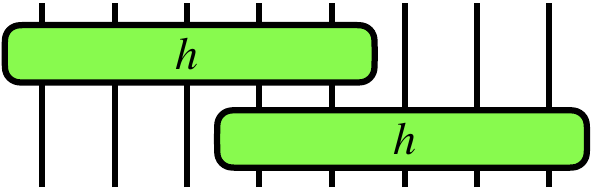}
	+\includegraphics[valign=c,height=0.5cm]{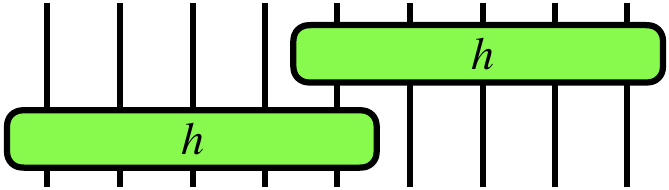}.
\end{aligned}
\end{equation}
We divided the energy variance into two parts: $\delta\varepsilon_k^2 = S_1 + S_2$, one for the all-connected term:
\begin{equation}
	S_1 = \frac{1}{\mathcal N_d^2} \includegraphics[valign=c,height=1.0cm]{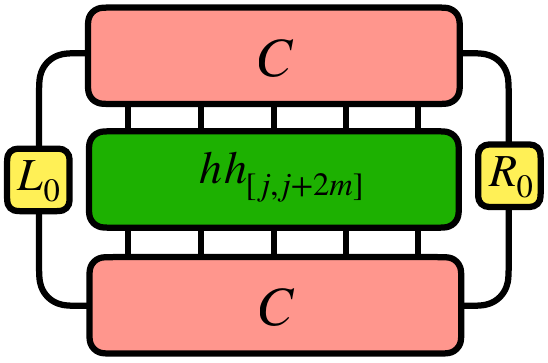},
\end{equation}
which is a short-hand notation for the following summation: 
\begin{equation*}
\begin{aligned}
	S_1 =&\ \frac{1}{\mathcal N_d^2}\includegraphics[valign=c,height=1.0cm]{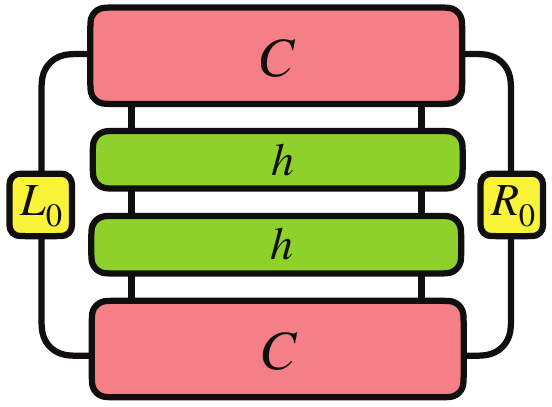}
	+ \frac{e^{ikm}}{\mathcal N^2_d}\includegraphics[valign=c,height=1.0cm]{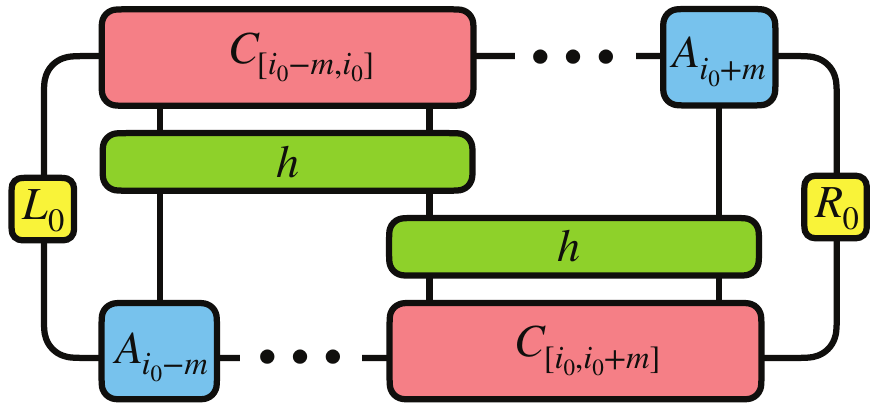} + \cdots \\
	&+ \frac{e^{ik(i_0-j)}}{\mathcal N_d^2}\includegraphics[valign=c,height=1.0cm]{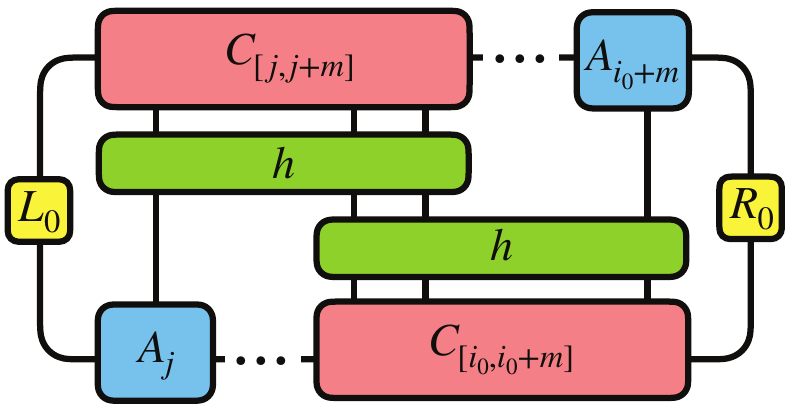} + \cdots
	+ \frac{e^{-ikm}}{\mathcal N_d^2}\includegraphics[valign=c,height=1.0cm]{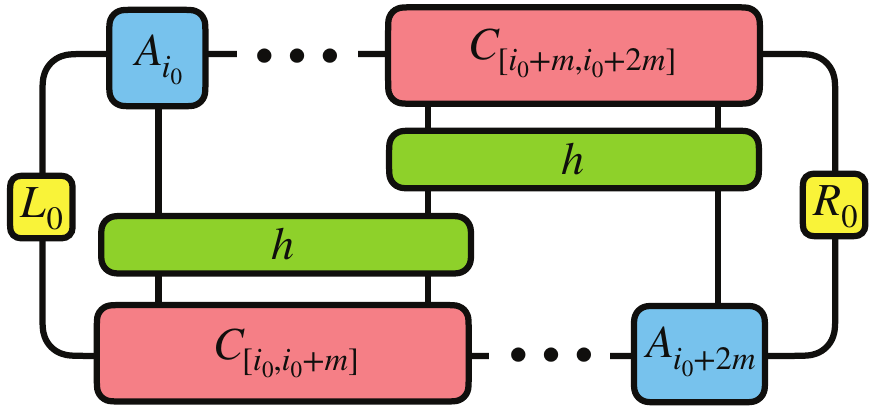}.
\end{aligned}
\end{equation*}

The critical aspect of the proof involves demonstrating the convergence of the infinite summation in the ``unconnected'' part.

Considering, without loss of generality, the tensor contraction where $B_j$ tensor in the upper left lies outside two Hamiltonian local terms. 
Two possibilities arise for this scenario. Firstly, 
\begin{equation}
	S_2^A = \frac{1}{\mathcal N_d^2}\sum_{p}\sum_{j<i_0}\sum_{l>j} e^{-ik(j-l)}\includegraphics[valign=c,height=0.9cm]{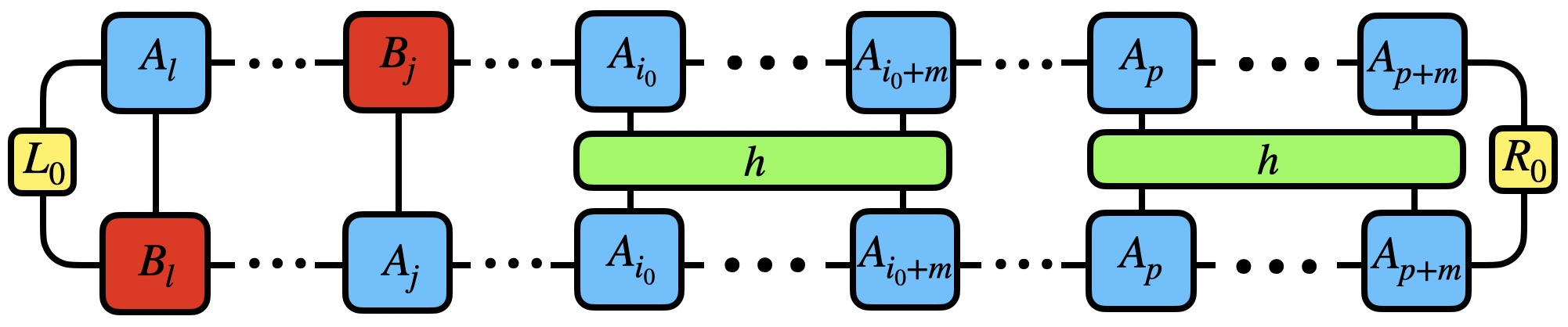},
\end{equation}
where two local Hamiltonian terms need not be separated. 
We claim that inserting resolution (\ref{eq:resolution}) between $B_j$ and Hamiltonian local terms will result in the dominant eigenvector leading to a zero, as the term involves the following contraction: 
\begin{equation*}
	\sum_{p} \includegraphics[valign=c,height=0.9cm]{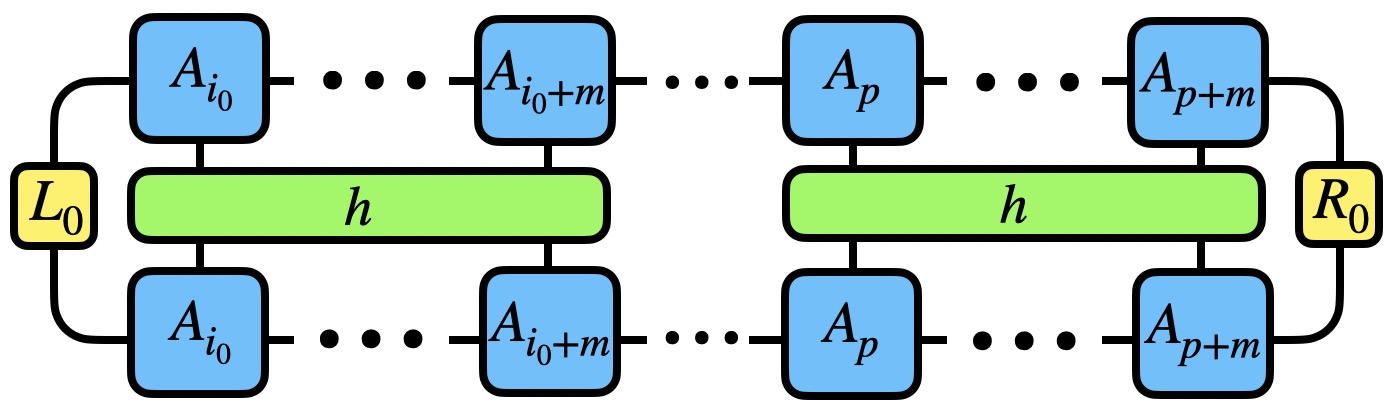} \sim \langle 
	\Psi_e|\hat H^2|\Psi_e\rangle = 0.
\end{equation*}
The other term in $S_2$ is
\begin{equation}
	S_2^B = \frac{1}{\mathcal N_d^2}\sum_{p}\sum_{j<i_0}\sum_{l<j} e^{-ik(j-l)}\includegraphics[valign=c,height=0.9cm]{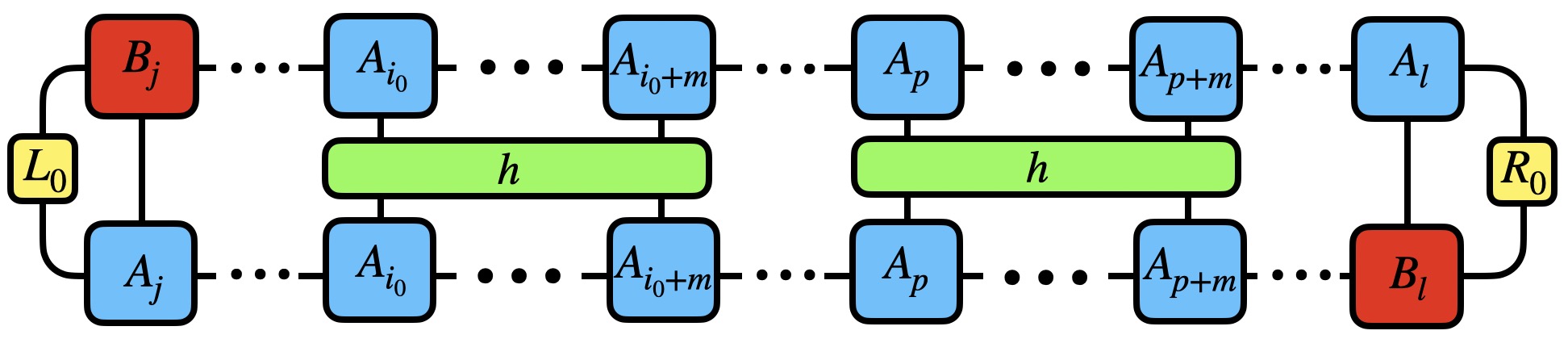},
\end{equation}
where $B_l$ in the bottom layer need not be located on the right of the Hamiltonian local terms. 
Upon inserting resolution (\ref{eq:resolution}) between $B_j$ and Hamiltonian local terms, the dominant eigenvector will also lead to a zero since it contains the contraction:
\begin{equation*}
	\frac{1}{\mathcal N_d} \sum_{p} \sum_{l>j}\includegraphics[valign=c,height=0.9cm]{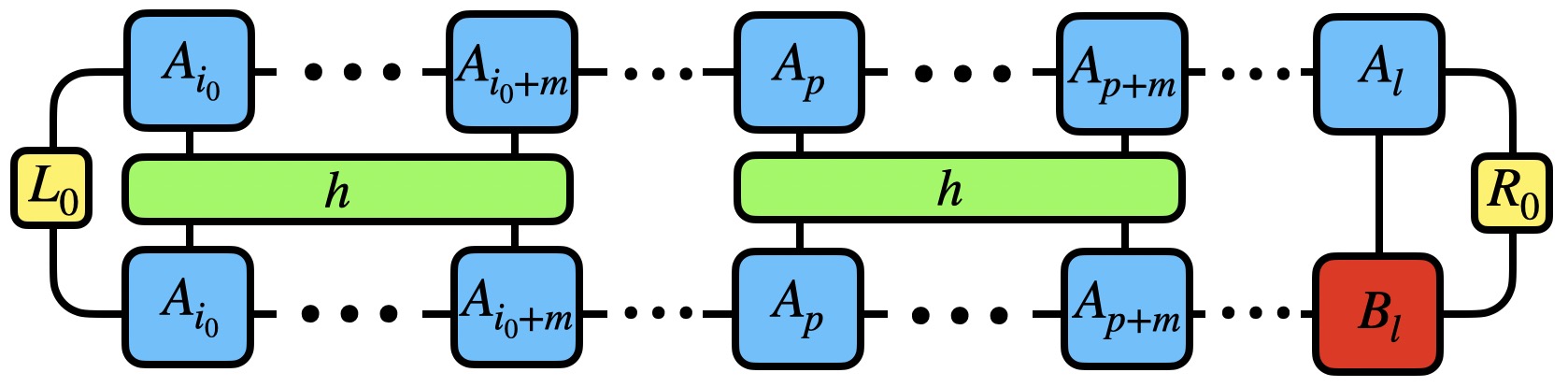} \sim \langle 
	\Psi_e|\hat H^2|k\rangle = 0.
\end{equation*}
Furthermore, when two Hamiltonian local terms are separate, there should be no dominant vector inserted between them; otherwise, a contraction in either
\begin{equation*}
	\includegraphics[valign=c,height=0.9cm]{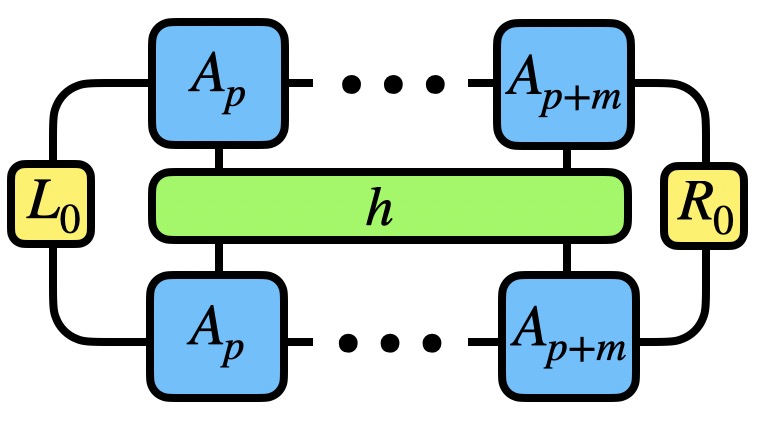} \sim \langle\Psi_e|\hat H|\Psi_e\rangle =0,
\end{equation*}
form or in the 
\begin{equation*}
	\frac{1}{\mathcal N_d}\sum_{l>j} \includegraphics[valign=c,height=0.9cm]{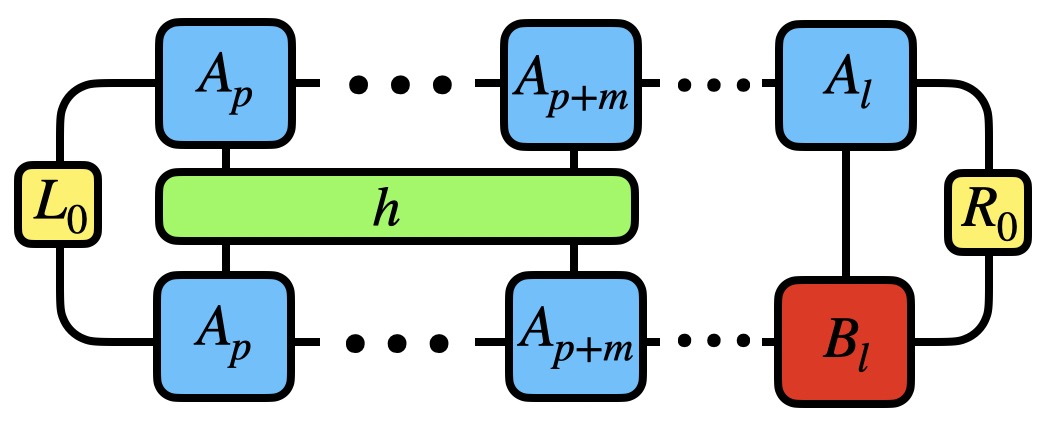} \sim \langle\Psi_e|\hat H|k\rangle =0,
\end{equation*}
form will result in zero.

Thus, we arrive at an effective quasi-short-range tensor: 
\begin{equation}
	\sum_{j<i_0, p>i_0+m} \lambda_n^{i_0-j}\lambda_{n'}^{p-i_0-m} \includegraphics[valign=c,height=0.9cm]{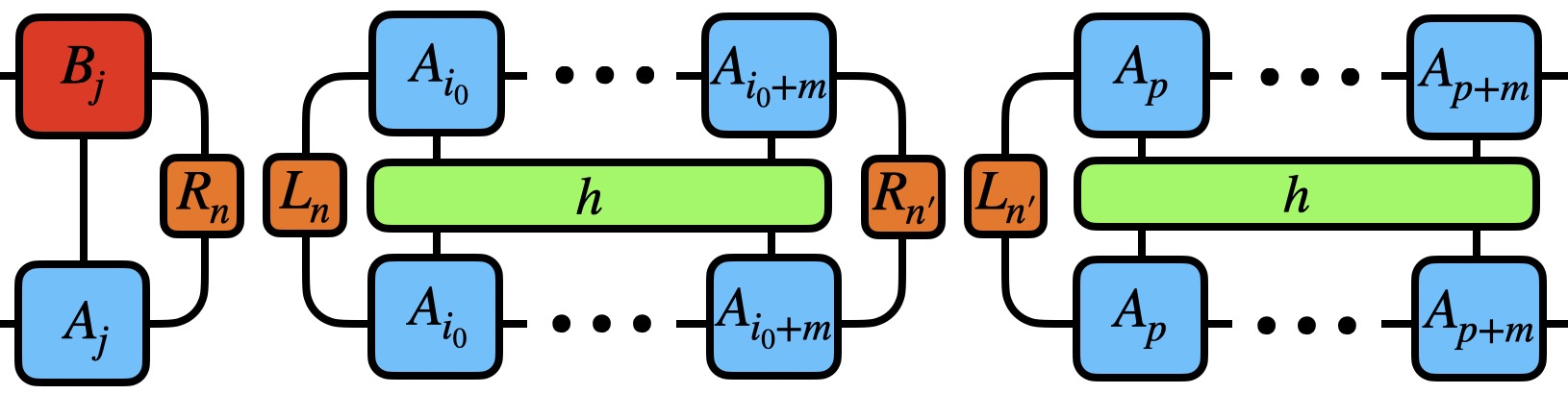}.
\end{equation}
The remaining convergence to verify is:
\begin{equation*}
	\sum_{j<i_0}\sum_{p>i_0+m}\lambda_n^{i_0-j}\lambda_{n'}^{p-i_0-m} \sum_{l\notin[j,p+m]} e^{ik(j-l)}\includegraphics[valign=c,height=0.9cm]{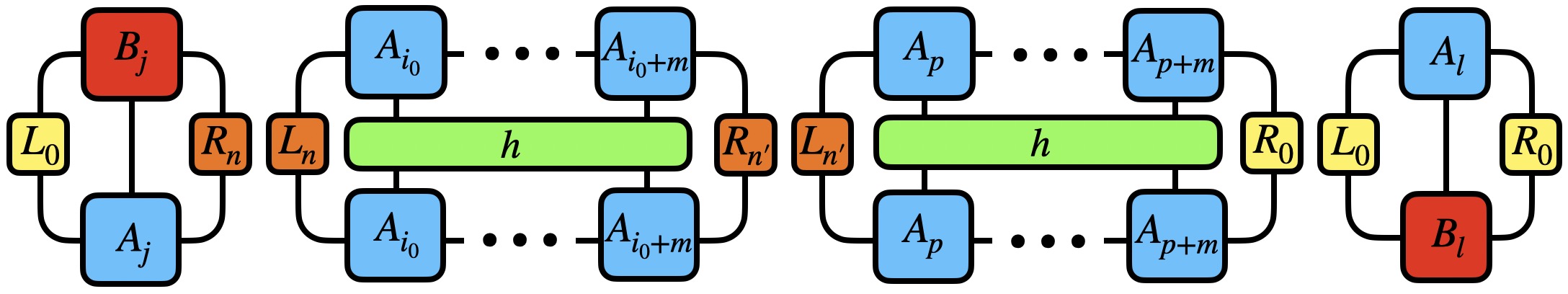}.
\end{equation*}
As before, the singled-out infinite summation can be converted into a finite summation. 
Consequently, all remaining infinite summations contain an exponentially decaying factor depending on the size of the nontrivial block, ensuring absolute convergence. 
Thus, we have established that $\delta\varepsilon_k^2$ is a continuous function of $k$.
\end{proof}
\end{widetext}

\section{Deformed Quasi-Nambu-Goldstone Modes in AKLT Model}
\label{apx:dqngmaklt}

\subsection{Generalized AKLT scar model}
Although the AKLT model (in the degenerate limit) is not frustration-free,\footnote{While the local terms of the original AKLT Hamiltonian are frustration-free, in the degenerate limit, we introduce an additional term $\hat{s}_j^z$ to induce degeneracy within the scar tower. This modification renders the scar tower no longer an eigenstate of local $\hat s_j$, thereby introducing frustration into the Hamiltonian.} since it has reflection symmetry, the dispersion cannot be linear and thus there is no velocity for the $k=\pi$ wave packet, as we can see from Fig.~\ref{fig:AKLT-origin}.
However, it is possible to induce velocity by adding scar tower-preserving perturbations into the AKLT Hamiltonian.
Such terms can be systematically found using the quantum inverse method \cite{inverse-method,inverse-method-2}, which takes a list of operators $\{\hat O^{(i)}, i=1,\dots,m\}$ as input, and calculate the $m\times m$ covariance matrix:
\begin{equation}
	C_{ij} = \operatorname{Tr}\left[\hat\rho \hat O^{(i)} \hat O^{(j)}\right] - \operatorname{Tr}\left[\hat\rho \hat O^{(i)} \right] \operatorname{Tr}\left[\hat\rho \hat O^{(j)}\right],
\end{equation}
where $\hat \rho$ is an arbitrary combination of the tower states:
\begin{equation*}
	\hat\rho = \sum_n p_n |\Psi_n\rangle\langle\Psi_n|,\quad \sum_n p_n=1.
\end{equation*}
If there is a combination
\begin{equation}
	\hat H(\mathbf{J}) = \sum_{i=1}^m J_i \hat O^{(i)}
\end{equation}
of which the scar tower is degenerate scar states, then
\begin{equation*}
\begin{aligned}
	\sum_{i,j=1}^m J_i C_{ij} J_j
	&= \operatorname{Tr}[\hat\rho \hat H^2] - \operatorname{Tr}[\hat\rho \hat H] \operatorname{Tr}[\hat\rho H] \\
	&= \sum_n p_n \langle\Psi_n|\hat H^2 |\Psi_n\rangle - E^2
	=0.
\end{aligned}	
\end{equation*}
The solution is the null space of covariance matrix $C_{ij}$.

\begin{figure}
	\centering
	\includegraphics[width=0.45\textwidth]{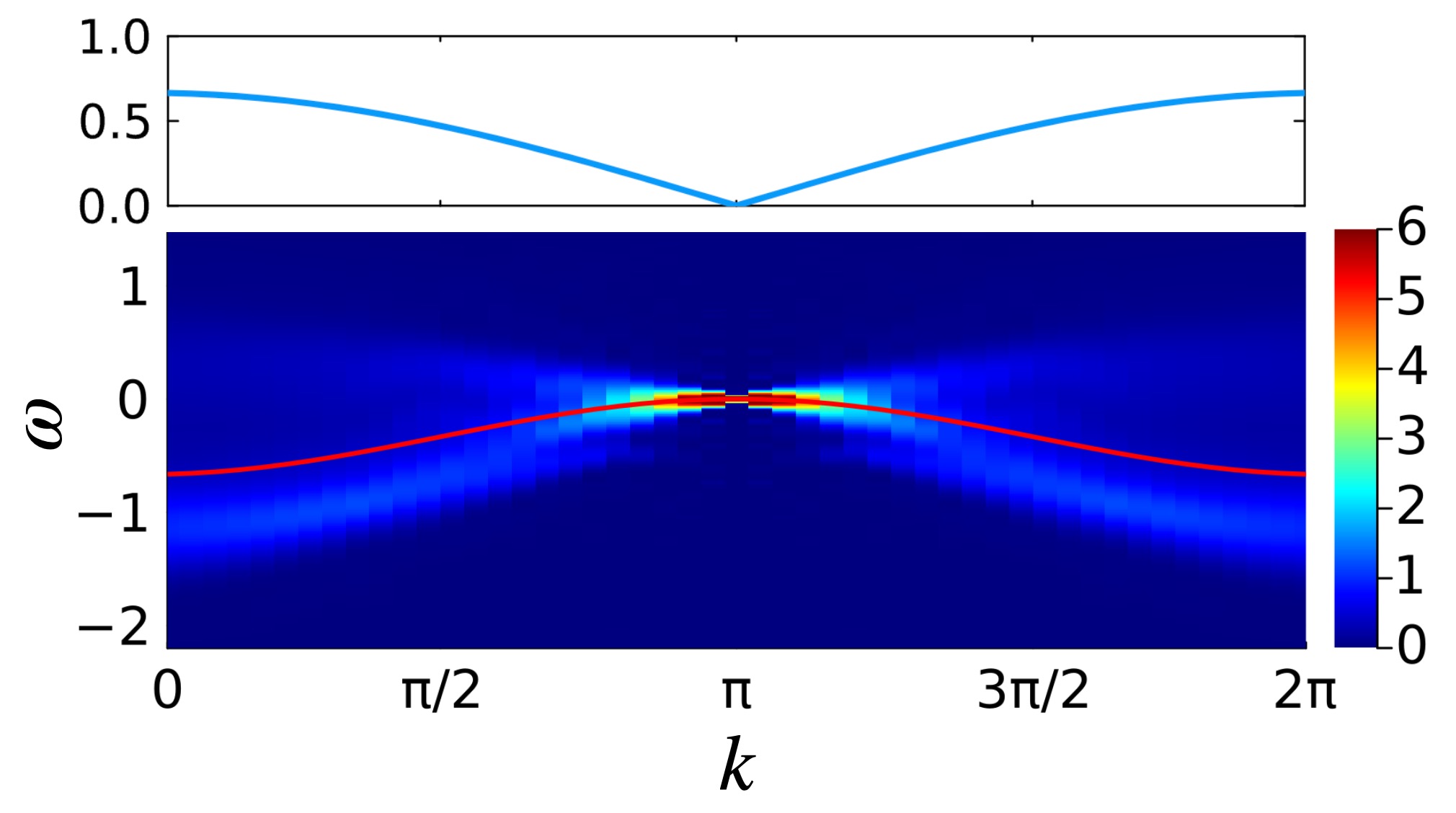}
	\caption{Bottom: Spectral function $A(k,\omega)$ of the original AKLT mode, and energy expectation $\varepsilon_k$. Top: energy variance $\delta\varepsilon_k$.}
	\label{fig:AKLT-origin}
\end{figure}

In practice, we usually focus on the translational-invariant Hamiltonian.
That is, the trial operator will be in the form:
\begin{equation}
	\hat O^{(i)} = \sum_{j=1}^N \hat o^{(i)}_j
\end{equation}
For spin-1 degree of freedom, we first choose the basis for the local operators.
A canonical choice is the Gell-Mann matrices:
\begin{equation}
\begin{aligned}
\lambda^{(1)} &= \left(
\begin{array}{ccc}
 0 & 1 & 0 \\
 1 & 0 & 0 \\
 0 & 0 & 0 \\
\end{array}
\right), &
\lambda^{(2)} &= \left(
\begin{array}{ccc}
 0 & -i & 0 \\
 i & 0 & 0 \\
 0 & 0 & 0 \\
\end{array}
\right), \\
\lambda^{(3)} &= \left(
\begin{array}{ccc}
 1 & 0 & 0 \\
 0 & -1 & 0 \\
 0 & 0 & 0 \\
\end{array}
\right), &
\lambda^{(4)} &= \left(
\begin{array}{ccc}
 0 & 0 & 1 \\
 0 & 0 & 0 \\
 1 & 0 & 0 \\
\end{array}
\right), \\
\lambda^{(5)} &= \left(
\begin{array}{ccc}
 0 & 0 & -i \\
 0 & 0 & 0 \\
 i & 0 & 0 \\
\end{array}
\right), &
\lambda^{(6)} &= \left(
\begin{array}{ccc}
 0 & 0 & 0 \\
 0 & 0 & 1 \\
 0 & 1 & 0 \\
\end{array}
\right), \\
\lambda^{(7)} &= \left(
\begin{array}{ccc}
 0 & 0 & 0 \\
 0 & 0 & -i \\
 0 & i & 0 \\
\end{array}
\right), & 
\lambda^{(8)} &= \frac{1}{\sqrt{3}}\left(
\begin{array}{ccc}
 1 & 0 & 0 \\
 0 & 1 & 0 \\
 0 & 0 & -2 \\
\end{array}
\right).
\end{aligned}
\end{equation}
For two-site operators, we define
\begin{equation}\label{eq:lambda-def}
	\hat \lambda^{(i,j)} \equiv \hat \lambda^{(i)}\otimes\hat \lambda^{(j)}, \quad
	\hat \lambda^{[i,j]} \equiv \hat \lambda^{(i,j)} - \hat \lambda^{(j,i)}.
\end{equation}
To find perturbations that induce nonzero velocity, we need to consider the odd-parity operators.
For simplicity, we only consider $S^z_\text{tot}$-conserving, translational-invariant Hamiltonian.
There are 6 independent odd-parity two-site operators:
\begin{equation}
\begin{aligned}
	\hat o^{(1)} &= \hat \lambda^{[38]},\\
	\hat o^{(2)} &= \hat \lambda^{[12]},\\
	\hat o^{(3)} &= \hat \lambda^{[45]},\\
	\hat o^{(4)} &= \hat \lambda^{[67]},\\
	\hat o^{(5)} &= \hat \lambda^{[16]} + \hat \lambda^{[27]},\\
	\hat o^{(6)} &= \hat \lambda^{[17]} + \hat \lambda^{[62]}.
\end{aligned}
\end{equation}
We do not find any solution for perturbation in the nearest-neighbor odd-parity operator space.
However, we can further search for the three-site operator in the form:
\begin{equation*}
	\hat \lambda_j^{(a)} \hat o_{j-1,j+1}^{(b)}, \quad a=3,8;\ b = 1,\dots,6.
\end{equation*}
We find two independent solutions:
\begin{equation*}
\begin{aligned}
	\hat V_1 &= \sum_{j=1}^N \left[\hat \lambda^{(3)}_j \otimes \hat o^{(2)}_{j-1,j+1}  - \hat s^z_j \otimes \hat o^{(3)}_{j-1,j+1} + \hat b_j \otimes \hat o^{(4)}_{j-1,j+1} \right], \\
	\hat V_2 &= \sum_{j=1}^N \hat s_j^z \otimes \left[\hat o^{(2)}+\hat o^{(4)}+2\hat o^{(6)}\right]_{j-1,j+1},
\end{aligned}
\end{equation*}
where we have defined the operator 
\begin{equation}
	\hat b = \frac{\sqrt 3}{2}\hat \lambda^{(8)} - \frac{1}{2}\hat \lambda^{(3)} = \operatorname{diag}(0,1,-1).
\end{equation}

\subsection{Quasi-Nambu-Goldstone modes}
The generalized scar Hamiltonian for AKLT tower is
\begin{equation}\label{eq:gen-aklt-ham}
	\hat H(\alpha,\beta) = \hat H_\text{AKLT} + \alpha \hat V_1 + \beta \hat V_2,
\end{equation}
where the (degenerate) AKLT model is fixed to 
\begin{equation*}
	\hat H_\text{AKLT} =\sum_{j}\left[\frac{1}{2}\hat{\mathbf{s}}_j\cdot\hat{\mathbf{s}}_{j+1}+\frac{1}{6}(\hat{\mathbf{s}}_j\cdot\hat{\mathbf{s}}_{j+1})^2-\hat{s}^z_j\right],
\end{equation*}
and $\alpha,\beta$ are two free parameters.
In the main text, the parameters are fixed to:
\begin{equation}\label{eq:aklt-pert}
	\hat V = -0.3 \hat V_1 -0.1 \hat V_2.
\end{equation}

\begin{figure}
	\centering
	\includegraphics[width=0.95\linewidth]{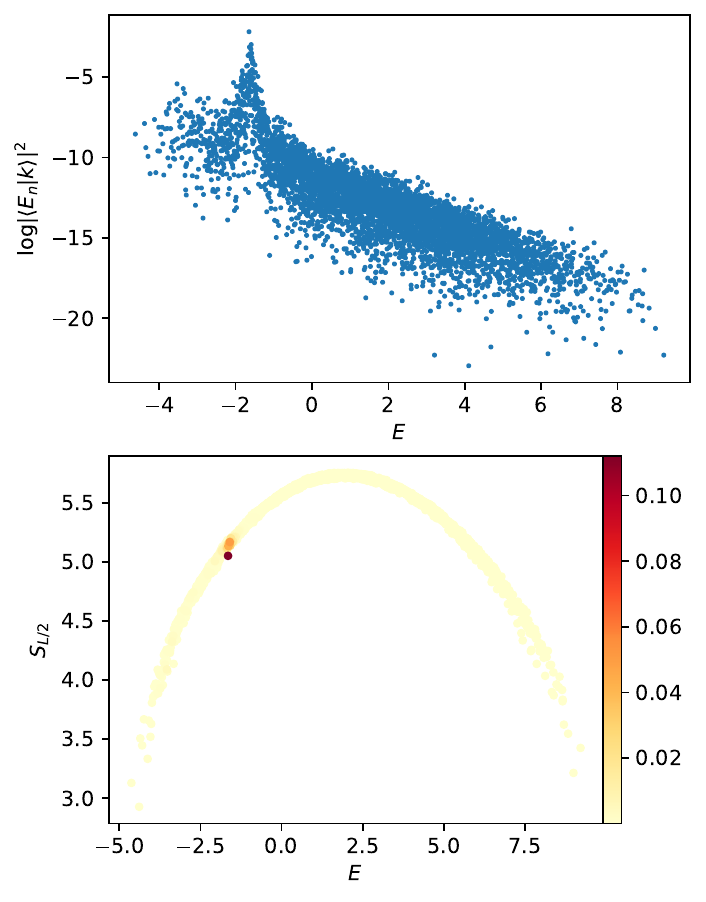}
	\caption{Exact diagonalization of Eq.~(\ref{eq:gen-aklt-ham}) for a finite system with $L=12$ in the $(N=8,k=\pi+2\pi/L)$ sector. Top: Overlap between the initial state and each eigenstate. Bottom: Scatter plot of eigenenergy versus energy, with a color map indicating the overlaps with the z-quasi-Goldstone mode..}
	\label{fig:aklt-ent}
\end{figure}

In Fig.~\ref{fig:aklt-ent}, we show explicitly that Eq.~(\ref{eq:gen-aklt-ham}) in the 
\begin{equation*}
	k=\pi + \frac{2\pi}{L}
\end{equation*}
sector is indeed thermal and the z-quasi-Goldstone mode is the superposition of those thermal states.

\section{Deformed Quasi-Nambu-Goldstone Modes in Other Scar Models}
\label{apx:dqngmother}

In this Appendix, we delve deeper into examples showcasing deformed quasi-Goldstone modes within existing exact scar models. 
Specifically, we explore the Rydberg-antiblockaded scar model, the Onsager scar model, and an additional scar in the spin-1 XY model.
Note that the original Hamiltonians of these scar models exhibit reflection symmetry, resulting in zero velocity at $k=0$ or $k=\pi$. 
To yield more nontrivial results, we introduce a deformation term $\hat V$ to the Hamiltonian:
\begin{equation}
	\hat H' = \hat H + \hat V,
\end{equation}
where $\hat V$ is reflection anti-symmetric, preserves the scar space, and induces nonzero velocity at the $k=0,\pi$.

\subsection{Rydberg-antiblockaded model}

\begin{figure}
	\centering
	\includegraphics[width=0.9\linewidth]{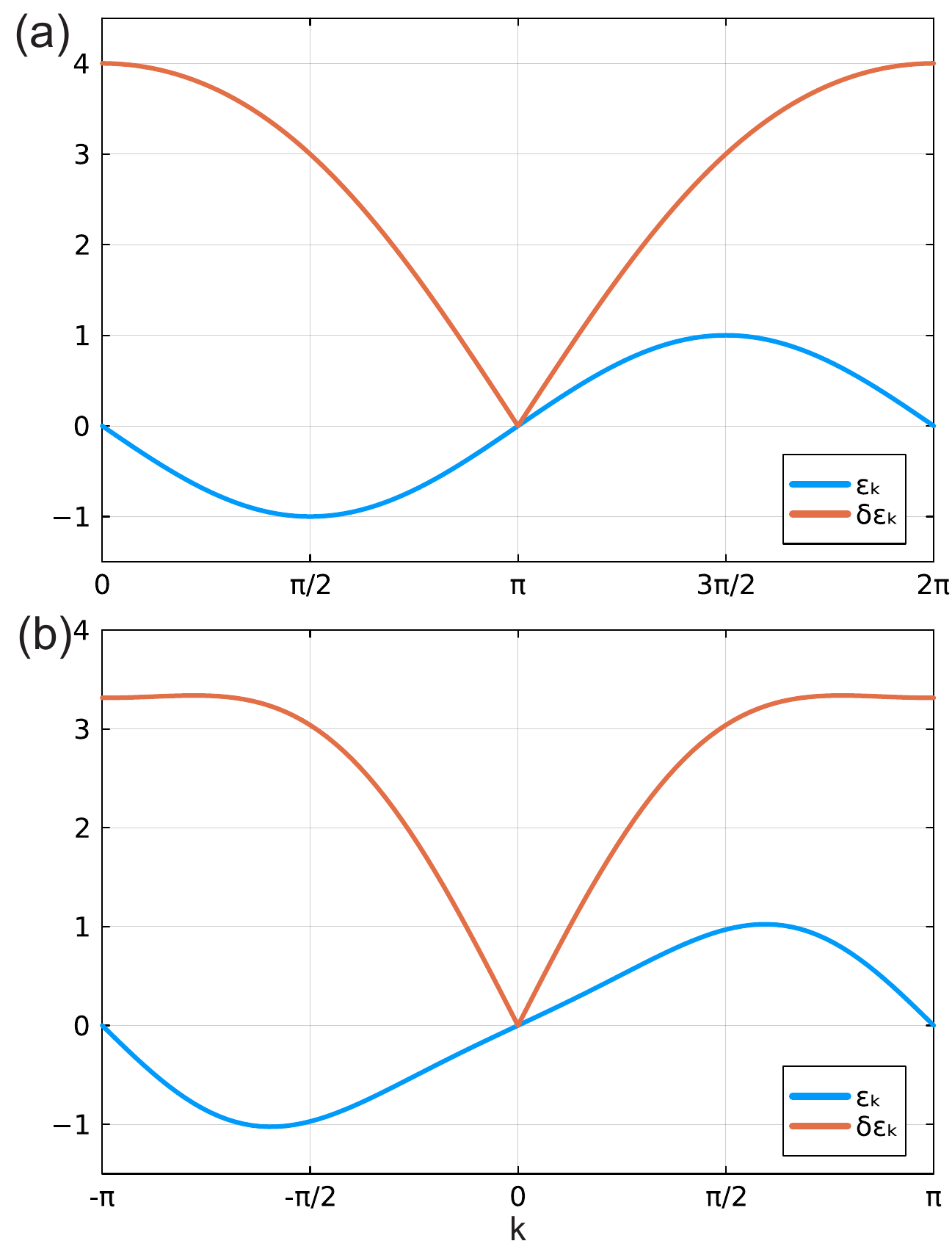}
	\caption{Energy expectation $\varepsilon_k$ and variance $\delta\varepsilon_k$ under deformed scar Hamiltonian Eq.~(\ref{eq:deformed-rydberg-ham}) with $J=1, h=0, g=2$. (a) $\varepsilon_k$ and and $\delta\varepsilon_k$ of the z-qNGM in Eq.~(\ref{eq:rydberg-zqngm}). (b) $\varepsilon_k$ and $\delta\varepsilon_k$ of the x-qNGM in Eq.~(\ref{eq:rydberg-xqngm}) with $\xi=1/2$.}
	\label{fig:rydberg}
\end{figure}

The model is introduced in Ref. \cite{domain-wall-0,domain-wall}, with the Hamiltonian:
\begin{equation}\label{eq:rydberg-scar}
	\hat H_\text{Ryd} = \sum_j \left[2J(\hat s_j^x - 4\hat s_{j-1}^z \hat s_j^x \hat s_{j+1}^z) + h \hat s_j^z \right]
\end{equation}
It was proved that there is an exact scar tower, starting from the spin-all-down state $|\Psi_0\rangle = \left|\downarrow \cdots \downarrow\right\rangle$, and generated by the ladder:
\begin{equation}
\begin{aligned}
	\hat Q^+ &= \sum_j (-1)^j\hat P_{j-1}^{\downarrow} \hat s_j^+ \hat P_{j+1}^{\downarrow}, \\
	|\Psi_n\rangle &= (\hat Q^+)^n|\Psi_e\rangle.
\end{aligned}
\end{equation}
In Ref.~\cite{dsymm}, the scar tower $\{|\Psi_n\rangle\}$ can be regarded as a deformed SU(2)-symmetric space deformed by an MPO with tensor:
\begin{equation}
	W = \begin{bmatrix}
        0 & \hat s^+ \\
        \hat s^- & \hat P^{\downarrow}
    \end{bmatrix}.
\end{equation}
Since the scar tower has a ladder operator, the z-qNGM is simply defined as:
\begin{equation}\label{eq:rydberg-zqngm}
\begin{aligned}
	|k,z\rangle &= \hat{W}\left[ \hat S_k^+ \left(\includegraphics[width=0.07\textwidth,valign=c]{pics/spin-down}\otimes  \cdots \otimes \includegraphics[width=0.07\textwidth,valign=c]{pics/spin-down} \right)\right] \\
	&= \sum_j e^{ikj} \hat P_{j-1}^{\downarrow} \hat s_j^+ \hat P_{j+1}^{\downarrow}|\Psi_0\rangle \\
	&= \sum_j e^{ikj} \hat s_j^+ |\Psi_0\rangle.
\end{aligned}
\end{equation}
We consider the deformed Hamiltonian:
\begin{equation}\label{eq:deformed-rydberg-ham}
\begin{aligned}
	\hat H &= \hat H_\text{Ryd} + g \hat V, \\
	\hat V &= \sum_j \left(\hat s_{j-1}^x \hat s_j^y \hat s_{j+1}^z -
	\hat s_{j-1}^z \hat s_j^y \hat s_{j+1}^x\right).
\end{aligned}
\end{equation}
Fig.~\ref{fig:rydberg}(a) shows the energy expectation and variance of the z-qNGM.

For the scar initial state being the following matrix-product state:
\begin{equation}
\begin{aligned}
	|\psi_\xi\rangle 
	&= \hat{W}\left[\includegraphics[width=0.07\textwidth,valign=c]{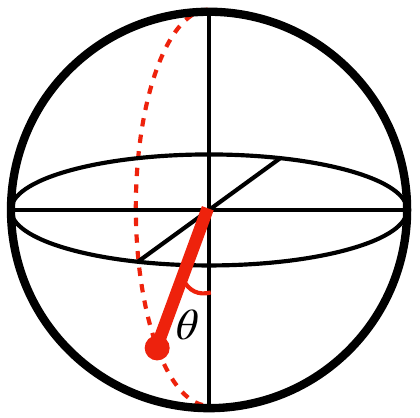} \otimes \includegraphics[width=0.07\textwidth,valign=c]{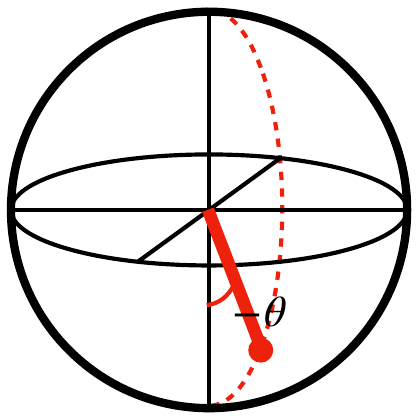} \otimes  \cdots\right]\\
	&= |[ABAB\cdots AB]\rangle.
\end{aligned}
\end{equation}
Here the tensors $A,B$ are:
\begin{equation}
\begin{aligned}
	A &= W \cdot (\left|\downarrow\right\rangle + \xi\left|\uparrow\right\rangle)
	= \begin{bmatrix}
        0 & \left|\uparrow\right\rangle  \\
        \xi\left|\downarrow\right\rangle  & \left|\downarrow\right\rangle
    \end{bmatrix}, \\
    B &= W \cdot (\left|\downarrow\right\rangle - \xi\left|\uparrow\right\rangle)
    = \begin{bmatrix}
        0 & \left|\uparrow\right\rangle  \\
        -\xi\left|\downarrow\right\rangle  & \left|\downarrow\right\rangle 
    \end{bmatrix}.
\end{aligned}
\end{equation}
where $\xi$ is related to the angle $\theta$:
\begin{equation}
	\xi = \tan(\theta/2)
\end{equation}
The energy expectation of the x-qNGM depends on $\xi$, in this specific model, when choosing $\xi=1$, the velocity at $k=0$ is zero.
For this reason we instead choose $\xi=1/2$ so that the x-qNGM has nonzero velocity. 

The x-qNGM is then defined as:
\begin{equation}\label{eq:rydberg-xqngm}
	|k,x\rangle = \hat S_k^z|\psi_\xi\rangle=\sum_j e^{ikj}\hat s^+_j |\psi_\xi\rangle.
\end{equation}
In Fig.~\ref{fig:rydberg}(b), we display the energy expectation and variance of the x-qNGM.

\subsection{Spin-1/2 Onsager scar model}

\begin{figure}
	\centering
	\includegraphics[width=0.9\linewidth]{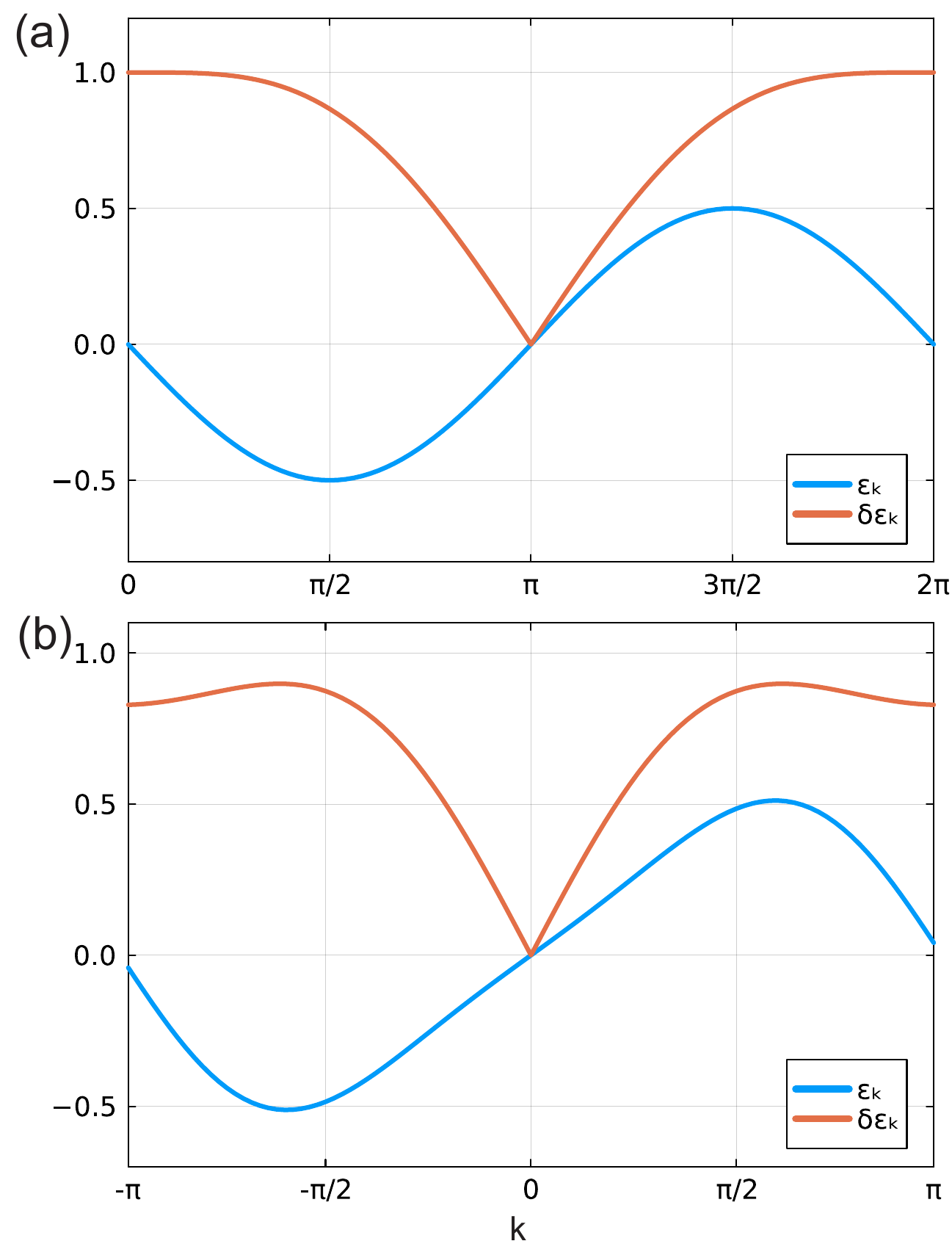}
	\caption{Energy expectation $\varepsilon_k$ and variance $\delta\varepsilon_k$ under Hamiltonian Eq.~(\ref{eq:onsager-scar}) with $n=2, h=0, \lambda=1/2$. (a) $\varepsilon_k$ and and $\delta\varepsilon_k$ of the z-qNGM in Eq.~(\ref{eq:onsager-zqngm}). (b) $\varepsilon_k$ and $\delta\varepsilon_k$ of the x-qNGM in Eq.~(\ref{eq:onsager-xqngm}) with $\xi=1/2$.}
	\label{fig:onsager}
\end{figure}

In Ref.~\cite{Onsager}, the following scar Hamiltonian was proposed
\begin{equation}\label{eq:onsager-scar}
	\hat H_\text{Ons} = \hat H_n + \lambda \hat H_\text{pert} + h \hat S^z.
\end{equation}
Here, $\hat H_n$ is defined as
\begin{equation}
\begin{aligned}
	\hat H_n =&\ -\sum_{j=1}^L \sum_{a=1}^{n-1} \frac{1}{4 \sin (\pi a / n)}\left[n(-1)^a\left(\hat s_j^{-} \hat s_{j+1}^{+}\right)^a+\text {H.c.} \right. \\
	& \left. +(n-2 a) \omega^{a / 2} \hat\tau_j^a\right],
\end{aligned}
\end{equation}
where $\omega = \exp(2\pi i/n)$, and the operator $\hat\tau = \operatorname{diag}(1,\omega,\dots,\omega^{n-1})$.
We only consider the simplest $n=2$ case, where $\hat H_n$ is reduced to the spin-1/2 XY model:
\begin{equation}
	\hat H_2=\sum_{j=1}^L \hat s_j^x \hat s_{j+1}^x + \hat s_j^y \hat s_{j+1}^y.
\end{equation}
For the $n=2$ case, the ladder operator 
\begin{equation}
	\hat Q^+ = \sum_j (-1)^j \hat s_j^+\hat s_{j+1}^+.
\end{equation}
In Ref.~\cite{dsymm}, this scar tower has been regarded as the deformed-SU(2) symmetric space.
The MPO tensor of the deforming operator is:
\begin{equation}
	W = \begin{bmatrix} 
        \left|\downarrow\right\rangle \left\langle \downarrow \right| &
        \left|\uparrow\right\rangle \left\langle \downarrow \right| \\
        \left|\uparrow\right\rangle \left\langle \uparrow \right| & 0
    \end{bmatrix}
\end{equation}
Also, here we choose the $\hat H_\text{pert}$ as
\begin{equation}
	\hat H_\text{pert}
	= \sum_j \left(\hat s_{j-1}^x \hat s_{j+1}^y - \hat s_{j-1}^y \hat s_{j+1}^x \right).
\end{equation}
The z-qNGM is
\begin{equation}\label{eq:onsager-zqngm}
\begin{aligned}
	|k,z\rangle &= \hat{W}\left[ \hat S_k^+ \left(\includegraphics[width=0.07\textwidth,valign=c]{pics/spin-down}\otimes  \cdots \otimes \includegraphics[width=0.07\textwidth,valign=c]{pics/spin-down} \right)\right] \\
	&= \hat Q^+_k|\Psi_0\rangle \\
	&= \sum_j e^{ikj} \hat s_j^+ \hat s_{j+1}^+|\Psi_0\rangle.
\end{aligned}
\end{equation}
In Fig.~\ref{fig:onsager}(a), we show the energy expectation and variance of the z-qNGM.

For scar initial state being the matrix-product state
\begin{equation}
\begin{aligned}
	|\psi_\xi\rangle 
	&= \hat{W}\left[\includegraphics[width=0.07\textwidth,valign=c]{pics/spin_theta} \otimes \includegraphics[width=0.07\textwidth,valign=c]{pics/spin_mtheta} \otimes  \cdots\right]\\
	&= |[ABAB\cdots AB]\rangle,
\end{aligned}
\end{equation}
where the tensors $A,B$ are:
\begin{equation}
\begin{aligned}
	A &= W\cdot (\left|\downarrow\right\rangle + \xi\left|\uparrow\right\rangle)
	= \begin{bmatrix} 
        \left|\downarrow\right\rangle & \left|\uparrow\right\rangle \\
        \xi \left|\uparrow\right\rangle & 0
    \end{bmatrix}, \\
    B &= W\cdot (\left|\downarrow\right\rangle - \xi\left|\uparrow\right\rangle)
    = \begin{bmatrix} 
        \left|\downarrow\right\rangle & \left|\uparrow\right\rangle \\
        -\xi\left|\uparrow\right\rangle & 0
    \end{bmatrix}.
\end{aligned}
\end{equation}
The x-qNGM is then defined as:
\begin{equation}\label{eq:onsager-xqngm}
	|k,x\rangle = \hat S_k^z|\psi_\xi\rangle=\sum_j e^{ikj}\hat s^z_j |\psi_\xi\rangle.
\end{equation}
In Fig.~\ref{fig:onsager}(b), we display the energy expectation and variance of the x-qNGM.

\subsection{Type-II scar tower in spin-1 XY model}

\begin{figure}
	\centering
	\includegraphics[width=0.9\linewidth]{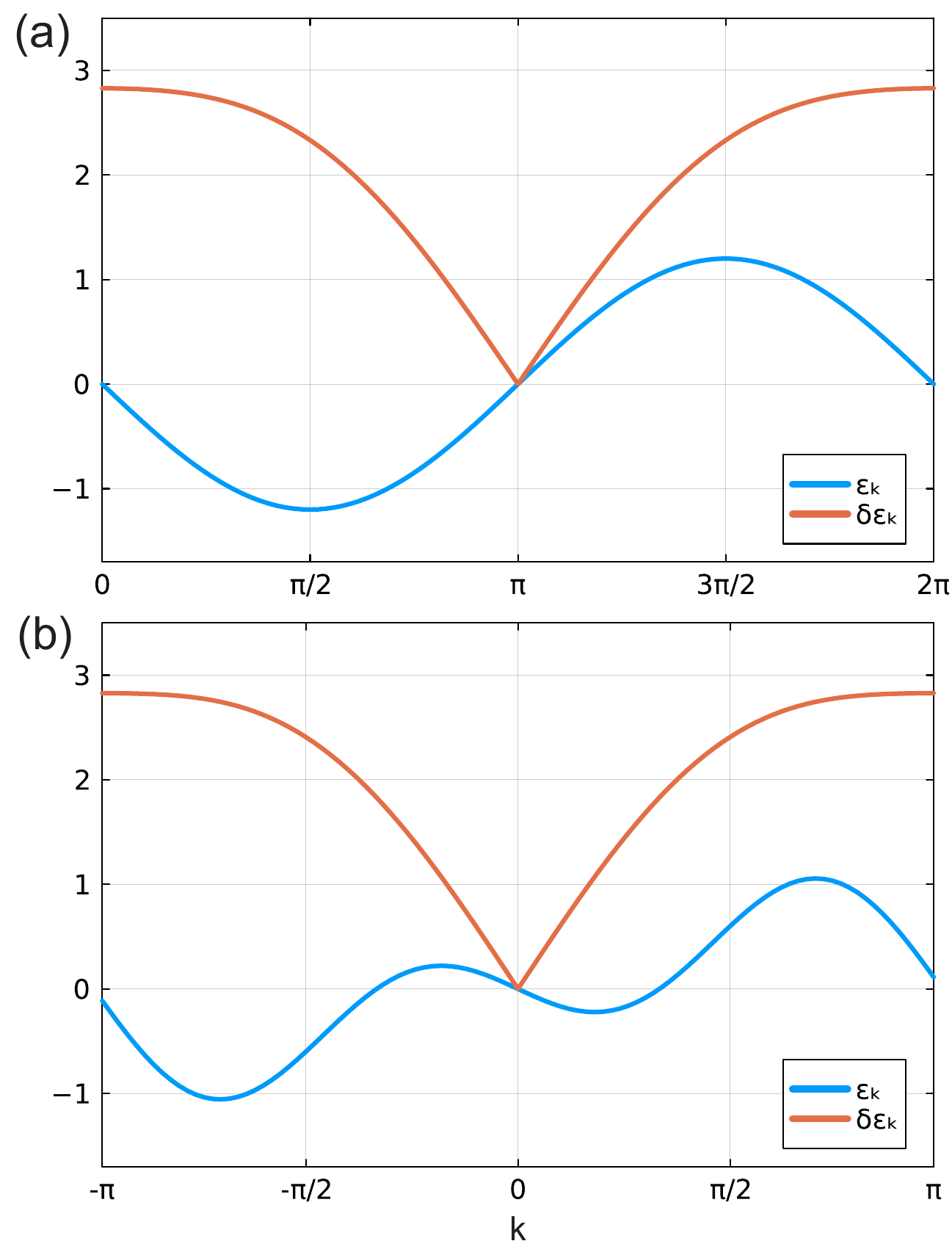}
	\caption{Energy expectation $\varepsilon_k$ and variance $\delta\varepsilon_k$ under Hamiltonian Eq.~(\ref{eq:deformed_xy_ham}) with $h=0,\epsilon=0.1, \lambda=0.3$. (a) $\varepsilon_k$ and and $\delta\varepsilon_k$ of the z-qNGM in Eq.~(\ref{eq:xy-2-zqngm}). (b) $\varepsilon_k$ and $\delta\varepsilon_k$ of the x-qNGM in Eq.~(\ref{eq:xy-2-xqngm}).}
	\label{fig:XY-II}
\end{figure}

In Ref.~\cite{XY-1}, it was discovered that in spin-1 XY model, besides the original scar tower, another special scar tower emerges when the Hamiltonian is given by:
\begin{equation}
\begin{aligned}
	\hat H_\text{XY} =&\ \sum_j \left[
		\hat s_j^x \hat s_{j+1}^x + \hat s_j^y \hat s_{j+1}^y + h\hat s_j^z \right. \\
		& \left.+
		\epsilon (\hat s_j^+)^2(\hat s_{j+1}^-)^2 + \epsilon (\hat s_j^-)^2(\hat s_{j+1}^+)^2
	\right].
\end{aligned}
\end{equation}
This additional scar tower originates from $|\Psi_0\rangle = |-\cdots-\rangle$, and is expressed as:
\begin{equation}
	\left|\Psi_n\right\rangle \propto \sum_{i_1 \neq i_2 \neq \cdots \neq i_n}(-1)^{\sum_{j=1}^n i_j} \bigotimes_{j=1}^n\left(\hat s_{i_j}^{+} \hat s_{i_j+1}^{+}\right)\left|\Psi_0\right\rangle.
\end{equation}
It is important to note that this scar tower lacks a local ladder operator.
In Ref.~\cite{XY-2}, and later in Ref.~\cite{dsymm}, this tower was interpreted as a projected SU(2)-symmetric space.
The projector can be represented by an MPO, described by the tensor:
\begin{equation}
	W = \begin{bmatrix}
        |+\rangle \left\langle \uparrow \right| & |0\rangle \left\langle \downarrow \right| \\
        |0\rangle \left\langle \uparrow \right| & |-\rangle \left\langle \downarrow \right|
    \end{bmatrix}.
\end{equation}
For the z-qNGM, it is defined as:
\begin{equation}\label{eq:xy-2-zqngm}
\begin{aligned}
	|k,z\rangle &= \hat{W}\left[ \hat S_k^+ \left(\includegraphics[width=0.07\textwidth,valign=c]{pics/spin-down}\otimes  \cdots \otimes \includegraphics[width=0.07\textwidth,valign=c]{pics/spin-down} \right)\right] \\
	&= \sum_j e^{ikj} \hat s_j^+ \hat s_{j+1}^+ |\Psi_0\rangle.
\end{aligned}
\end{equation}
For the x-qNGM, the definition is similar as before:
\begin{equation}\label{eq:xy-2-xqngm}
\begin{aligned}
	|\psi_{\xi=1}\rangle 
	&= \hat{W}\left[\includegraphics[width=0.07\textwidth,valign=c]{pics/spin-front} \otimes \includegraphics[width=0.07\textwidth,valign=c]{pics/spin-back} \otimes  \cdots\right]\\
	&= |[ABAB\cdots AB]\rangle,
\end{aligned}
\end{equation}
where $|\Psi_\xi\rangle = |[AB\cdots AB]\rangle$ has the MPS representation:
\begin{equation}
\begin{aligned}
	A &= W\cdot (\left|\downarrow\right\rangle + \xi\left|\uparrow\right\rangle)
	= \begin{bmatrix}
        |+\rangle & |0\rangle \\
        |0\rangle & |-\rangle
    \end{bmatrix}, \\
    B &= W\cdot (\left|\downarrow\right\rangle - \xi\left|\uparrow\right\rangle)
    = \begin{bmatrix}
        -|+\rangle & |0\rangle \\
        -|0\rangle & |-\rangle
    \end{bmatrix}.
\end{aligned}
\end{equation}
The energy expectation and variance are displayed in Fig.~\ref{fig:XY-II}.
For this additional tower, we introduce the deformation:
\begin{equation}
	\hat H = \hat H_\text{XY} + \lambda \hat V,
\end{equation}
where the perturbation term is defined as:
\begin{equation}\label{eq:deformed_xy_ham}
\begin{aligned}
	\hat V =&\ \sum_j \left[\hat\lambda^{(3)}_j\otimes \hat\lambda^{[1,2]}_{j-1,j+1} + \hat b_j \otimes \hat\lambda^{[6,7]}_{j-1,j+1} \right. \\
    &+ \left. \hat s^z_j \otimes \left(\hat\lambda^{[1,7]}_{j-1,j+1}-\hat\lambda^{[2,6]}_{j-1,j+1}\right)\right]
\end{aligned}
\end{equation}
using the notation in Eq.~(\ref{eq:lambda-def}).

\section{Quasi-Nambu-Goldstone Modes in SU(3) Scar Model}

Now we consider the scar tower with SU(3) quasisymmetry, constructed in Ref.~\cite{qsymm-1}.
This is a spin-1 chain hosting a maximal weight representation of SU(3).

\subsection{SU(3) scar model}

In the SU(3) scar model, the anchor state is fixed to the polarized state $|\Psi_e\rangle = \left|-\cdots-\right\rangle$.
The SU(3)-symmetric space, according to the definition, is
\begin{equation}
	\mathcal{H}_\mathrm{SU(3)}
	\equiv \operatorname{span}\{\hat u(g)^{\otimes N}|\Psi_e\rangle|\forall g\in \mathrm{SU(3)}\}.
\end{equation}
From the representation theory of Lie algebra, the space $\mathcal{H}_\mathrm{SU(3)}$ is the maximal weight irreducible representation.
An orthonormal basis for $\mathcal{H}_\mathrm{SU(3)}$ can be generated by acting the following ladder operator
\begin{equation}
\begin{aligned}
	\hat Q_1^+ &= \sum_j (\hat q_1^+)_j = \begin{pmatrix}
		0 & & \\
		 & 0 & 1 \\
		 &  & 0
	\end{pmatrix}_j, \\
	\hat Q_2^+ &= \sum_j (\hat q_2^+)_j = \sum_j \begin{pmatrix}
		0 & & 1 \\
		 & 0 & \\
		 &  & 0
	\end{pmatrix}_j.
\end{aligned}
\end{equation}
to the anchor state $|\Psi_e\rangle$.
That is, 
\begin{equation}
	\mathcal{H}_\mathrm{SU(3)}
	= \operatorname{span}\{(\hat Q_1^+)^n (\hat Q_2^+)^m|\Psi_e\rangle|\forall n,m\in \mathrm{SU(3)}\}.
\end{equation}
With the scar space defined, we can use the quantum inverse method to find all the possible translational invariant nearest-neighbor scar Hamiltonians.

\label{apx:qngmsu3}
\begin{figure}
	\centering
	\includegraphics[width=0.99\linewidth]{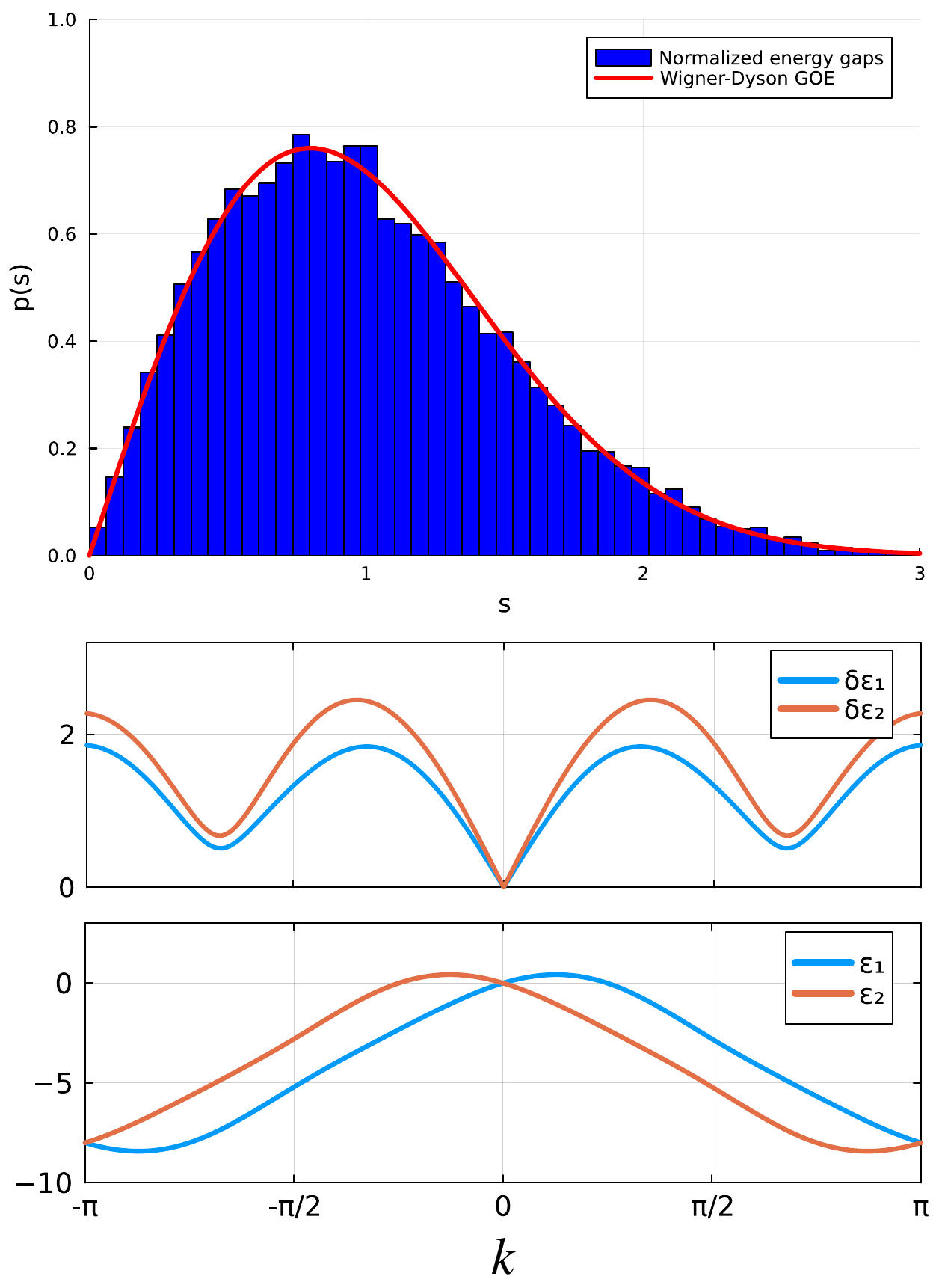}
	\caption{Top: Spectrum of the $L=12$ SU(3)-scar model with Hilber space dimension $\dim \mathcal H = 52488$. The model parameters are fixed to be $J_1=1,\ J_2=-0.6,\ J_3=0.3,\ J_4=0.2,\ J_5=0.1,\ J_6=0.3$. Statistics of energy gaps in the middle half of the ($k=0$) sector show the Wigner-Dyson GOE form, indicating non-integrability. Middle:  energy variance of two modes in SU(3) scar model. Bottom: energy expectation of two modes.}
	\label{fig:su3-levels}
\end{figure}

There are in total 9 solutions from the quantum inverse method:
\begingroup
\allowdisplaybreaks
\begin{align*}
	\hat o^{(1)}_{j,j+1} &= \sum_{a=1}^8 \hat \lambda_{j,j+1}^{(aa)} , \\
	\hat o^{(2)}_{j,j+1} &= \hat \lambda_{j,j+1}^{[12]}+\hat \lambda_{j,j+1}^{[45]}, \\
	\hat o^{(3)}_{j,j+1} &= \hat \lambda_{j,j+1}^{[45]}+\hat \lambda_{j,j+1}^{[67]}, \\
	\hat o^{(4)}_{j,j+1} &= 2\hat \lambda_{j,j+1}^{[13]}-\hat \lambda_{j,j+1}^{[46]}-\hat \lambda_{j,j+1}^{[57]}, \\
	\hat o^{(5)}_{j,j+1} &= 2\hat \lambda_{j,j+1}^{[23]}+\hat \lambda_{j,j+1}^{[47]}-\hat \lambda_{j,j+1}^{[56]}, \\
	\hat o^{(6)}_{j,j+1} &= \hat \lambda_{j,j+1}^{[14]}+\hat \lambda_{j,j+1}^{[25]}-\hat \lambda_{j,j+1}^{[36]}-\sqrt{3}\hat \lambda_{j,j+1}^{[48]}, \\
	\hat o^{(7)}_{j,j+1} &= \hat \lambda_{j,j+1}^{[15]}-\hat \lambda_{j,j+1}^{[24]}-\hat \lambda_{j,j+1}^{[37]}-\sqrt{3}\hat \lambda_{j,j+1}^{[78]}, \\
	\hat o^{(8)}_{j,j+1} &= \hat \lambda_{j,j+1}^{[16]}-\hat \lambda_{j,j+1}^{[27]}+\hat \lambda_{j,j+1}^{[34]}-\sqrt{3}\hat \lambda_{j,j+1}^{[48]}, \\
	\hat o^{(9)}_{j,j+1} &= \hat \lambda_{j,j+1}^{[17]}+\hat \lambda_{j,j+1}^{[26]}+\hat \lambda_{j,j+1}^{[35]}-\sqrt{3}\hat \lambda_{j,j+1}^{[58]}.
\end{align*}
\endgroup
Note that one of them is Heisenberg-like (SU(3) singlet), and 8 are DM-like (SU(3)-octet). 
Besides, these interactions are not limited to nearest neighbor.
That is, $\sum_j \hat o_{j,j+d}^{(n)}$ for any $d \ge 1$ is a parent Hamiltonian for the scar tower.
Choose the Hamiltonian form:
\begin{equation}\label{eq:su(3)-ham}
\begin{aligned}
	\hat H =&\ \sum_j \left(
	J_1 \hat o^{(1)}_{j,j+1} + J_2 \hat o^{(2)}_{j,j+1} + J_3 \hat o^{(3)}_{j,j+1} \right. \\
	&\left. + J_4 \hat o^{(4)}_{j,j+2} + J_5 \hat o^{(6)}_{j,j+2} + J_6 \hat o^{(6)}_{j,j+3}
	\right),
\end{aligned}
\end{equation}
where the local terms of $\hat O^{(4)}$ and $\hat O^{(6)}$ act on next-nearest neighbor sites. 
We choose 
\begin{equation}
\begin{aligned}
	J_1&=1,\quad
	J_2=-0.6,\quad 
	J_3=0.3,\\ 
	J_4&=0.2,\quad 
	J_5=0.1, \quad
	J_6=0.3.
\end{aligned}
\end{equation}
The finite-size exact diagonalization shows clear Wigner-Dyson level statistics, see Fig.~\ref{fig:su3-levels}a.

\subsection{Multiple quasi-Nambu-Goldstone modes}
The nature definitions of quasi-Goldstone modes in this case are
\begin{equation}
\begin{aligned}
	|k,1\rangle &= \sum_j e^{ikj} (\hat{q}_1^+)_j |\Psi_e\rangle, \\
	|k,2\rangle &= \sum_j e^{ikj} (\hat{q}_2^+)_j |\Psi_e\rangle.
\end{aligned}
\end{equation}
Note that the transition amplitude for the two modes are
\begin{equation}
	\langle k,a|H|k,b\rangle
	= \begin{bmatrix}\varepsilon_1 & v \\ v^* & \varepsilon_2\end{bmatrix}_{ab}, 
\end{equation}
where
\begin{equation*}
	\begin{aligned}
		\varepsilon_1 &= 4 J_1 [\cos (k)-1]-4 J_3 \sin (k), \\
		\varepsilon_2 &= 4 J_1 [\cos (k)-1]-4 \left(J_2+J_3\right) \sin (k), \\
		v &= 4 i \sin(2 k) J_4.
	\end{aligned}
\end{equation*}
Because of the transition amplitude, we choose to define the qNGMs as the linear combinations of the $|k,1\rangle$ and $|k,2\rangle$, so that the two modes are decoupled, with energy expectation:
\begin{equation}
	\begin{aligned}
\varepsilon_\pm(k) =&\ 4 J_1 [\cos (k)-1]-\left(2J_2+4 J_3\right) \sin (k) \\
&\pm 2 \sin(k)\sqrt{J_2^2+16 J_4^2 \cos ^2(k)} \\
=&\ -2\left[J_2+2J_3\pm\sqrt{J_2^2+16J_4^2}\right]k + O\left(k^2\right).
\end{aligned}
\end{equation}
Because of the linear combination, the exact expression of the energy variance for the mode becomes complicated.
In Fig.~\ref{fig:su3-levels}b, we plot the energy expectation and energy variance for each independent mode.

\section{Quasi-Nambu-Goldstone Modes in Deformed PXP Model}
\label{apx:qngmpxp}

In this section, we explore the PXP model~\cite{pxp-1,pxp-2}:
\begin{equation}\label{eq:pxp_ham}
	\hat H = \sum_j \hat P^{\downarrow}_{j-1} \hat \sigma_j^x \hat P^{\downarrow}_{j-1},
\end{equation}
which serves as the effective model for the Rydberg atom array experiment~\cite{lukin2017}. 
This model is notable for featuring an approximate scar tower without an exact symmetry structure. 
Despite the absence of precise symmetry, the PXP model exhibits an approximate decoupling 
\begin{equation}\label{eq:pxp-decouple}
	\hat H_\text{PXP} \approx \hat H_\text{FSA} \oplus \hat{H}_\text{th},
\end{equation}
where a special subspace is generated by the forward scattering approximation (FSA)~\cite{pxp-1}.

In the FSA, the Hamiltonian $\hat{H}$ Eq.~(\ref{eq:pxp_ham}) is decomposed into $\hat{H} = \hat{H}^+ + \hat{H}^-$, where 
\begin{equation*}
\begin{aligned}
	\hat H^- &= \sum_{j\in \text{odd}} \hat P^{\downarrow}_{j-1} \hat \sigma_j^+ \hat P^{\downarrow}_{j-1}  +\sum_{j\in\text{even}} \hat P^{\downarrow}_{j-1} \hat \sigma_j^- \hat P^{\downarrow}_{j-1}, \\
	\hat H^+ &= \sum_{j\in \text{odd}} \hat P^{\downarrow}_{j-1} \hat \sigma_j^- \hat P^{\downarrow}_{j-1} +\sum_{j\in\text{even}} \hat P^{\downarrow}_{j-1} \hat \sigma_j^+ \hat P^{\downarrow}_{j-1}.
\end{aligned}
\end{equation*}
Within this approximation, the quasi-periodic dynamics originating from the initial Neel state 
\begin{equation*}
	|\mathbb Z_2\rangle \equiv \left|\uparrow\downarrow\cdots\uparrow\downarrow\right\rangle.
\end{equation*}
is primarily contained within the FSA subspace:
\begin{equation}
	\mathcal{H}_\text{FSA} = \operatorname{span}\left\{ \left.(\hat H^+)^n|\mathbb Z_2\rangle \right| \forall n \in \mathbb Z^+\right\}.
\end{equation}
This subspace is spanned by states generated through the application of $\hat{H}^+$ on the initial $|\mathbb{Z}_2\rangle$ state, eventually reaching another Neel state:
\begin{equation*}
	|\mathbb Z_2'\rangle = \left|\downarrow\uparrow\cdots\downarrow\uparrow \right\rangle.
\end{equation*}

\subsection{Approximate SU(2) and the corresponding quasi-Nambu-Goldstone modes}

\begin{figure}
	\centering
	\includegraphics[width=\linewidth]{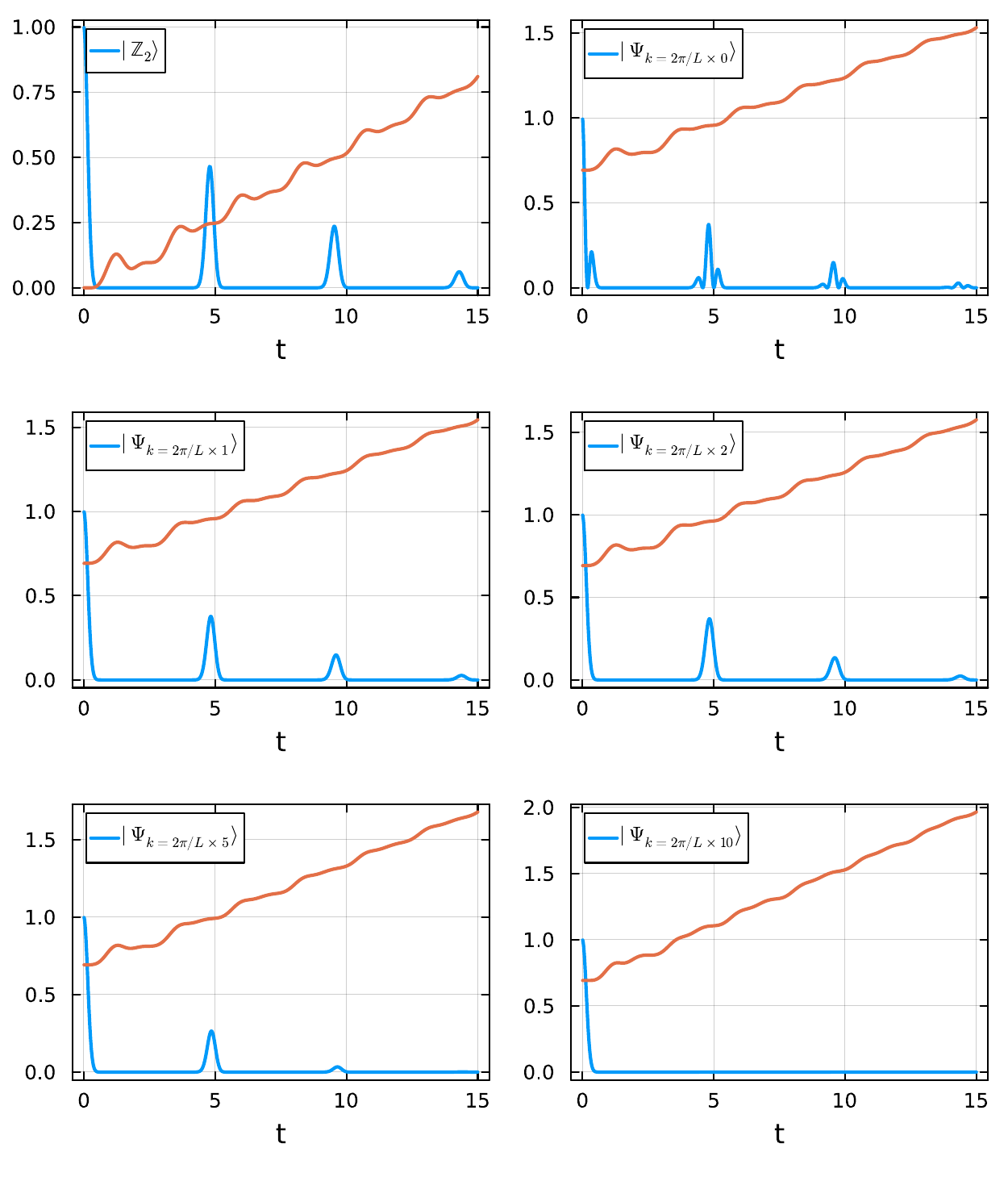}
	\caption{Scar dynamics (for system size $L=50$) initiated from $|\mathbb{Z}_2\rangle$ and $|\Psi_{k=2\pi n/L}\rangle$'s, where $n\in\{0,1,2,5,10\}$. The evolution is based on MPS TDVP simulation with bond dimension $\chi =60$. The blue lines represent the evolution of fidelity $F(t) = |\langle \psi|\psi(t)\rangle|^2$, and the orange line represents the evolution of the bipartite entanglement entropy.}
	\label{fig:pxp0-dynamics}
\end{figure}

The operators $\hat{H}^+$ and $\hat{H}^-$, along with their commutator $[\hat{H}^+, \hat{H}^-]$, form an approximate SU(2) algebra:
\begin{equation}\label{eq:pxp-su2}
	\hat H^\pm = \frac{\Delta}{2} \hat S^\pm, \quad
	\hat S^z = \frac{1}{2}[\hat S^+,\hat S^-] = \frac{1}{\Delta}[\hat H^+,\hat H^-].
\end{equation} 
In this context, the scar initial state $|\mathbb{Z}_2\rangle$ acts as the ``bottom state'' in the $s = L/2$ representation of this approximate SU(2) symmetry, which is evidenced by the relation $\hat{H}^- |\mathbb{Z}_2\rangle = 0$. 
Within this subspace, the Hamiltonian simplifies:
\begin{equation}
	\hat H = \hat H^+ + \hat H^- \simeq \Delta \hat S^x,
\end{equation}
indicating scar dynamics akin to the rotation of a large spin with a period $T = 2\pi/\Delta$. 
This approximate symmetry allows for other scar states, such as:
\begin{equation}
	|\Psi_0\rangle \propto \hat H^+|\mathbb Z_2\rangle,
\end{equation}
which also evolves approximately within the SU(2) subspace. The approximate.
The approximate qNGM state is similarly defined by introducing small spatial modulation:
\begin{equation}\label{eq:pxp-qngm}
\begin{aligned}
	|\Psi_k\rangle = &\ \sum_{j\in \text{odd}} e^{ikj}\hat P^{\downarrow}_{j-1} \hat \sigma_j^- \hat P^{\downarrow}_{j-1} |\mathbb Z_2\rangle \\
	&+\sum_{j\in\text{even}}e^{ikj} \hat P^{\downarrow}_{j-1} \hat \sigma_j^+ \hat P^{\downarrow}_{j-1}|\mathbb Z_2\rangle.
\end{aligned}
\end{equation}
For $|\mathbb Z_2\rangle$ state, the above expression has a simple form:
\begin{equation*}
	|\Psi_k\rangle \propto\sum_{j\in \text{odd}}e^{ikj} \hat \sigma_j^- |\mathbb Z_2\rangle.
\end{equation*}
A unique aspect of the PXP model is that there is no tuning parameter for the spacing of scar eigenstates, precluding the possibility of a degenerate limit. 
Consequently, the lifetime of the qNGM cannot be quantified by energy variance. 
Instead, we can demonstrate the long lifetime by directly simulating the evolution of the many-body fidelity and the growth of entanglement entropy.

To investigate the scar dynamics, we conducted numerical simulations starting from $|\Psi_k\rangle$ in a system of size $L = 50$, employing the time-dependent variational principle (TDVP) algorithm~\cite{tdvp-1,tdvp-2}.
Fig.~\ref{fig:pxp0-dynamics} illustrates the time evolution beginning from $|\mathbb{Z}2\rangle$ and $|\Psi{k=2\pi n/L}\rangle$. Notably, for $k = 2\pi/L$, the fidelity exhibits nearly identical evolution to that of $k = 0$, the latter being regarded as within the FSA subspace.

\subsection{Deformed PXP model and its quasi-Nambu-Goldstone modes}
\begin{figure}
	\centering
	\includegraphics[width=\linewidth]{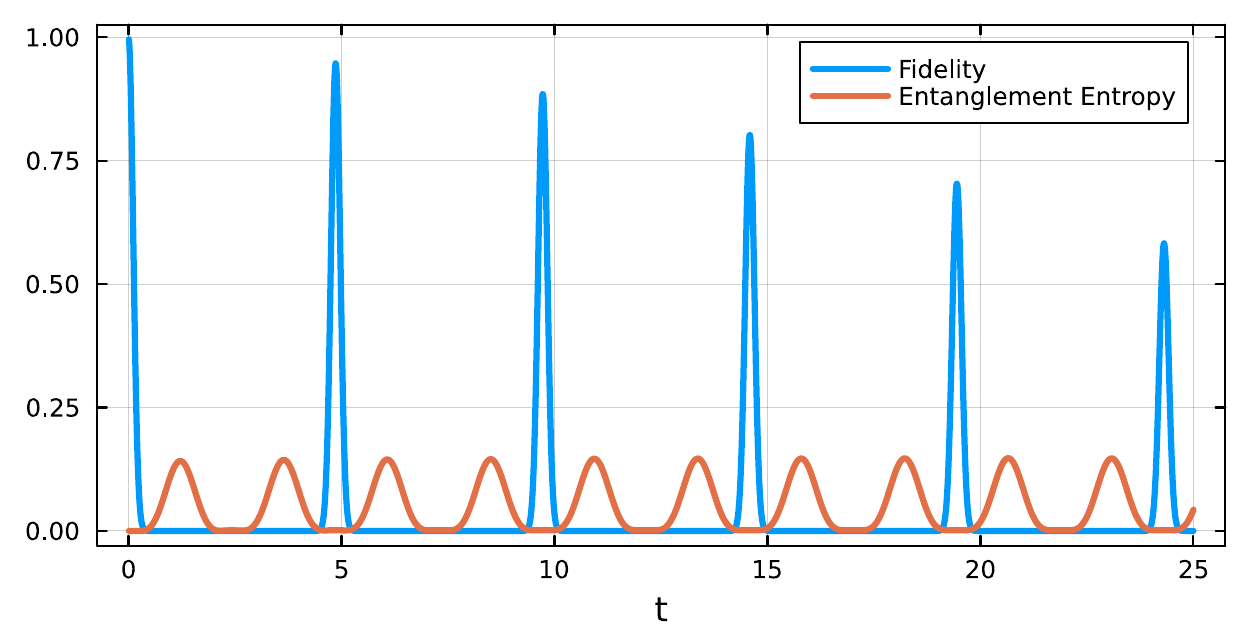}
	\caption{Dynamics of $|\mathbb Z_2\rangle$ under deformed PXP Hamiltonian Eq.~(\ref{eq:pxpd_ham}). The system size is $L=100$, and the simulation is based on the MPS TDVP algorithm with maximum bond dimension $\chi=30$. The many-body fidelity shows quasi-periodic and the entanglement entropy remains small during the simulation time, which allows us to choose a small bond dimension.}
	\label{fig:pxpd-z2}
\end{figure}

Although the forward scattering approximation provides partial insight into the revival phenomena, it remains a coarse approximation for the PXP model. 
Moreover, the PXP scar dynamics exhibit a relatively short lifetime, complicating the quantification of the qNGM lifetime. 
To address this, we employ a deformed PXP model~\cite{PXP-4}, introducing modifications to the original Hamiltonian to enhance the SU(2) algebra in Eq.~(\ref{eq:pxp-su2}):
\begin{equation}\label{eq:pxpd_ham}
	\hat H = \sum_j \hat P^{\downarrow}_{j-1} \hat \sigma_j^x \hat P^{\downarrow}_{j-1}\left[1-\sum_{d=2}^R h_d \left(\hat\sigma_{j-d}^z+\hat\sigma_{j+d}^z\right) \right].
\end{equation}
For clarity, we introduce the following notation:
\begin{equation}
	\hat K_j \equiv \hat{\mathbb I} - \sum_{d=2}^R h_d \left(\hat\sigma_{j-d}^z+\hat\sigma_{j+d}^z\right).
\end{equation}
By setting:
\begin{equation}
	h_d = \frac{h_0}{\left[\phi^{d-1}-\phi^{-(d-1)}\right]^2},
\end{equation}
where $\phi = (\sqrt{5}+1)/2$ and $h_0$ is optimized numerically to approximately $0.051$ \cite{PXP-4}, the decoupling in Eq.~(\ref{eq:pxp-decouple}) is significantly improved, though remaining approximate. 
Subsequent analysis is based on this deformed PXP model with these specified parameters.

\begin{figure}
	\centering
	\includegraphics[width=\linewidth]{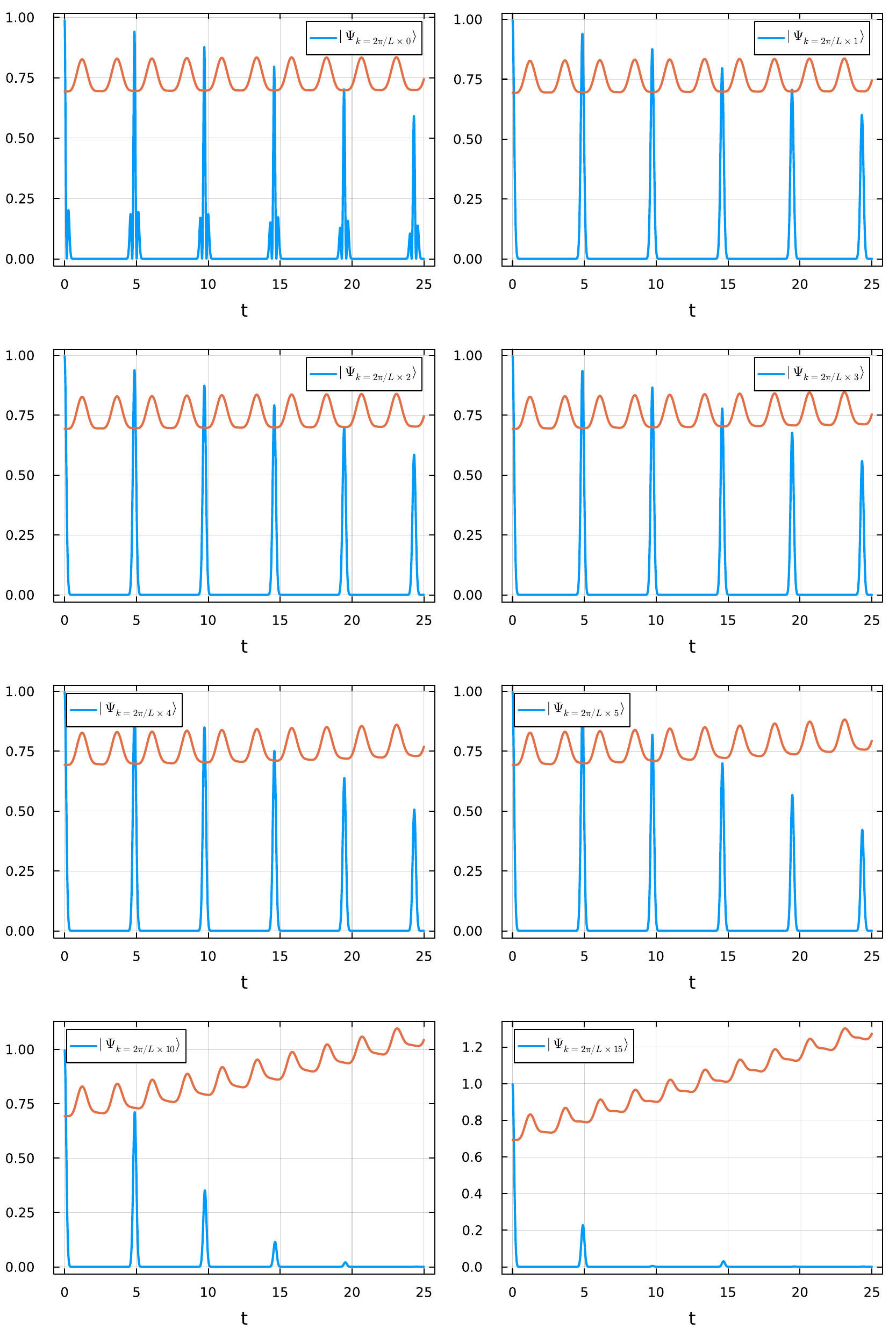}
	\caption{Scar dynamics initiated from $|\mathbb{Z}_2\rangle$ and $|\Psi_{k=2\pi n/L}\rangle$'s, where $n\in\{0,1,2,3,4,5,10,15\}$. The blue lines represent the evolution of fidelity $F(t) = |\langle \psi|\psi(t)\rangle|^2$, and the orange line represents the evolution of the bipartite entanglement entropy.}
	\label{fig:pxp-dynamics}
\end{figure}

The deformed PXP model can also be expressed as $\hat{H} = \hat{H}^- + \hat{H}^+$, where
\begin{equation*}
\begin{aligned}
	\hat H^- &= \sum_{j\in \text{odd}} \hat P^{\downarrow}_{j-1} \hat \sigma_j^+ \hat P^{\downarrow}_{j-1}  \hat K_j +\sum_{j\in\text{even}} \hat P^{\downarrow}_{j-1} \hat \sigma_j^- \hat P^{\downarrow}_{j-1} \hat K_j, \\
	\hat H^+ &= \sum_{j\in \text{odd}} \hat P^{\downarrow}_{j-1} \hat \sigma_j^- \hat P^{\downarrow}_{j-1}  \hat K_j +\sum_{j\in\text{even}} \hat P^{\downarrow}_{j-1} \hat \sigma_j^+ \hat P^{\downarrow}_{j-1} \hat K_j.
\end{aligned}
\end{equation*}
With these parameters, in Eq.(\ref{eq:pxp-su2}), we find $\Delta \approx 1.29$, resulting in a period $T \approx 4.86$. 
Fig.~\ref{fig:pxpd-z2} presents a simulation of scar dynamics for a system of size $L=100$ starting from the $|\mathbb{Z}_2\rangle$ initial state. 
Notably, for $L=100$, any seemingly nonzero fidelity indicates that the state remains close to the original one. 
Another observation is the low entanglement entropy during the dynamics. 
In contrast to the original PXP model, as shown in the first panel of Fig.~\ref{fig:pxp0-dynamics}, where the entanglement entropy generally increases with minor dips, the deformed model shows prolonged periods of low entanglement entropy.

The qNGM is generated by the twisted ladder operator
\begin{equation}
	\hat H^+_k = \sum_{j\in \text{odd}/\text{even}} e^{ikj}\hat P^{\downarrow}_{j-1} \hat \sigma_j^{-/+} \hat P^{\downarrow}_{j-1}  \hat K_j.
\end{equation}
The additional $\hat{K}_j$ term does not affect the $|\mathbb{Z}_2\rangle$ state. 
Therefore, the qNGM remains the same as in the original PXP model:
\begin{equation*}
	|\Psi_k\rangle = \hat H^+_k = \sum_{j\in \text{odd}}e^{ikj} \hat \sigma_j^- |\mathbb Z_2\rangle.
\end{equation*}
Refer to Fig.~\ref{fig:pxp-dynamics}for the time evolution starting from $|\mathbb{Z}2\rangle$ and $|\Psi{k=2\pi n/L}\rangle$. 
Besides the quasi-oscillation of many-body fidelity, the bipartite entanglement entropy also displays approximate periodicity and remains low for extended durations.

\begin{figure}
	\centering
	\includegraphics[width=\linewidth]{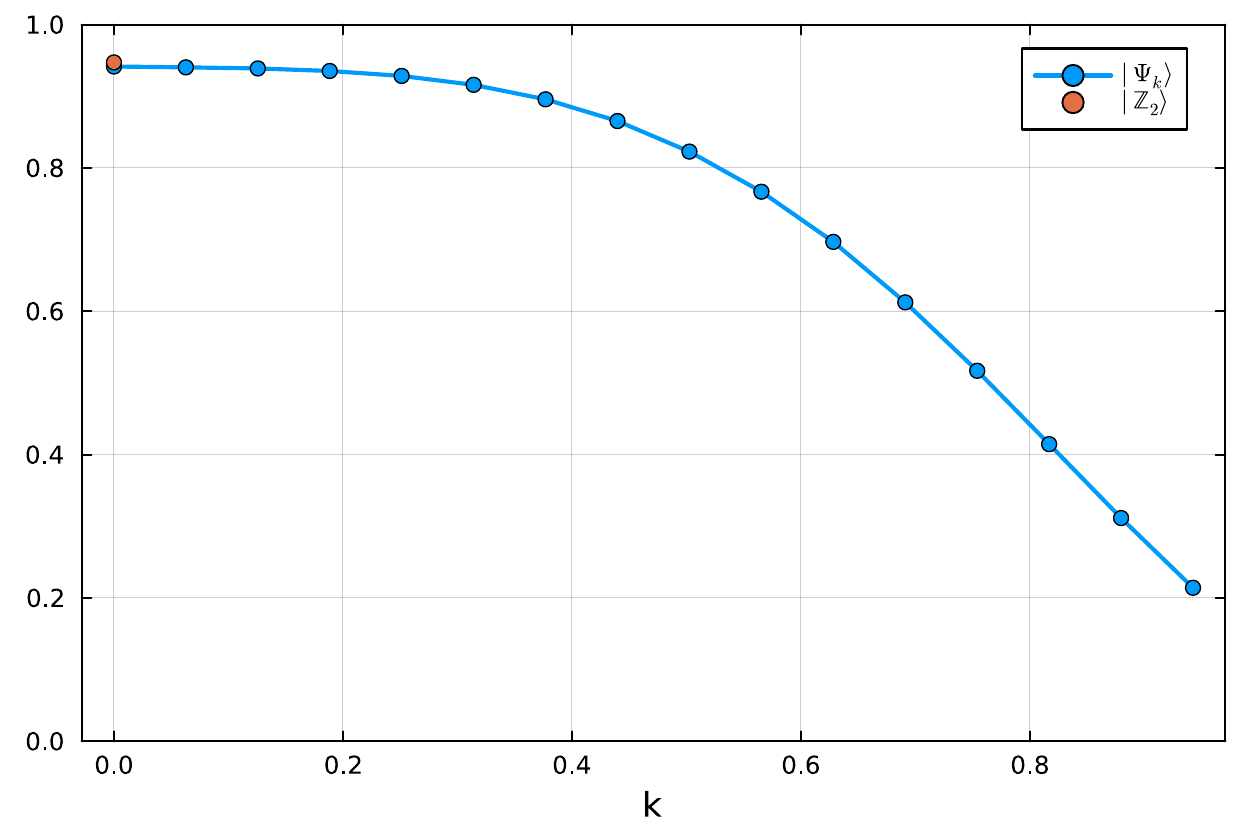}
	\caption{One-period fidelity $F(T)$ for initial state $|\mathbb{Z}_2\rangle$ as well as $|\Psi_{k}\rangle$'s. Note that $F_k(T)$ behaves like a continuous function of $k$.}
	\label{fig:pxp-band}
\end{figure}

Since the PXP model does not fall within the quasisymmetry framework and lacks a ``degenerate limit,'' we instead consider the evolution operator:
\begin{equation}
	U(T) = \exp(-i \hat H_\text{PXP} T),\quad T \approx 4.86.
\end{equation}
We interpret this operator as the Floquet operator, characterized by the following eigensystem:
\begin{equation}
	U(T)|n\rangle = e^{-i\phi_n}|n\rangle.
\end{equation}
For the initial state $|\psi(0)\rangle$, exhibiting nearly periodic time evolution, its Floquet spectral decomposition can be expressed as:
\begin{equation}
	|\psi(0)\rangle \approx \sum_\alpha c_\alpha |\alpha\rangle,
\end{equation}
where the dominant $c_\alpha$ coefficients correspond to nearly identical quasienergies $\phi_\alpha$. 
A natural approach to quantify the lifetime of the qNGM is to consider the one-period fidelity:
\begin{equation}
	F_k(T) = |\langle\Psi_k|\Psi_k(T)\rangle|^2.
\end{equation}
As shown in Fig.~\ref{fig:pxp-band}, $F_k(T)$ 
behaves as a continuous function of $k$; for small $k$, the fidelity decay closely resembles that of states in the subspace $\mathcal{H}_\text{FSA}$.

\subsection{Spectrum}

\begin{figure}
	\centering
	\includegraphics[width=\linewidth]{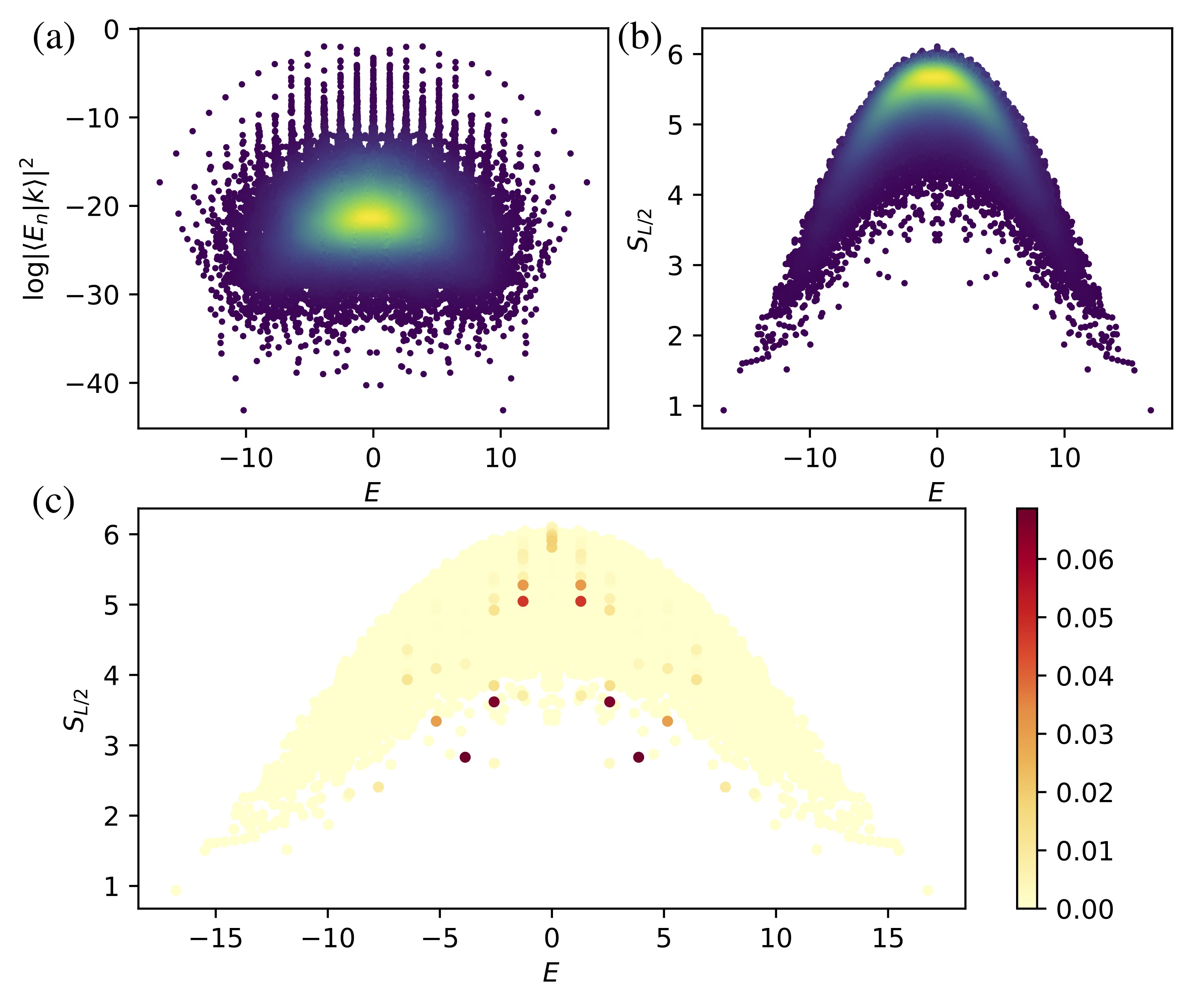}
	\caption{Exact diagonalization of a finite system with $L=28$ in the $(k=2\pi/L) + (k=\pi+2\pi/L)$ sector. (a) depicts the overlap between the initial state and each eigenstate. (b) illustrates the entanglement entropy of each eigenstate. (c) presents a scatter plot of eigenenergy versus energy, with a color map indicating the overlaps with the qNGMs defined in Eq.~(\ref{eq:pxp-qngm}).}
	\label{fig:pxp-ent}
\end{figure}

We examine the qNGM with the smallest nonzero momentum $k = \frac{2\pi}{L}$. 
Note that since the initial state is invariant only under a two-site shift, it spans both $k = \frac{2\pi}{L}$ and $k = \pi + \frac{2\pi}{L}$ sectors:
\begin{equation}
	|\Psi_k\rangle = \frac{1}{\sqrt 2} \left( |\Psi_k^e\rangle + |\Psi_k^o\rangle \right),
\end{equation}
where $|k_e\rangle$ and $|k_o\rangle$ are projections of $|k\rangle$ onto each sector. 
Specifically:
\begin{equation}
\begin{aligned}
	|\Psi_k^e\rangle &\propto \sum_{j\in\text{odd}} e^{ikj}\hat \sigma_j^- |\mathbb Z_2\rangle
	+ \sum_{j\in\text{even}} e^{ikj}\hat \sigma_j^- |\mathbb Z_2'\rangle, \\
	|\Psi_k^o\rangle &\propto \sum_{j\in\text{odd}} e^{ikj}\hat \sigma_j^- |\mathbb Z_2\rangle
	- \sum_{j\in\text{even}} e^{ikj}\hat \sigma_j^- |\mathbb Z_2'\rangle,
\end{aligned}
\end{equation}
where $|\mathbb Z_2'\rangle$ serves as the ``top state'' in the approximate SU(2) representation.

Numerically, we perform exact diagonalization on a system of size $L=28$ to investigate the eigenstate overlap $|\langle E_n|k\rangle|^2$ and the bipartite entanglement entropy of each eigenstate. 
Specifically, we perform the diagonalization in the subspace:
\begin{equation}
	\mathcal H = \mathcal{H}_{k=\frac{2\pi}{L}} \oplus \mathcal{H}_{k=\pi+\frac{2\pi}{L}}.
\end{equation}
As depicted in Fig.~\ref{fig:pxp-ent}(a), the eigenstate components exhibit strong overlaps with approximately equally-spaced eigenstates, suggesting nearly periodic evolution. 
However, as shown in Fig.~\ref{fig:pxp-ent}(b), the subspace $\mathcal H$ lacks a set of low-entangled states, unlike the $(k=0)$ sector. 
Further illustrated in Fig.~\ref{fig:pxp-ent}(c), although some states are relatively low-entangled, the majority exhibit high entanglement.

\end{appendix}

\bibliography{ref}

\end{document}